\documentclass[aos]{imsart}

\RequirePackage[OT1]{fontenc}
\RequirePackage{amsthm,amsmath}
\RequirePackage{natbib}
\RequirePackage[colorlinks,citecolor=blue,urlcolor=blue]{hyperref}

\usepackage{amssymb}
\usepackage{graphicx}
\usepackage{multirow}
\usepackage{rotfloat} % For [H] option in sidewaystable
\usepackage[framemethod=TikZ]{mdframed}
\usepackage[colorlinks]{hyperref} % it is needed for todonotes package. 
\usepackage[colorinlistoftodos, textsize=scriptsize]{todonotes} % For making notes. disable option removes all notes.  

\newtheorem{theorem}{Theorem}[section]
\newtheorem{lemma}[theorem]{Lemma}

\newtheorem{remark}{Remark}[section]

\newcommand{\ech}{\color{black}  ~}    % end change

\newcommand{\bea}{\begin{eqnarray*}}
\newcommand{\eea}{\end{eqnarray*}}
\newcommand{\bean}{\begin{eqnarray}}
\newcommand{\eean}{\end{eqnarray}}

\newcommand{\bfX}{{\bf X}}

\newcommand{\V}{{\rm Var}}
\newcommand{\C}{{\rm Cov}}
\newcommand{\sg}{\Sigma}
\newcommand{\what}{\widehat}

\newcommand{\calC}{\mathcal{C}}
\newcommand{\calD}{\mathcal{D}}

\newcommand{\calU}{\mathcal{U}}

\newcommand{\bbP}{\mathbb{P}} 
\newcommand{\bbR}{\mathbb{R}}
\newcommand{\bbE}{\mathbb{E}}

\startlocaldefs
\endlocaldefs

\begin{document}

\begin{frontmatter}
	
	% "Title of the paper"
	\title{Minimax Posterior Convergence Rates and Model Selection Consistency in High-dimensional DAG Models based on Sparse Cholesky Factors}
	\runtitle{Convergence Rate and Selection Consistency for DAG Model}
	%\thankstext{T1}{Footnote to the title with the ``thankstext'' command.}
	
	% indicate corresponding author with \corref{}
	\author{\fnms{Kyoungjae} \snm{Lee}\corref{}\thanksref{m1}\ead[label=e1]{klee25@nd.edu}}
	%\author{\fnms{Kyoungjae} \snm{Lee}\corref{}\thanksref{m1}\ead[label=e1]{klee25@nd.edu}\thanksref{t1}}
	\and
	\author{\fnms{Jaeyong} \snm{Lee}\thanksref{m2}\ead[label=e2]{leejyc@gmail.com}}
	\and
	\author{\fnms{Lizhen} \snm{Lin}\thanksref{m1}\ead[label=e3]{lizhen.lin@nd.edu}}
	\address{Department of Applied and Computational\\ 
		Mathematics and Statistics\\
		The University of Notre Dame\\
		Notre Dame, Indiana 46556\\
		USA \\
		\printead{e1}}
	\address{Department of Statistics \\
		Seoul National University\\
		1 Gwanak-ro, Gwanak-gu\\
		Seoul 08826\\
		South Korea\\
		\printead{e2}}
	\address{Department of Applied and Computational\\ 
		Mathematics and Statistics\\
		The University of Notre Dame\\
		Notre Dame, Indiana 46556\\
		USA \\
		\printead{e3}}
	\affiliation{The University of Notre Dame\thanksmark{m1} and Seoul National University\thanksmark{m2}}
	
	\runauthor{Kyoungjae Lee, Jaeyong Lee and Lizhen Lin}
	
	\begin{abstract}
		In this paper, we study the high-dimensional sparse directed acyclic graph (DAG) models under the empirical sparse Cholesky prior.
		Among our results, strong model selection consistency or graph selection consistency is obtained under more general conditions than those in the existing literature. 
		Compared to \cite{cao2016posterior}, the required conditions are weakened in terms of the dimensionality, sparsity and lower bound of the nonzero elements in the Cholesky factor.
		Furthermore, our result does not require the irrepresentable condition, which is necessary for Lasso type methods.
		We also derive the posterior convergence rates for precision matrices and Cholesky factors with respect to various matrix norms.
		The obtained posterior convergence rates are the fastest among those of the existing Bayesian approaches.
		In particular, we prove that our posterior convergence rates for Cholesky factors are the minimax  or at least nearly minimax depending on the relative size of true sparseness for the entire dimension.
		The simulation study confirms that  the proposed method outperforms the competing methods. 
	\end{abstract}
	
	\begin{keyword}[class=MSC]
		\kwd[Primary ]{62C20}
		\kwd[; secondary ]{62F15}
		\kwd{62C12.}
	\end{keyword}
	
	\begin{keyword}
		\kwd{DAG model}
		\kwd{Precision matrix}
		\kwd{Cholesky factor}
		\kwd{Posterior convergence rate}
		\kwd{Strong model selection consistency}
	\end{keyword}
\end{frontmatter}

%%%%%%%%%%%%%%%%%%%%%%%%%%%%%%%%%%%%%%%%%%%%%
\section{Introduction}

Detecting the dependence structure of multivariate data is one of important and challenging tasks,
%major concerns in recent years, 
especially when the number of variables is much larger than the sample size. 
Due to  advancements in technology, such data are routinely collected in a wide range of areas including genomics, climatology, proteomics and neuroimaging.
%The estimation of the precision (or covariance) matrix is a crucial initial step for further analysis such as principle component analysis, linear (or quadratic) discriminant analysis and clustering analysis.
The estimation of the covariance (or precision) matrix is crucial to reveal the dependence structure.
Under the high-dimensional setting, however, the traditional sample covariance matrix is no longer a consistent estimator of the true covariance matrix \citep{johnstone2009consistency}.
For the consistent estimation of the high-dimensional covariance or precision matrices, various restrictive matrix classes have been proposed to reduce the number of effective parameters. 
They include the bandable matrices \citep{bickel2008regularized,cai2010optimal,cai2012adaptive,banerjee2014posterior}, sparse matrices \citep{cai2012minimax, cai2012optimal,banerjee2015bayesian} and low-dimensional structural matrices such as the sparse spiked covariance \citep{cai2015optimal,gao2015rate} and sparse factor models \citep{fan2008high,pati2014posterior}.
When the class of sparse matrices is of interest, the sparsity pattern can be encoded in many different ways.
Sparsity  can be imposed on the covariance matrix, precision matrix or Cholesky factor, which lead  to  different graph models.
In this paper, we focus on  imposing  sparsity  on the Cholesky factor of the precision matrix.

Consider a sample of data $X_1,\ldots,X_n \overset{i.i.d.}{\sim} N_p(0, \sg_n)$, where $N_p(\mu, \sg)$ is the $p$-dimensional normal distribution with the mean vector $\mu \in \bbR^p$ and covariance matrix $\sg \in \bbR^{p\times p}$.
For every positive definite matrix $\Omega_{n} =  \sg_n^{-1}$, the modified Cholesky decomposition (MCD) guarantees the  existence of unique Cholesky factor $A_n$ and diagonal matrix $D_n$ such that $\Omega_{n}= (I_p -A_n)^T D_n^{-1} (I_p - A_n)$.
The sparsity of a Gaussian directed acyclic graph (DAG)  can be uniquely encoded  by the Cholesky factor $A_n$ through the structure of the graph. 
In this paper, we assume that the parent ordering of the variables is known, which is a common assumption  used in the  literature such as in \cite{ben2015high}, \cite{khare2016convex}, \cite{yu2017learning} and \cite{cao2016posterior}.  
The details on this concept will be provided in subsection \ref{subsec:gaussian DAG}. 
For the estimation of Cholesky factor $A_n$, the banded assumption and the sparsity assumption are two popular assumptions. 
Under the banded assumption, the elements of the matrix far from the diagonal are assumed to be all zero, while under the sparsity assumption,  there is no constraint on the zero-pattern other than  assuming most of the entries are zero.
In recent years, various penalized likelihood estimators have been proposed with the sparsity assumption on $A_n$ \citep{huang2006covariance,rothman2010new,shojaie2010penalized,van2013ell,khare2016convex}
and banded assumption on $A_n$ \citep{yu2017learning}.

%On the Bayesian side, a variety of prior distributions have been suggested for the DAG models.
%Smith and Kohn (2002)\cite{smith2002parsimonious} proposed a parsimonious hierarchical Bayesian model by introducing exactly zeros in the Cholesky factor. 
%For the decomposable graph models, the hyper inverse Wishart distributions on the covariance matrices were developed by Dawid and Lauritzen (1993)\cite{dawid1993hyper}, which induce the $G$-Wishart distributions on the precision matrices (Letac and Massam, 2007\cite{letac2007wishart}).
%Note that the decomposable graphs correspond to the special case of the DAGs, which is called the perfect DAGs.
%Thus, the priors for the decomposable graph models are only applicable to the special case of the DAG models.
%For the arbitrary DAG models, Ben-David et al. (2015)\cite{ben2015high} proposed the DAG-Wishart distributions.
%However, especially in high-dimensional settings, there are only a few works available that have studied the asymptotic properties of the posteriors. 

On the Bayesian side, relatively few works  have dealt with  asymptotic properties of the posteriors of high-dimensional Gaussian DAG models.
Posterior convergence rates for the precision matrices with $G$-Wishart priors \citep{roverato2000cholesky} were derived by \cite{banerjee2014posterior} and \cite{xiang2015high}, where $G$ is a decomposable graph. 
% Note that decomposable graphs are only a special case of the DAGs, the perfect DAGs. 
Note that a decomposable graph can be converted to a perfect DAG, a special case of the DAGs, by ignoring directions.
\cite{lee2017estimating} obtained the posterior convergence rates and minimax lower bounds for the precision matrices, but  only  bandable Cholesky factors were considered.
Posterior convergence rates for the precision matrices  as well as  strong model selection consistency  were recently derived by \cite{cao2016posterior} for  sparse DAG models.
%They also proved the strong model selection consistency using the DAG-Wishart priors (Ben-David et al., 2015)\cite{ben2015high}.
However, their results are not adaptive to the unknown sparsity $s_0$, and the conditions required  for obtaining such results  are somewhat restrictive.

%In this paper, we focus on the high-dimensional Gaussian DAG models under the sparsity assumption on the DAGs.
In this paper, we consider high-dimensional  sparse Gaussian DAG models where sparsity is imposed via the sparse Cholesky factor. 
% under the sparsity assumption on the DAGs.
We adopt an empirical Bayes approach with a fractional likelihood.
%The sparse Cholesky prior is an empirical prior and can be interpreted as a generalization of the prior suggested in Martin et al. (2017)\cite{martin2017empirical}, which was originally designed for the estimation of the regression coefficients.
The empirical Bayes approach is justified by showing desirable asymptotic properties of the induced posterior such as strong model selection consistency and optimal posterior convergence rates.    
In addition, our theoretical results are adaptive to the unknown sparsity $s_0$.

There are four main contributions of this work.
First, we show  strong model selection consistency under much more general conditions than those in the literature.  
Specifically, the required conditions on the dimensionality, sparsity, structure of the Cholesky factor $A_n$ and the lower bound of the nonzero elements in $A_n$ (the {\it beta-min} condition, which will be described later) are significantly weakened.
Second, we derive the minimax or nearly minimax posterior convergence rates for the Cholesky factors under two scenarios:  with or without the beta-min condition  for the true Cholesky factor.
We show that at least one of the posterior convergence rates is  minimax  depending on the relative size of true sparseness for the entire dimension.
To the best of our knowledge, this is the first result on minimax posterior convergence rates in high-dimensional DAG models.
Third, we obtain the posterior convergence rates for precision matrices with respect to the spectral norm and matrix $\ell_\infty$ norm, which is the fastest among those of existing Bayesian approaches. 
%Cao et al. (2016)\cite{cao2016posterior} also obtained the posterior convergence rate with respect to the spectral norm.
%If we assume the bounded eigenvalue condition for the precision matrix, their posterior convergence rate is $s_0^2\sqrt{\log p/n}$, while our posterior convergence rate is $s_0^{3/4}\sqrt{(s_0 + \log p)/n}$, where $s_0$ is the maximum number of nonzero elements in each row of the true Cholesky factor.
%Furthermore, we assume much more general conditions compared to them, except the bounded eigenvalue condition.
Compared to \cite{cao2016posterior}, we achieve faster posterior convergence rate under more general conditions, except the bounded eigenvalue condition.
Furthermore, their results depend on the unknown sparsity $s_0$, whereas ours do not. 
Fourth, our method significantly improves the model selection performance in practice.
In particular, our method outperforms the other state-of-the-art methods in a simulation study. 
The theoretical choice of hyperparameters provided good  guidelines for practical performance.
 %gave reasonably good performance in practice.
Note that the choice of the hyperparameter, the individual edge probability $q_n$, in the hierarchical DAG-Wishart prior \citep{cao2016posterior} can be problematic in practice, as the posterior with the theoretical choice of $q_n$ tends to be stuck at very small size models.

The rest of paper is organized as follows.
In section \ref{sec:prelim}, we introduce notations, Gaussian DAG models, the empirical sparse Cholesky prior, the fractional posterior and the parameter class for the precision matrices.
In section \ref{sec:main},  strong model selection consistency,  posterior convergence rates and minimax lower bounds for the Cholesky factor, and posterior convergence rates for the precision matrices are established.
%A simulation study and real data analysis focusing on the model selection property are represented in section \ref{sec:numerical}, and a discussion is given in section \ref{sec:disc}. 
A simulation study focusing on the model selection property are represented in section \ref{sec:numerical}. 
%and a discussion is given in section \ref{sec:disc}. 
The proofs of the main results are provided in the supplemental article \citep{lee2018supp}.

%%%%%%%%%%%%%%%%%%%%%%%%%%%%%%%%%%%%%%%%%%%%%
\section{Preliminaries}\label{sec:prelim}
\subsection{Norms and Notations}
For any $a ,b \in \bbR$, we denote $a\vee b$ and $a \wedge b$ as the maximum and minimum of $a$ and $b$, respectively.
For any $a \in \bbR$, we denote $\lfloor a \rfloor$ as the largest integer strictly smaller than $a$.
For any sequences $a_n$ and $b_n$, $a_n=o(b_n)$ denotes $a_n/b_n \to 0$ as $n\to\infty$. We denote $a_n = O(b_n)$, or equivalently $a_n \lesssim b_n$, if $a_n \le C b_n$ for some constant $C>0$, where $C$ is an universal constant.
We denote the indicator function for a set $A$ as $I(\cdot \in A)$ or $I_A(\cdot)$.
For a given $p$-dimensional vector $u = (u_1,\ldots,u_p)^T$ and set $S \subseteq \{1,\ldots, p\}$, we define $u_{S} = (u_{j})^T_{j\in S} \in \bbR^{|S|}$, where $|S|$ is the cardinality of $S$. 
For given index sets $S$, $S' \subseteq \{1,\ldots, p\}$ and real matrix $B \in \bbR^{p\times p}$, we denote $B_{(S, S')}$ as the $|S|\times |S'|$ submatrix consisting only of $S$th columns and $S'$th rows of $B$, and let $B_{S} = B_{(S, S)}$.
For a real matrix $B$, we denote $S_B$ as the index set for nonzero elements of $B$ and call $S_B$ the {\it support} of $B$.
We define $\calC_p$ as the class of all $p\times p$ dimensional positive definite matrices. 
For any $p\times p$ symmetric matrix $B$, $\lambda_{\min}(B)$ and $\lambda_{\max}(B)$ are the minimum and maximum eigenvalues of $B$, respectively.

For any $p$-dimensional vector $u=(u_1,\ldots,u_p)^T$, we define vector norms $\|u\|_1 = \sum_{j=1}^p |u_j|$, $\|u\|_2 = (\sum_{j=1}^p u_j^2)^{1/2}$ and $\|u\|_{\max} = \max_{1\le j\le p} |u_j|$.
For any $p\times p$ matrix $B=(b_{ij})$, we define the spectral norm, matrix $\ell_1$ norm, matrix $\ell_\infty$ norm and Frobenius norm by
\bea
\|B\| &=& \sup_{x\in\bbR^p \atop \|x\|_2=1} \|B x\|_2 \,\,=\,\, \left( \lambda_{\rm max}(B^T B) \right)^{1/2} , \\
\|B\|_1 &=& \sup_{x\in\bbR^p \atop \|x\|_1=1} \|B x \|_1 \,\,=\,\, \max_{1\le j\le p} \sum_{i=1}^p |b_{ij}| ,\\
\|B\|_{\infty} &=& \sup_{x\in\bbR^p \atop \|x\|_{\max}=1} \|B x \|_{\max} \,\,=\,\, \max_{1\le i\le p} \sum_{j=1}^p |b_{ij}|, \quad\text{ and} \\
\|B \|_F &=& \Big(\sum_{i=1}^p \sum_{j=1}^p  b_{ij}^2  \Big)^{1/2},
\eea
respectively.

%For a given random sample $Y_1,\ldots, Y_n \overset{iid}{\sim} P$, we denote the empirical estimate for $\V(Y)$ by
%$\what{\V}(Y) \equiv \what{\C}(Y,Y) = \frac{1}{n} \sum_{i=1}^n Y_{i} Y_i^T$, where $Y\sim P$. 
%In the rest of the paper, $\what{T}$ indicates the empirical estimate for the parameter $T$. 
For a given positive integer $m$, we denote $\chi^2_m$ as the chi-square distribution with degrees of freedom $m$.
For any random variables $Y_1,Y_2$ and $Y_3$, we denote $Y_1 \overset{d}{\equiv} Y_2 \oplus Y_3$ if the distribution of $Y_1$ is equal to that of $Y_2+Y_3$, and $Y_2$ and $Y_3$ are independent.
For given positive numbers $a$ and $b$, $Gamma(a,b)$ and $IG(a,b)$ are the gamma distribution and inverse-gamma distribution with shape parameter $a$ and rate parameter $b$, respectively. 
$Beta(a,b)$ is the beta distribution whose density function at $x\in (0,1)$ is proportional to $x^{a-1}(1-x)^{b-1}$.
We denote $N_p(X \mid \mu,\sg)$ as the density function of $N_p(\mu,\sg)$ at $X \in \bbR^p$.
We denote the inverse-Wishart distribution by $IW_p(\nu, \Phi)$, where  the degree of freedom and scale matrix are $\nu>p-1$ and $\Phi \in \calC_p$, respectively.

\subsection{Gaussian DAG Models}\label{subsec:gaussian DAG}

We consider the model
\bean\label{model}
X_1,\ldots,X_n \mid \Omega_n &\overset{i.i.d.}{\sim}& N_p(0, \Omega_n^{-1}),
\eean
where $\Omega_n = \sg_n^{-1}$ is a $p\times p$ precision matrix and $X_i = (X_{i1},\ldots, X_{ip})^T\in \bbR^p$ for all $i=1,\ldots,n$. 
The MCD guarantees that there exists unique lower triangular matrix $A_n=(a_{jl})$ and diagonal matrix $D_n=diag(d_j)$ such that
\bean\label{chol}
\Omega_n &=& (I_p - A_n)^T D_n^{-1} (I_p - A_n) ,
\eean
where $a_{jj} =0$ and $d_j >0$ for all $j=1,\ldots, p$. 
Let $S_{A_n}$ be the support of the Cholesky factor $A_n$, and $S_j$ be the support of the $j$th row of $A_n$.
Let $\bbP_{\Omega_{n}}$ and $\bbE_{\Omega_{n}}$ be the probability measure and expectation corresponding to the model \eqref{model}, respectively. 

The model \eqref{model} can be interpreted as a Gaussian DAG model depending on the sparsity pattern of $A_n$.
For a set of vertices $V = \{1,\ldots,p\}$ and a set of directed edges $E$, a graph $\calD = (V,E)$ is said to be a DAG if there is no directed cycles.
It is also called the Bayesian network or belief network.
In this paper, we assume that the variables have a known natural ordering in which no edges exist from larger vertices to smaller vertices. 
It has been commonly assumed in literature including \cite{shojaie2010penalized}, \cite{ben2015high}, \cite{khare2016convex}, \cite{yu2017learning} and \cite{cao2016posterior}. 
There are relatively few works on DAG models when the  ordering of variables is unknown \citep{kalisch2007estimating,rutimann2009high,van2013ell}. 
%For example, \cite{van2013ell} suggested an $\ell_0$-penalized maximum likelihood estimator and showed its convergence rates  in Frobenius norm under the regularity conditions.    
As discussed in \cite{van2013ell}, when the ordering is unknown, a very different technique is needed relative to the known ordering case. 

%\bchj In this paper, we deal with only the known ordering case. \ech 

For $i \in V$, define the set of all $i$'s parents as the subset of $V$ smaller than $i$ and sharing an edge with $i$ and denote it as $pa_i(\calD)$.  
Any multivariate Gaussian distribution that obeys the directed Markov property with respect to a DAG $\calD$ is said to be a {\it Gaussian DAG model over} $\calD$.
To be specific, if $X = (X_{1},\ldots, X_{p})^T \sim N_p(0, \Omega^{-1})$
and $N_p(0, \Omega^{-1})$ belongs to a Gaussian DAG model over $\calD$, then 
\bea
X_{j} &\perp& \{X_{j'} \}_{j'< j ,\, j' \notin pa_j(\calD)}  \,\,\Big| \,\, (X)_{pa_j(\calD) } ,
\eea
for each $j =1,\ldots,p$. 
Furthermore, if we adopt the MCD as in \eqref{chol}, with the known ordering of variables, $N_p(0,\Omega^{-1})$ belongs to a Gaussian DAG model over $\calD$ if and only if $a_{jl} =0$ whenever $l \notin pa_j(\calD)$ \citep{cao2016posterior}.
In other words, \emph{the support of $A$ uniquely determines a DAG $\calD$} under the known ordering assumption.  \ech
The model $X=(X_1,\ldots,X_p)^T \sim N_p(0,\Omega^{-1})$ given $S_{A}$ is equivalent to a Gaussian DAG model, and it can be represented as a linear autoregressive model,
\bean\label{model2}
\begin{split}
	{X}_{1} \mid d_1 \,\,&\sim\,\, N(0,\, d_1 ), \\
	{X}_{j} \mid a_{S_j}, d_j, S_j \,\,&\overset{ind}{\sim}\,\,  N \Big( \sum_{l\in S_j} {X}_{l} a_{jl} ,\, d_j \Big) ,~ ~j=2,\ldots,p ,
\end{split}
\eean
where $a_{S_j} = a_{j,S_j} =(a_{j j'})_{j'\in S_j}^T$. 
For more details on the expression \eqref{model2}, refer to \cite{bickel2008regularized} and \cite{ben2015high}.
The autoregressive model interpretation enables us to adopt the priors introduced in the linear regression literature. 
Since $a_{S_j}$ corresponds to nonzero elements among $a_j = (a_{j1},\ldots, a_{j,j-1})^T$, one can use a prior designed for sparse regression coefficient vectors for $a_{j}$, which is our strategy introduced in section \ref{subsec:SCprior}.

In this paper, we consider the high-dimensional setting where $p = p_n$ is a function of $n$ increasing to infinity as $n\to \infty$ and $p \ge n$.
%The theoretical results in this paper can be extended to the case $p \ge n^{\beta'}$ for some constant $0<\beta'<1$, but we only consider the case $p \ge n$ for simplicity.
We assume that the data were generated from a true precision matrix $\Omega_{0n}$, where $\sg_{0n} = \Omega_{0n}^{-1}$ is the true covariance matrix.
Denote the MCD \eqref{chol} of the true precision matrix by $\Omega_{0n} = (I_p - A_{0n})^T D_{0n}^{-1}(I_p- A_{0n})$, where $A_{0n} =(a_{0,jl})$, $a_{0j}= (a_{0,j1},\ldots, a_{0,j,j-1})^T$ and $D_{0n} = diag(d_{0j})$.
For notational convenience, let $\bbP_0= \bbP_{\Omega_{0n}}$ and $\bbE_0 = \bbE_{\Omega_{0n}}$.

We now define some notations related to the data set. Let $\bfX_n =(X_1,\ldots,X_n)^T$ $\in \bbR^{n\times p}$ be the data of size $n$, and $\tilde{X}_j = (X_{1j},\ldots, X_{nj})^T \in \bbR^n$ be the $j$th column of $\bfX_n$.
For a given index set $S \subseteq \{1,\ldots,p\}$, let $\bfX_{S}=(\tilde{X}_{j})_{j \in S} \in \bbR^{n\times |S|}$ be the data matrix consisting only of $S$th columns of $\bfX_n$.
Let $Z_{ij} = (X_{i1},\ldots, X_{i,j-1})^T \in \bbR^{j-1}$ and $\tilde{Z}_j = (Z_{1j},\ldots, Z_{nj})^T \in \bbR^{n\times (j-1)}$ for all $j=2,\ldots,p$. 
%Denote $Z_{i,S_j} = (X_{ij'})_{j'\in S_j}^T \in \bbR^{|S_j|}$ for all $i=1,\ldots,n$.

For a given positive integer $1\le s \le p$, we define $\Psi_{\max}(s)^2 = \sup_{S:0< |S| \le s} \lambda_{\max}( \bfX_{S}^T \bfX_{S})$ and $\Psi_{\min}(s)^2 =\inf_{S:0< |S| \le s} \lambda_{\min}( \bfX_{S}^T \bfX_{S})$, where the supremum and infimum are taken over all index sets $S \subseteq \{1,\ldots,p\}$.
\emph{We say that the restricted eigenvalue condition is met for some integer $s$ if $n^{-1}\Psi_{\min}(s)^2$ is bounded away from zero uniformly for all large $n$. }
This condition has been used in the high-dimensional regression literature to control the behavior of the design matrix. 
The autoregressive model representation \eqref{model2} connects the eigenvalues of the precision matrix $\Omega_{0n}$ with those of the design matrix in \eqref{model2} because the quantity $\bfX_{S_j}$ corresponds to the design matrix based on the representation.
Thus, the bounded eigenvalue assumption \hyperref[A1]{\rm (A1)} in section \ref{subsec:class} essentially corresponds to the restricted eigenvalue condition.

%We define the empirical estimates $\what{\V}(X_j) = n^{-1} \sum_{i=1}^n X_{ij}^2$, $\what{\V}(Z_j) = n^{-1} \sum_{i=1}^n Z_{ij} Z_{ij}^T$ and $\what{\C}(Z_j, X_j) = n^{-1}\sum_{i=1}^n Z_{ij} X_{ij}$.
%For a given index set $S\subseteq \{1,\ldots,p\}$, define $\what{\V}(Z_{S})= n^{-1}\sum_{i=1}^n Z_{i,S}Z_{i,S}^T$ and $\what{\C}(Z_{S}, X_j) = n^{-1}\sum_{i=1}^n Z_{i,S}X_{ij}$.

\subsection{Empirical Sparse Cholesky  Prior}\label{subsec:SCprior}

In this paper, we suggest the following prior distribution for our model:
\begin{eqnarray}\label{SCprior}
\begin{split}
a_{S_j} \mid d_j, S_j \,\,&\overset{ind}{\sim}\,\, N_{|S_j|} \left( \what{a}_{S_j},\,\,   \frac{d_j}{\gamma } \big(\bfX_{S_j}^T \bfX_{S_j}\big)^{-1}  \right) ,~~ j=2,\ldots,p ,\\
\pi(d_{j})  \,\,&\overset{i.i.d.}{\propto}\,\, d_j^{-\nu_0/2 -1 }, ~~ j=1,\ldots,p , \\
\pi_j(S_j=S_j') \,\,&\propto\,\, \binom{j-1}{|S_j'|}^{-1} f_{nj}(|S_j'|) , \,\, j=2,\ldots,p ,\,\, S_j' \subseteq \{1,\ldots, j-1\},\\
f_{nj}(|S_j'|) \,\,&\propto\,\, c_1^{-|S_j'|} p^{-c_2 |S_j'|} I(0\le |S_j'| \le R_j \wedge (j-1)), \,\, j=2,\ldots,p,
\end{split}
\end{eqnarray}
for some positive constants $\nu_0,c_1,c_2,R_2,\ldots,R_p$ and $\gamma$, where $f_{nj}$ is a probability mass function on $\{0,1,\ldots, R_j\wedge(j-1) \}$ and $\what{a}_{S_j}  = \big(\bfX_{S_j}^T \bfX_{S_j}\big)^{-1} \bfX_{S_j}^T \tilde{X}_j$. 
The proposed prior is empirical in the sense that it depends on the data, so we call the prior \eqref{SCprior} the empirical sparse Cholesky (ESC) prior.  
To obtain the desired asymptotic properties, appropriate conditions for hyperparameters in \eqref{SCprior} will be introduced in section \ref{sec:main}. 
Note that the prior for $d_j$ can be generalized to the proper prior $IG(\nu_0/2, \nu_0')$ for some constant $\nu_0'>0$ and the results in section \ref{sec:main} also hold for this prior choice. 
However, for  computational convenience, we describe and prove the main results with the improper prior $\pi(d_j) \propto d_j^{-\nu_0/2-1}$.

For the conditional prior of $a_j$ given $d_j$, we first introduce zero components through the prior $\pi_j$ and impose the Zellner's g-prior \citep{zellner1986assessing} on the nonzero components, $a_{S_j}$. 
The use of Zellner's g-prior simplifies the calculation of the marginal posterior for $S_j$. 
\cite{martin2017empirical} suggested a similar prior in the high-dimensional linear regression model.
Also note that the ESC prior has a connection to the DAG-Wishart prior \citep{ben2015high, cao2016posterior}.
Theorem 7.3 in \cite{ben2015high} shows that the DAG-Wishart prior on $(A_n, D_n)$ given a DAG implies the inverse-gamma distribution on $d_j$ and multivariate normal distribution on the nonzero elements of $a_j$ given $d_j$,  where  $(a_j, d_j)$ are mutually independent for all $j=1,\ldots,p$.
Thus, the ESC prior \eqref{SCprior} is quite close to the DAG-Wishart prior when the support of $A_n$ is given.

\cite{cao2016posterior} used the DAG-Wishart prior to recover the sparse DAG and estimate the precision matrix in high-dimensional settings. 
Thus, their prior on $(A_n,D_n)$ is quite close to ours, and can be viewed as a set of priors for autoregressive model \eqref{model2} as discussed in the previous paragraph.
For the support of DAGs, they imposed the element-wise sparsity using independent Bernoulli distributions with the hyperparameter $q_n$, which has a nice interpretation as the individual edge probability. 
Based on the hierarchical DAG-Wishart prior, they obtained the strong model selection consistency for the DAG and the posterior convergence rate for the precision matrix with respect to the spectral norm.
However, they did not directly adopt the autoregressive model interpretation as in \eqref{model2}, which is different from our approach.   
By using the ESC prior, we can adopt state-of-the-art techniques on selection consistency for the regression coefficient \citep{martin2017empirical} and achieve the strong model selection consistency under much weaker conditions than those in \cite{cao2016posterior}.
Furthermore, compared to the existing literature, we obtain faster posterior convergence rates for precision matrices and Cholesky factors under weaker conditions using the techniques introduced by \cite{lee2017optimal, lee2017estimating} and \cite{martin2017empirical}.
Indeed, the posterior convergence rates for Cholesky factors are nearly or exactly optimal depending on the relative size of true sparseness for the entire dimension.

\subsection{$\alpha$-posterior Distribution}

We suggest adopting the fractional likelihood with power $\alpha \in (0,1)$,   
\bean\label{flik}
L_n(A_n, D_n)^\alpha &=& \prod_{i=1}^n \Big\{ N_p \big( X_i \mid 0,\, (I_p-A_n)^{-1} D_n ((I_p-A_n)^T)^{-1}  \big) \Big\}^\alpha .
\eean 
The use of fractional likelihood has received increased attention in recent years \citep{martin2014asymptotically, syring2015scaling, miller2018robust}.
In this paper, we use the fractional likelihood mainly because of its appealing theoretical properties under relatively weaker conditions compared to the actual posterior \citep{bhattacharya2016bayesian}.
In the proof of the main results of this paper, the use of the fractional likelihood enables us to effectively deal with the ratio of estimated residual variances $\what{d}_{S_j}$ (the  proof of Theorem \ref{thm:sel}) and the ratio of likelihood $L_{nj}(a_j, d_j)$ (the  proof of Lemma 10.2), where $\what{d}_{S_j}$ and $L_{nj}(a_j, d_j)$ will be defined later.

Here we give a more detailed justification of using the fractional likelihood.
The proposed conditional prior for $a_{S_j}$ in \eqref{SCprior} tracks the data closely because it is centered at the least square estimate. 
It may cause the unexpected inconsistency \citep{walker2001bayesian}.
The fractional likelihood approach can prohibit it by preventing the posterior from following the data too closely.
To be more specific, the use of fractional likelihood can be interpreted as an empirical Bayes procedure by considering
\bea
L_n(A_n, D_n)^\alpha \, \pi(A_n, D_n) 
&=& L_n(A_n, D_n) \, \frac{\pi(A_n, D_n)}{L_n(A_n, D_n)^{1-\alpha}}.
\eea
Hence, the resulting posterior consists of an ordinary likelihood function and a data-dependent prior $\pi(A_n,D_n)/L_n(A_n,D_n)^{1-\alpha}$. 
Note that the power $\alpha$ {\it only appears in the prior}.
From this point of view, the prior is rescaled by a fractional likelihood which has an effect of penalizing parameter values that track the data too closely, while the penalty effect is controlled by the {\it hyperparameter} $\alpha$.

The choice of $\alpha$ can be important from a practitioner's point of view even though {\it theoretical results in this paper hold for any choice of} $0<\alpha<1$.
In practice, we suggest using $\alpha$ close to 1 to mimic the usual likelihood in finite sample scenario if there is no suspect of model failure, i.e. misspecification.
As long as one chooses $\alpha$ sufficiently close to 1, e.g. $\alpha=0.999$ or $\alpha =0.9999$, our experience confirms that the $\alpha$-posterior can be hardly distinguishable from the ``usual" posterior even in a finite sample scenario. 

\begin{remark}
	\cite{grunwald2017inconsistency} suggested using $I$-log-SafeBayes (or $R$-log-SafeBayes) to determine $\alpha$, which gives the minimizer $\hat{\alpha}$ of the posterior-expected posterior-randomized loss of prediction (or its variant).
	The induced posterior is robust to model misspecification in some cases \citep{grunwald2017inconsistency}.
\end{remark}

The prior \eqref{SCprior} and fractional likelihood \eqref{flik} lead to the following joint posterior distribution,
\begin{eqnarray}\label{post}
\begin{split}
a_{S_j} \mid d_{j}, S_j, \bfX_n \,\,&\overset{ind}{\sim}\,\, N_{|S_j|} \left(  \what{a}_{S_j}   , ~\frac{d_{j}}{(\alpha+\gamma)}\big(\bfX_{S_j}^T \bfX_{S_j}\big)^{-1}  \right), \quad j=2,\ldots,p, \\
d_{j} \mid S_j, \bfX_n \,\,&\overset{ind}{\sim} \,\, IG \left(\frac{\alpha n + \nu_0 }{2}, \frac{\alpha n}{2}  \what{d}_{S_j} \right) , \quad j=1,\ldots, p ,\\
\pi_\alpha(S_j \mid \bfX_n) \,\,&\propto\,\, \pi_j(S_j) \left(1+ \frac{\alpha}{\gamma} \right)^{-\frac{|S_j|}{2}} (\what{d}_{S_j})^{- \frac{\alpha n + \nu_0}{2}}, \quad j=2,\ldots,p ,
\end{split}
\end{eqnarray}
where $\what{d}_{S_j} = n^{-1} \tilde{X}_j^T( I_n - \tilde{P}_{S_j} ) \tilde{X}_j$ and $\tilde{P}_{S_j} = \bfX_{S_j} (\bfX_{S_j}^T \bfX_{S_j})^{-1}\bfX_{S_j}^T $. 
We refer to the posterior \eqref{post} as the $\alpha$-posterior and denote it by $\pi_\alpha(\cdot \mid \bfX_n)$ to clarify that we consider the $\alpha$-fractional likelihood.
Throughout the paper, $\alpha \in (0,1)$ is a fixed constant.

\subsection{Parameter Class}\label{subsec:class}

%For the true Cholesky factor $A_{0n}$, we denote $S_{0j}$ as the index set for the nonzero elements in the $j$th row of $A_{0n}$ and define $s_{0j}= |S_{0j}|$, for all $j=2,\ldots,p$. 

For given positive constants $0<\alpha<1$, $0<\epsilon_0< 1/2$, $C_{\rm bm}$ and a sequence of positive integers $s_0$, we introduce  conditions \hyperref[A1]{\rm (A1)}-\hyperref[A4]{\rm (A4)} for the true precision matrix:\\

\noindent
{\bf (A1)}\label{A1} $\epsilon_0 \le \lambda_{\min}(\Omega_{0n})\le \lambda_{\max}(\Omega_{0n})\le \epsilon_0^{-1}$.\\
{\bf  (A2)}\label{A2} $\max_{1\le j\le p} \sum_{l=1}^p I(a_{0,jl}\neq 0) \le s_{0}$.\\
{\bf  (A3)}\label{A3} 
\bea
\min_{(j,l): a_{0,jl}\neq 0} |a_{0,jl}|^2 &\ge& { \frac{16}{\alpha(1-\alpha) \, \epsilon_0^2 (1-2\epsilon_0)^2 } \,\,  C_{\rm bm}  \frac{\log p}{n} }  .
\eea
{\bf  (A4)}\label{A4} $\max_{1\le l\le p} \sum_{j=2}^p I(a_{0,jl}\neq 0) \le s_0$. \\

Condition \hyperref[A1]{\rm (A1)} ensures that the eigenvalues of $\Omega_{0n}$ are bounded by fixed constants, which has been commonly used for the estimation of the high-dimensional precision matrices \citep{ren2015asymptotic,cai2016estimating,banerjee2015bayesian} as well as the high-dimensional DAGs \citep{yu2017learning,khare2016convex}. 
In this paper, condition \hyperref[A1]{\rm (A1)} is mainly used to get upper bounds of $d_{0j}$, $d_{0j}^{-1}$ and $\|A_{0n}\|$.

Condition \hyperref[A2]{\rm (A2)} restricts the number of nonzero elements in each row of $A_{0n}$ to be smaller than $s_0$. Note that $s_0$ may increase to infinity as $n$ gets larger. 
In our setting, it is equivalent to say that the cardinality of $pa_j(\calD_0)$ is less than $s_0$ for any $j=2,\ldots,p$, where $\calD_0$ is the DAG corresponding to $A_{0n}$.  
%It means that the number of edges directed to any vertex is less then $s_0$.

Condition \hyperref[A3]{\rm (A3)} is the well-known {\it beta-min} condition, which determines the lower bound for the nonzero signals. 
The beta-min condition has been used for the exact support recovery of the high-dimensional linear regression coefficients \citep{wainwright2009information,buhlmann2011statistics,castillo2015bayesian,yang2016computational,martin2017empirical} as well as the high-dimensional DAGs \citep{khare2016convex,yu2017learning,cao2016posterior}.
%In this paper, we use the beta-min condition \hyperref[A3]{\rm (A3)} with
%\bean\label{betamin_const}
%M_{\rm bm} &=& \frac{4(1+2\epsilon_0)^2 (2 + \epsilon_0^{-2}) }{\alpha (1-2\epsilon_0)^2 \epsilon_0^2} \left(1 - C_{\alpha} \right)^{-1}  C_{\rm bm}
%\eean
%for some constant $C_{\rm bm}>0$ and $0<\alpha<1$, where $C_{\alpha} := \alpha(1+ [(1-\alpha)/10]^2 )/(1- 4 [(1-\alpha)/10] -5[(1-\alpha)/10]^2 )$.

Condition \hyperref[A4]{\rm (A4)} restricts the number of nonzero elements in each column of $A_{0n}$ to be smaller than $s_0$.  
In other words,  the number of edges directed from any vertex is less than $s_0$.
This assumption is required to deal with the posterior probability of $\|A_n - A_{0n}\|_1$.
Note that if we consider only the banded structure for the Cholesky factor as in \cite{yu2017learning}, conditions \hyperref[A2]{\rm (A2)} and \hyperref[A4]{\rm (A4)} automatically hold for some $s_0$.

Now, we define a class of precision matrices
\bea
\calU_p \,\,=\,\, \calU_p(\epsilon_0, s_0, \alpha, C_{\rm bm})  &=& \bigg\{ \Omega \in \calC_p: \,\, \Omega \text{ satisfies \hyperref[A1]{\rm (A1)}-\hyperref[A3]{\rm (A3)}}\,\,  \bigg\} .
\eea
In section \ref{sec:main}, we show that one can achieve the strong model selection consistency for any $\Omega_{0n}\in \calU_p$. 
Furthermore, we derive the posterior convergence rates for $A_{0n}$ and show that these are optimal or nearly optimal for the class $\calU_p$ (or $\calU_p$ without condition \hyperref[A3]{\rm (A3)}).

\begin{remark}
	\cite{cao2016posterior} weakened the bounded eigenvalue condition \hyperref[A1]{\rm (A1)} by replacing a constant $\epsilon_0$ with a sequence $\epsilon_{0,n}$, which can go to zero at certain rate.
%	Thus, they allowed eigenvalues to approach zero and infinity at a certain rate. 
	%However, we emphasize that the results in this paper are adaptive to the unknown sparsity $s_0$ in condition \hyperref[A2]{\rm (A2)}, while those in \cite{cao2016posterior} are not.
	%This point will be clarified in section \ref{subsec:sel}.
	Our results also still hold under the similar weakened bounded eigenvalue condition with $\epsilon_{0,n}$, but it will sacrifice the other conditions.
	For example, by using a sequence $\epsilon_{0,n}$ in place of a fixed $\epsilon_0$ in the proof of Theorem \ref{thm:sel}, one can see that $s_0 \log p \le C n \epsilon_{0,n}^2$ for some $C>0$ and the beta-min condition \hyperref[A3]{\rm (A3)} with $\epsilon_{0,n}$ in place of $\epsilon_0$ are required.
%	Thus, one can use weaker bounded eigenvalue condition, but it would lead to stronger condition on $(n,p,s_0)$ and beta-min condition.
%	For example, for the proof of Theorem \ref{thm:sel}, one can adopt the weakened bounded eigenvalue condition with $\epsilon_{0,n}$ instead of $\epsilon_0$, and the condition on $(n,p,s_0)$ will be still weaker than that in \cite{cao2016posterior}.
%	This point will be clarified in section \ref{sec:proofsec} with more detailed explanation.
\end{remark}

%%%%%%%%%%%%%%%%%%%%%%%%%%%%%%%%%%%%%%%%%%%%%
\section{Main Results}\label{sec:main}

We introduce Condition \hyperref[condP]{(P)} on the hyperparameters in the ESC prior \eqref{SCprior}, which is necessary for the results in this section.
Note that this condition is for the hyperparameters of the prior distribution, which does not affect the true parameter space. 
\\

{\bf Condition (P)}\label{condP} {\it 
	Assume that $\nu_0=o(n)$, $c_1=O(1)$, $c_2 \ge 2$ and $\gamma=O(1)$. 
	For given positive constants $0<\alpha<1$ and $0<\epsilon_0 <1/2$ used in conditions \hyperref[A1]{(A1)} and \hyperref[A3]{(A3)}, assume that $R_j  =  \lfloor n\,(\log p)^{-1} \{(\log n)^{-1} \vee c_3 \} \rfloor$ for any $j=2,\ldots,p$ and some small constant  $0<c_3 < (\epsilon')^2 \epsilon_0^2 /\{128(1+2\epsilon_0)^2 \} $, where $\epsilon' = \{(1-\alpha)/10\}^2$. 
}  \\

The condition $c_2 \ge 2$ is similar to the condition $\kappa \ge 2$ in \cite{yang2016computational}.
Note that the constants $c_1$ and $c_2$ in the ESC prior control the row-wise sparsity of the Cholesky factor $A_{n}$: large values of them make the posterior prefer small values for $|S_j|$. 
Thus, the above condition means that we need certain amount of penalty on $|S_j|$ to achieve desirable asymptotic properties.    
The condition on $R_j$ means that $R_j$ is of order $n(\log p)^{-1}$ and smaller than $n(\log p)^{-1} (\epsilon')^2 \epsilon_0^2 /\{128(1+2\epsilon_0)^2 \}$, so it can be replaced by the condition $R_j = \lfloor n \,(\log p)^{-1} (\epsilon')^2 \epsilon_0^2 /\{128(1+2\epsilon_0)^2 \} \rfloor$. 
To assure $s_0 \le R_n$, we will assume that $s_0 \le n(\log p)^{-1} c_3/2$ later.
In general, assuming $s_0 = O(n(\log p)^{-1} )$ or even $s_0 = o(n(\log p)^{-1})$ is essential to prove theoretical properties such as selection consistency and convergence rates. 
However, it can be unrealistically small for some finite sample size $n$.
More importantly, the quantity $\epsilon_0$ is unknown in typical applications, so it is desirable to make the prior work for any choice of $\epsilon_0$.
Condition \hyperref[condP]{(P)} argues that there is such a prior.
We suggest choosing a small enough $c_3$ so that $R_j$ can be regarded as $R_j = \lfloor n \,(\log p \cdot \log n)^{-1} \rfloor$ for finite samples.

\begin{remark}
	\cite{yang2016computational} suggested a prior for the linear model  similar to the ESC prior but for the mean vector of the prior $\pi(a_{S_j}\mid d_j, S_j)$, they used zero mean vector while we used $\what{a}_{S_j}$.
	There are two consequences from the use of the data-dependent mean $\what{a}_{S_j}$.
	First, we do not need an upper bound condition for $\| \bfX_{S_{0j}} a_{0, S_{0j}}\|_2$ or $\|a_{0, S_{0j}}\|_2$, while \cite{yang2016computational} assumed $\| \bfX_{S_{0j}} a_{0, S_{0j}}\|_2 \le g \, d_{0j} \log p$, where $g = \gamma^{-1}$ in this paper.
	It is known that this type of condition is required if we use the Zellner's $g$-prior with zero mean \citep{shang2011consistency}.
	Second, to prove model selection consistency, \cite{yang2016computational} assumed $g = p^{2c}$ for some $c\ge 1/2$ corresponding to $\gamma = p^{-2c}$ in our notation.
	This is the so-called information paradox of Zellner's $g$-priors \citep{liang2008mixtures}.
	We do not require this condition and just assume $\gamma = O(1)$.
\end{remark}

\subsection{Strong Model Selection Consistency}\label{subsec:sel}

When the recovery of the DAG is of interest, it is desirable to use a Bayesian procedure that guarantees the strong model selection consistency.
We show that the $\alpha$-posterior warrants this property under mild conditions.
%More specifically, it is proven that $\alpha$-posterior can detect the nonzero elements in $A_{0n}$ for any precision matrix $\Omega_{0n}\in \calU_p$ with $\bbP_0$-probability tending to 1.
As mentioned earlier, the Gaussian DAG model has an interpretation as a sequence of autoregressive model \eqref{model2}, which enables us to  adopt the state-of-the-art techniques for the selection consistency of the regression coefficient in \cite{martin2017empirical}. 

To use the results in \cite{martin2017empirical}, there are two main issues  that need to be addressed.
 The first   is the {\it restricted eigenvalue condition} for the design matrix.
%In our setting, it is a challenging problem because the design matrices corresponding to the each row of $A_n$ are random and correlated with each other (Cao et al., 2016)\cite{cao2016posterior}.
In our setting, the design matrices consist of columns of data matrix $\bfX_n$, thus each row follows a multivariate normal distribution.
%To establish the restricted eigenvalue condition for these random matrices, we use the concentration inequality for the Wishart distribution. 
We show that under the bounded eigenvalue condition \hyperref[A1]{(A1)}, the restricted eigenvalue condition   for any integer $R = o(n)$ automatically holds on some {\it large} set $N^c$ having $\bbP_0$-probability tending to 1 (Lemma 9.1 in Supplementary Material). 
A similar result appears in \cite{narisetty2014bayesian}. 
The second issue is more challenging than the first. 
\cite{martin2017empirical} considered only the known (fixed) residual variance case, which corresponds to the known $d_{0j}$ case in our setting.
The assumption on the known residual variance results in a relatively straightforward proof for selection consistency. 
We extended their techniques to the unknown residual variance case by applying (non-central) chi-square concentration inequalities for the estimated residual variances $\what{d}_{S_j}$ for some index set $S_j$, which is motivated by \cite{shin2015scalable}. 
It reveals that the ratio of the marginal posteriors $\pi_\alpha(S_j \mid \bfX_n) / \pi_\alpha(S_{0j}\mid \bfX_n)$ actually behaves like the ratio of the conditional posteriors given $d_{0j}$, $\pi_\alpha(S_j \mid d_{0j}, \bfX_n) / \pi_\alpha(S_{0j} \mid d_{0j}, \bfX_n)$, with $\bbP_0$-probability tending to 1, where $S_{0j}$ is the index set for the nonzero elements in the $j$th row of $A_{0n}$.

We also note here that unlike the Lasso type (or its variants) of results with the random design matrix \citep{wainwright2009sharp}, our theory does not require the {\it irrepresentable condition} on the true covariance matrix.
For example, \cite{yu2017learning} and \cite{khare2016convex} require the irrepresentable condition for the asymptotic properties of estimators in DAG models.
See section IV of \cite{wainwright2009sharp} for more details on the irrepresentable condition.

%Theorem \ref{thm:sel} proves strong model selection consistency.

\begin{theorem}[Strong model selection consistency]\label{thm:sel} 
	For given positive constants $0< \alpha < 1$, $0<\epsilon_0< 1/2$, $C_{\rm bm} > c_2 +2$ and an integer $s_0$, assume that $\Omega_{0n}$ satisfies conditions \hyperref[A1]{\rm (A1)}, \hyperref[A2]{\rm (A2)} and \hyperref[A3]{\rm (A3)}, i.e. $\Omega_{0n} \in \calU_p$.
	Consider model \eqref{model} and the ESC prior \eqref{SCprior} with Condition \hyperref[condP]{\rm (P)}.
	If $s_0\log p \le n \, c_3/2$,
	\bea
	\sup_{\Omega_{0n}\in \calU_p} \bbE_0 \Big[  \pi_\alpha(S_{A_n} \neq S_{A_{0n}} \mid \bfX_n) \Big] &=& o(1).
	\eea
\end{theorem}

The assumption $s_0\log p=o(n)$ or $s_0\log p \le cn$ for some constant $c>0$ is widely used in the high-dimensional sparse covariance or precision matrix estimation literature.
In Theorem \ref{thm:sel}, we assume less restrictive condition $s_0 \log p \le n\, c_3/2$, which automatically guarantees $s_0 \le R_j$ for all $j=2,\ldots,p$.
Note that the constant $c_3$ is defined in Condition \hyperref[condP]{(P)}.

It is worthwhile to compare our result to those of \cite{cao2016posterior}, \cite{yu2017learning} and \cite{khare2016convex}.
 Note that in these works it  is  also assumed that the ordering of variables is known. 
\cite{cao2016posterior} showed the strong model selection consistency using the hierarchical DAG-Wishart prior.
They assumed variants of conditions \hyperref[A1]{\rm (A1)}, \hyperref[A2]{\rm (A2)} and \hyperref[A3]{\rm (A3)}.
First, they relaxed  condition \hyperref[A1]{(A1)} by letting $\epsilon_{0,n} \to 0$ such that $(\log p/n)^{1/2 -1/(2+k)}  = o(\epsilon_{0,n}^4)$ for some $k>0$, instead of a fixed $\epsilon_0>0$.
Second, they assumed the same condition \hyperref[A2]{\rm (A2)} but further assumed $s_0^{2+k} \sqrt{\log p/n} = o(1)$ and  $(\log p/n)^{k/(4k+8)} \log n =o(1)$ and considered only the DAGs with the total number of edges at most $8^{-1} s_0 (n/ \log p)^{(1+k)/(2+k)}$, which can be restrictive. 
Note that, when $p\ge n$, it does not include the banded Cholesky factor having $s_0$ nonzero elements for each row.
Third, they assumed somewhat strong beta-min condition compared with \hyperref[A3]{(A3)}, which requires
$\min_{j,l: a_{0,jl}\neq 0} |a_{0,jl}|^2 \ge M_n^2 s_0^2 \epsilon_{0,n}^{-1}\, (\log p/n )^{1/(2+k)}$ 
for some $k>0$ and some sequence $M_n \to \infty$.
Thus, their assumptions on the tuple $(n, p, s_0)$ as well as the parameter class are  much more restrictive than  ours, except for the bounded eigenvalue condition.
Furthermore, the choice of hyperparameter in the hierarchical DAG-Wishart prior depends on the unknown sparsity parameter $s_0$, thus it is not adaptive to the unknown parameter.
More specifically, the hyperparameter $q_n$ in the hierarchical DAG-Wishart prior should be set at $q_n = s_0(\log p/n)^{1/(2+k)}$ for some $k>0$ to achieve the strong model selection consistency.

\cite{yu2017learning} suggested a penalized maximum likelihood estimation for the Cholesky factor of the precision matrix and proved the exact signed support recovery under the condition $\rho^{-2} \|D_{0n}\| \epsilon_0^{-1} (12\pi^2 s_0 + 32)\log p < n$.
They considered the class of precision matrices satisfying condition \hyperref[A1]{(A1)} and having a banded structure with the row-specific bandwidths $s_{0j}= |S_{0j}|$ such that $a_{0,jl} = 0$ for all $1\le l < j - s_{0j}$ and $2\le j \le p$. 
Thus, by taking $s_0 = \max_j s_{0j}$, their class satisfies conditions \hyperref[A2]{(A2)} and \hyperref[A4]{(A4)}.
They also assumed the beta-min condition, 
$\min_{j,l: a_{0,jl}\neq 0} |d_{0j}^{-1/2} a_{0,jl}| \ge 8 \rho^{-1} \sqrt{ 2 \|D_{0n}\| \log p/n } \big(4 \max_j \| \sg_{0n, S_{0j}}^{-1} \|_\infty + 5\epsilon_0^{-1} \big)$.
In general, it holds that $\| \sg_{0n, S_{0j}}^{-1} \|_\infty = O(s_{0j}^{1/2})$ without further assumption, thus the above condition implies that the minimum nonzero $|d_{0j}^{-1/2} a_{0,jl}|$ is bounded below by $ \sqrt{s_0 \log p/n}$ with respect to a constant multiple, thus stronger than  condition \hyperref[A3]{(A3)}. 
Furthermore, they assumed the irrepresentable condition 
$$\max_{2\le j\le p} \max_{1\le l \le j \atop l \in S_{0j}^c} \| (\sg_{0n})_{(l, S_{0j})} (\sg_{0n, S_{0j}})^{-1} \|_1 \le \frac{6(1- \rho)}{\pi^2}$$
for some constant $\rho \in (0,1]$.
Therefore, they only considered the banded Cholesky factor and used somewhat strong beta-min condition as well as the irrepresentable condition.
However, the comparison with our result (Theorem \ref{thm:sel}) is not straightforward because their exact signed support recovery property is stronger than the selection consistency proved in Theorem \ref{thm:sel}.

\cite{khare2016convex} proved the signed support recovery property of the convex sparse Cholesky selection (CSCS) method when the data vectors $X_1,\ldots,X_n$ are random sample from  a  sub-Gaussian distribution.
They assumed condition \hyperref[A1]{(A1)} as well as the (stronger) variants of conditions \hyperref[A2]{(A2)} and \hyperref[A3]{(A3)}: they assumed $\sum_{j=2}^p s_{0j} = o(n/\log n)$ (which is stronger than $s_0 \log p \le n c_3/2$) and 
$\min_{j,l: a_{0,jl}\neq 0} |a_{0,jl}|^2 \ge M_n s_0^2 \log n/n$ for some $M_n\to\infty$.
Furthermore, they considered only the moderate high-dimensional setting, i.e. $p = O(n^c)$ for some constant $c>0$.
They also required the irrepresentable condition similar to  those in \cite{yu2017learning}.

\subsection{Posterior Convergence Rates for Cholesky Factors}\label{subsec:conv_chol}

In this subsection, we derive the posterior convergence rates for the Cholesky factors in two different scenarios depending on the existence of the beta-min condition \hyperref[A3]{(A3)}.
At first, under the beta-min condition, we show the posterior convergence rates and minimax lower bounds with respect to the matrix $\ell_\infty$ norm and Frobenius norm.
The obtained posterior convergence rates are {\it nearly} minimax, and become exactly minimax if $\log p =O(s_0)$ and $\log j = O(s_{0j})$ for all $j=2,\ldots,p$.
We also derive the posterior convergence rate and minimax lower bound with respect to the matrix $\ell_\infty$ norm without the beta-min condition.
The obtained posterior convergence rate turns out to be nearly minimax, and it will be exactly minimax if $s_0 \le p^\beta$ for some $0<\beta<1$.
Note that regardless of the relation between $s_0$ and $p$, at least one of the scenarios achieves the minimax rate.
Especially, we attain the minimax rate for both scenarios if $C \log p \le s_0 \le p^\beta$ for some constant $C>0$.
Figure \ref{fig:minimax} describes the range for $s_0$ in which the minimax rate can be obtained.

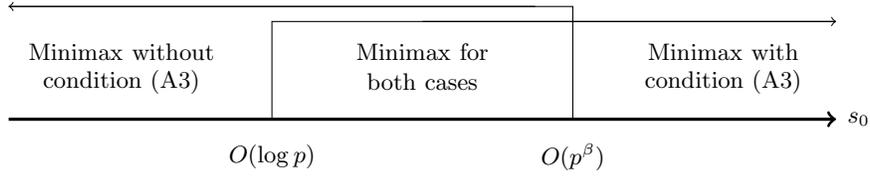
\begin{figure}\centering
	\begin{tikzpicture}
	\draw (-2,0) --(-2,1.3) -- (5.5,1.3);
	\draw [->] (-0,1.3) -- (5.5,1.3);
	\draw (2,0) --(2,1.5) -- (-5.5,1.5);
	\draw [->] (1.5,1.5) -- (-5.5,1.5);
	\draw [->][very thick] (-5.5,0) -- (5.5,0);
	\node at (5.8,0) {$s_0$};
	\node at (-2,-0.5) {$O(\log p)$};
	\node at (2,-0.5) {$O(p^\beta)$};
	\node at (-4,0.9) {Minimax without};
	\node at (-4,0.5) {condition (A3)};
	\node at (-0,0.9) {Minimax for};
	\node at (-0,0.5) {both cases};
	\node at (4,0.9) {Minimax with};
	\node at (4,0.5) {condition (A3)};
	\end{tikzpicture}
	\caption{
		For a given $0<\beta<1$, it describes the range for $s_0$ in which the minimax rate for the Cholesky factor can be obtained. 
		(A3) means the beta-min condition.
	}
	\label{fig:minimax}
\end{figure}

\subsubsection{Posterior Convergence Rates for Cholesky Factors under Beta-min Condition}

Define $\what{A}_n = (\what{a}_{jl} )$, where $(\what{a}_{jl})_{l \in  S_{0j}} = \what{a}_{S_{0j}}$ and $(\what{a}_{jl})_{l\in S_{0j}^c} = 0$.
Thus, $\what{A}_n$ is the empirical estimates of $A_{0n}$ with  true support $S_{A_{0n}}$.
To obtain the posterior convergence rate for the Cholesky factor, we use a divide and conquer strategy that is similar to \cite{lee2017optimal, lee2017estimating}.
Specifically, we decompose the posterior contraction probability into two parts as follows:
\bean\label{twoparts_chol}
&& \pi_\alpha \Big( \|A_n - A_{0n} \| \ge 2\epsilon_n' \mid \bfX_n  \Big) \nonumber \\
&\le& \pi_\alpha \Big(\|A_n - \what{A}_{n} \| \ge \epsilon_n' \mid \bfX_n \Big) 
+ \pi_\alpha \Big(\|\what{A}_{n} - A_{0n} \| \ge \epsilon_n' \mid \bfX_n \Big) \quad \,\,
\eean
for some positive sequence $\epsilon_n'$.
As in subsection \ref{subsec:sel}, we concentrate on a {\it large} set $N^c$ allowing us to handle the posterior contraction probability easily.
The first part of the right hand side of \eqref{twoparts_chol} describes how the posterior distribution concentrates around the empirical estimate $\what{A}_n$.
We use the selection consistency result in Theorem \ref{thm:sel}, and we focus only on the set $S_{A_n} = S_{A_{0n}}$.
It enables us to deal with the posterior distribution for $A_n$ easily, but with a cost of the beta-min condition \hyperref[A3]{(A3)} which is usually not essential for the convergence rate results.
Through the posterior distribution \eqref{post} given $S_{A_n} = S_{A_{0n}}$, we can obtain the contraction probability for $\|A_n -\what{A}_n\|$ using the concentration inequality for the chi-square random variables.
By taking expectation to the second part of the right hand side of \eqref{twoparts_chol}, it gives the contraction probability of $\what{A}_n$, $\bbP_0 \big[ \|\what{A}_{n} - A_{0n} \| \ge \epsilon_n'  \big]$.
%Theorem \ref{thm:chol_factor} shows the posterior convergence rates for the Cholesky factor $A_{0n}$ with respect to the matrix $\ell_\infty$ norm and Frobenius norm.

\begin{theorem}[Posterior convergence rates for $A_{0n}$ with beta-min condition]\label{thm:chol_factor}
	For given positive constants $0<\alpha<1$, $0<\epsilon_0< 1/2$, $C_{\rm bm} > c_2+2$ and an integer $s_0$, assume that $\Omega_{0n}$ satisfies conditions \hyperref[A1]{\rm (A1)}, \hyperref[A2]{\rm (A2)} and \hyperref[A3]{\rm (A3)}, i.e. $\Omega_{0n} \in \calU_p$.
	Consider model \eqref{model} and the ESC prior \eqref{SCprior} with Condition \hyperref[condP]{\rm(P)}.
	If $s_0 \log p =o(n)$,
	\bea
	\sup_{\Omega_{0n}\in \calU_p} \bbE_0 \Big[\pi_\alpha\big( \|A_n - A_{0n} \|_\infty \ge K_{\rm chol} \sqrt{s_0}\left(\frac{s_0 + \log p}{n} \right)^{1/2}   \,\,\big|\,\, \bfX_n \big)  \Big] 
	&=& o(1) , \\
	\sup_{\Omega_{0n}\in \calU_p} \bbE_0 \Big[\pi_\alpha\big( \|A_n - A_{0n} \|_F^2 \ge K_{\rm chol}   \frac{ \sum_{j=2}^p (s_{0j} +  \log j)}{n}   \,\,\big|\,\, \bfX_n \big)  \Big] 
	&=& o(1) 
	\eea
	for some constant $K_{\rm chol} >0$.
\end{theorem}

\cite{khare2016convex} obtained the convergence rate $\sum_{j=2}^p s_{0j} \lambda_n$ for estimating the Cholesky factor under the spectral norm in  a moderately high-dimensional setting, i.e.  $p = O(n^c)$ for some constant $c>0$, where $\lambda_n$ is the tuning parameter in CSCS method.
They also assumed condition \hyperref[A1]{(A1)} as well as the (stronger) variants of conditions \hyperref[A2]{(A2)} and \hyperref[A3]{(A3)} as described in section \ref{subsec:sel}.   
Because they assumed $\sqrt{\sum_{j=2}^p s_{0j} \log p/n}$ $= o(\lambda_n)$, $\sum_{j=2}^p s_{0j} \lambda_n$ is strictly slower than $(\sum_{j=2}^p s_{0j})^{3/2} \sqrt{\log p/n}$ in term of the rate, which implies that their convergence rate is slower than the posterior convergence rate obtained in this paper.

In fact, it turns out that the posterior convergence rates in Theorem \ref{thm:chol_factor} are nearly optimal.
Theorem \ref{theorem:AnLbound_betamin} describes that the rates of the frequentist minimax lower bounds for the class $\calU_p$, which are of independent interests.  Note that the rates of   Theorem \ref{thm:chol_factor}  are  exactly optimal if $\log p =O(s_0)$ and $\log j = O(s_{0j})$ for all $j=2,\ldots,p$ matching the minimax rates of Theorem \ref{theorem:AnLbound_betamin}.
%Since the posterior convergence rate cannot be faster than the rate of the minimax lower bound (Hjort et al., 2010)\cite{hjort2010bayesian}, the obtained results in Theorem \ref{thm:chol_factor} are nearly optimal in terms of the rate.
The key idea for proving the minimax lower bounds is to break down the model \eqref{model} into a set of linear regression models.

\begin{theorem}[Minimax lower bounds for $A_{0n}$ with beta-min condition]\label{theorem:AnLbound_betamin}
	For given positive constants $0<\alpha<1$, $\epsilon_0$, $C_{\rm bm}$ and an integer $s_0$, assume that $\Omega_{0n}$ satisfies conditions \hyperref[A1]{\rm (A1)}, \hyperref[A2]{\rm (A2)} and \hyperref[A3]{\rm (A3)}, i.e. $\Omega_{0n} \in \calU_p$.
	Consider model \eqref{model}. Then,
	\bea
	\inf_{\what{A}_n} \sup_{\Omega_{0n} \in \calU_p} \bbE_0 \|\what{A}_n - A_{0n} \|_\infty &\ge& c \cdot \frac{s_0}{\sqrt{n}} , \\
	\inf_{\what{A}_n} \sup_{\Omega_{0n} \in \calU_p} \bbE_0 \|\what{A}_n - A_{0n} \|_F^2 &\ge& c \frac{\sum_{j=2}^p s_{0j} }{n}
	\eea
	for some constant $c>0$, where the infimum  is taken over all possible estimators $\what{A}_n$.
\end{theorem}

\subsubsection{Posterior Convergence Rates for Cholesky Factors without Beta-min Condition}

%Martin et al. (2017)\cite{martin2017empirical} showed that their empirical Bayesian procedure achieve the minimax rate for the sparse linear regression models.
%We can extend their results to prove the posterior convergence rate for the Cholesky factor without beta-min condition.
For a given positive constant $\epsilon_0$ and a sequence of positive integers $s_0$, we define a class of precision matrices, 
\bea
\calU_p^0 \,\,=\,\, \calU_p^0(\epsilon_0, s_0)  &=& \bigg\{ \Omega \in \calC_p: \,\, \Omega \text{ satisfies \hyperref[A1]{\rm (A1)} and \hyperref[A2]{\rm (A2)}}\,\,  \bigg\}. 
\eea
Note that in the definition of $\calU_p^0$, we \emph{do not require the beta-min condition.}
Theorem \ref{theorem:chol_minimax_Ubound} gives the posterior convergence rate for the class  $\calU_p^0$.
For the Theorem \ref{theorem:chol_minimax_Ubound}, we use the ESC prior \eqref{SCprior} but let $d_j \sim IG(\nu_0/2, \nu_{0}')$ for some constant $\nu_{0}'>0$ instead of $\pi(d_j) \propto d_j^{-\nu_0/2 -1}$.
We call this the modified ESC (MESC) prior.
As mentioned before, Theorems \ref{thm:sel}, \ref{thm:chol_factor} and \ref{thm:conv} in section \ref{sec:main} also hold for the MESC prior, but we describe Theorems \ref{thm:sel}, \ref{thm:chol_factor} and \ref{thm:conv} with the ESC prior for the computational convenience.

We consider the denominator and numerator of the posterior probability $\pi_\alpha(\|A_n - A_{0n} \|_\infty \ge \epsilon_n')$ separately, for some positive sequence $\epsilon_n'$.
For any $j=2,\ldots,p$, let $R_{nj}(a_j, d_j) = L_{nj}(a_j, d_j)/ L_{nj}(a_{0j}, d_{0j})$ be the likelihood ratio, where 
$$L_{nj}(a_{j}, d_j) = (2\pi d_j)^{-n/2} \exp \big\{- \| \tilde{X}_j - \tilde{Z}_j a_{j} \|_2^2/(2d_j)  \big\}.$$
Dealing with the likelihood ratio $R_{nj}(a_j, d_j)$ is one of the main tasks for proving Theorem \ref{theorem:chol_minimax_Ubound}. 
Lemma 10.1, Lemma 10.2 and Lemma 10.3 in Supplementary Material describe how we can deal with the likelihood ratio $R_{nj}(a_j,d_j)$ in the denominator and numerator.

\begin{theorem}[Posterior convergence rate for $A_{0n}$ without beta-min condition]\label{theorem:chol_minimax_Ubound}
	For a given positive constant $0<\alpha <1$, $0<\epsilon_0< 1/2$ and an integer $s_0$, assume that $\Omega_{0n}$ satisfies conditions \hyperref[A1]{\rm (A1)} and \hyperref[A2]{\rm (A2)}, i.e. $\Omega_{0n} \in \calU_p^0$.
	Consider model \eqref{model} with  the MESC prior with Condition \hyperref[condP]{\rm(P)}.
	If $s_0 \log p = o(n)$ and $\nu_0 =O(1)$,
	then
	\bea
	\sup_{ \Omega_{0n} \in \calU_p^0 } \bbE_0 \left[  \pi_\alpha\big( \|A_n - A_{0n} \|_\infty \ge K_{\rm chol}' \, s_0 \left(\frac{ \log p}{n} \right)^{1/2}   \,\,\big|\,\, \bfX_n \big) \right] &=& o(1)
	\eea
	for some constant $K_{\rm chol}' >0$.
\end{theorem}

\cite{yu2017learning} obtained the convergence rate $\max_j \|\sg_{0n, S_{0j}}^{-1}\|_\infty \cdot \|A_{0n}\|_\infty s_0 \sqrt{\log p/n} + \max_j \|\sg_{0n, S_{0j}}^{-1}\|_\infty^2 s_0^2 \log p/n$ for the Cholesky factor with respect to the matrix $\ell_\infty$ norm.
As stated before, they assumed condition \hyperref[A1]{\rm (A1)}, the banded Cholesky factor structure (which corresponds to conditions \hyperref[A2]{\rm (A2)} and \hyperref[A4]{\rm (A4)} in this paper) and the irrepresentable condition.   
Note that their convergence rate coincides with ours only if $\|A_{0n}\|_\infty$ and $\max_j \|\sg_{0n, S_{0j}}^{-1}\|_\infty$ are bounded and $s_0^2 \log p = O(n)$.

To the best of our knowledge, it is the first result on the posterior convergence rate for the high-dimensional sparse Cholesky factor without the beta-min condition.
Interestingly, the obtained posterior convergence rate is the same with the minimax convergence rate for the $s_0$-sparse coefficient vector in the regression models when $s_0 \le p^{\beta}$ for some $0<\beta<1$.
Note that the condition $s_0 \le p^{\beta}$ is not restrictive in the high-dimensional setting, because this condition is met if $n \le p^{\beta}$.
Theorem \ref{theorem:chol_minimax_Lbound} confirms that the above posterior convergence rate is nearly minimax for any $\Omega_{0n} \in \calU_p^0$.
Similar to Theorem \ref{theorem:AnLbound_betamin}, the key idea for proving Theorem \ref{theorem:chol_minimax_Lbound} is to break down the model into a set of linear regression models.

\begin{theorem}[Minimax lower bound for $A_{0n}$ without beta-min condition]\label{theorem:chol_minimax_Lbound}
	For a given constant $\epsilon_0$ and an integer $s_0$, assume that $\Omega_{0n}$ satisfies conditions \hyperref[A1]{\rm (A1)} and \hyperref[A2]{\rm (A2)}, i.e. $\Omega_{0n} \in \calU_p^0$.
	Consider model \eqref{model}. 
	Then,
	\bea
	\inf_{\what{A}_n} \sup_{\Omega_{0n} \in \calU_p^0} \bbE_0 \|\what{A}_n - A_{0n} \|_\infty &\ge&
	c \cdot s_0 \left(\frac{ \log (p/s_0)}{n} \right)^{1/2} 
	\eea
	for some constant $c>0$.
\end{theorem}

\begin{remark}
	If we assume $s_0 \le p^{\beta}$ for some $0<\beta<1$, then $\log (p/s_0)$ has the same rate with $\log p$, 
	and the rate of the mininax lower bound in Theorem \ref{theorem:chol_minimax_Lbound} becomes $s_0 \sqrt{\log p/n}$.
	This assumption is reasonable in the high-dimensional setting.
\end{remark}

\subsection{Posterior Convergence Rates for Precision Matrices}\label{subsec:conv}

In this subsection, we derive the posterior convergence rates for the precision matrices with respect to various matrix norms.
Define $\what{\Omega}_n = (I_p - \what{A}_n)^T \what{D}_n^{-1} (I_p - \what{A}_n)$, where $\what{A}_n$ and $\what{D}_n= diag(\what{d}_{S_{0j}})$ are the empirical estimates of $A_{0n}$ and $D_{0n}$ with the true support $S_{A_{0n}}$.
Similar to the previous subsection, we use the divide and conquer strategy to deal with the posterior probability.
For the recovery of $\Omega_{0n} = (I_p - A_{0n})^T D_{0n}^{-1} (I_p - A_{0n})$, we further assume  condition \hyperref[A4]{\rm (A4)}.
For given positive constants $\epsilon_0, C_{\rm bm}$ and a sequence of positive integers $s_0$, define the parameter class as follows:
\bea
\calU_p^* \,\,=\,\,\calU_p^*(\epsilon_0, s_0, \alpha, C_{\rm bm})  &=& \bigg\{ \Omega \in \calC_p: \,\, \Omega \text{ satisfies \hyperref[A1]{\rm (A1)}-\hyperref[A4]{\rm (A4)}}\,\,  \bigg\} .
\eea
Theorem \ref{thm:conv} shows the posterior convergence rates for the precision matrix with respect to the spectral norm and matrix $\ell_\infty$ norm.

\begin{theorem}[Posterior convergence rates for $\Omega_{0n}$]\label{thm:conv}
	For given positive constants $0< \alpha < 1$, $0<\epsilon_0< 1/2$, $C_{\rm bm}> c_2+2$ and an integer $s_0$, assume that $\Omega_{0n}$ satisfies conditions \hyperref[A1]{\rm (A1)}-\hyperref[A4]{\rm (A4)}, i.e. $\Omega_{0n} \in \calU_p^*$.
	Consider model \eqref{model} and the ESC prior \eqref{SCprior} with Condition \hyperref[condP]{\rm(P)} and $\nu_0^2 = O(n\log p)$.
	If $s_0^{3/2}(s_0 + \log p) =o(n)$, then
	\bea
	&&\sup_{\Omega_{0n}\in \calU_p^*} \bbE_0 \Big[\pi_\alpha\big( \|\Omega_n - \Omega_{0n} \| \ge K_{\rm conv} s_0^{3/4}\left(\frac{s_0 + \log p}{n} \right)^{1/2}   \,\,\big|\,\, \bfX_n \big)  \Big] 
	= o(1) ,
	\eea
	and, if $s_0 (s_0 + \log p) =o(n)$, then
	\bea
	\sup_{\Omega_{0n}\in \calU_p^*} \bbE_0 \Big[\pi_\alpha\big( \|\Omega_n - \Omega_{0n} \|_\infty \ge K_{\rm conv} \cdot \|I_p - A_{0n}\|_\infty s_0 \left(\frac{s_0 + \log p}{n} \right)^{1/2}   \,\,\big|\,\, \bfX_n \big)  \Big] 
	= o(1)
	\eea
	for some constant $K_{\rm conv} >0$.
\end{theorem} 
%\footnote{ 
%	My guess:
%	It is hard to compare the typical minimax lower bound from Tony Cai's techniques and the minimax lower bound with beta-min condition.
%	(Note: If we use the typical techniques, we might have the minimax lower bound $s_0 \sqrt{\log p/n}$)
%	Because,
%	(1) The typical minimax lower bound techniques do not consider the beta-min condition when set the subclass achieving the minimax lower bound. For example, Tony Cai set the nonzero elements at $\nu\sqrt{\log p/n}$ and considered very small $\nu>0$. In our case, the beta-min condition requires $\nu$ is not very small. Thus, we cannot use the Tony Cai's results.
%	(2) Consider the posterior convergence rate for $\|A_n - A_{0n}\|_\infty$. As we proved, the minimax rates are different depend on the existence of the beta-min condition.
%}

It is worthwhile to compare our result to other existing results.
\cite{cao2016posterior} obtained the posterior convergence rate, $s_0^2 \, \epsilon_{0,n}^{-2} \sqrt{\log p/n}$, for the precision matrix with respect to the spectral norm.
As discussed in section \ref{subsec:sel}, they assumed variants of conditions \hyperref[A1]{\rm (A1)}, \hyperref[A2]{\rm (A2)} and \hyperref[A3]{\rm (A3)}.
They further assumed the condition \hyperref[A4]{\rm (A4)}.
Although they did not state clearly that condition \hyperref[A4]{(A4)} was used, this condition is necessary to use Lemma 3.1 of \cite{xiang2015high} in their proof.  
If we assume the bounded eigenvalue condition \hyperref[A1]{(A1)}, their convergence rate becomes $s_0^2 \sqrt{\log p/n}$, which is slower than the convergence rate in Theorem \ref{thm:conv}.
Note that they assumed $s_0^{2+k} \sqrt{\log p/n} = o(1)$ for some constant $k>0$, which is stronger than our assumption $s_0^{3/2}(s_0 + \log p) =o(n)$.    
Thus, we obtain the faster posterior convergence rates under more general condition on the tuple $(n,p,s_0)$ and parameter class, except for the bounded eigenvalue condition.

\cite{yu2017learning} considered the parameter class they used to prove the strong model selection consistency, but dropped the beta-min condition.
They derived the convergence rate 
$$ \max_j \| (\sg_{0n, S_{0j}} )^{-1}\|_\infty \| D_{0n}^{-1/2}(I_p - A_{0n})\|_\infty \, s_0 \, \Big( \frac{\log p}{n}\Big)^{1/2}$$ 
for the precision matrix with respect to the matrix $\ell_\infty$ norm.
Note that this convergence rate depends on the rate of $ \max_j \| (\sg_{0n, S_{0j}} )^{-1}\|_\infty \| D_{0n}^{-1/2}(I_p - A_{0n})\|_\infty$.
In general, it holds that $\max_j \| (\sg_{0n, S_{0j}} )^{-1}\|_\infty = O(s_{0j}^{1/2})$.
Thus, their convergence rate is slower than the posterior convergence rate in Theorem \ref{thm:conv}, without a further assumption on $\sg_{0n}$ guaranteeing 
$\max_j \| (\sg_{0n, S_{0j}} )^{-1}\|_\infty   = O(\sqrt{(s_0/\log p) + 1}).$
%Furthermore, as a by-product, they obtained the same convergence rate for the precision matrix with respect to the spectral norm.
%It is slower than the posterior convergence rate in \eqref{conv_pre_spec} without a further assumption guaranteeing $ \max_j \| (\sg_{0n, S_{0j}} )^{-1}\|_\infty \| D_{0n}^{-1/2}(I_p - A_{0n})\|_\infty s_0^{1/4}\sqrt{\log p} =O(\sqrt{s_0 + \log p})$. 

%%%%%%%%%%%%%%%%%%%%%%%%%%%%%%%%%%%%%%%%%%%%%
\section{Numerical Results}\label{sec:numerical}

The use of the ESC prior not only guarantees optimal or near optimal asymptotic properties but also allows us to conduct the posterior inference easily.
%To show the practical performance of our method, we carried out some simulation studies and real data analysis.
%To show the practical performance of our method, we carried out some simulation studies.
In this section, we carry out simulation studies to illustrate the model selection performance of our method.
For the comparison, we chose state-of-the-art methods for high-dimensional sparse DAG models and measured the performance of each method.
The simulation study confirms that our ESC prior outperforms the other existing methods.

%\subsection{Simulation Study}\label{subsec:simul}

%In this subsection, we carry out simulation studies to illustrate the model selection performance of our method.
%For the comparison, we chose state-of-the-art methods for high-dimensional sparse DAG models and measured the performance of each method.
% based on some criteria.

%\subsubsection{Metropolis-Hastings Algorithm}
\subsection{Metropolis-Hastings Algorithm}

Recall that, by \eqref{post}, the marginal posterior distribution for $S_j \subseteq \{1,\ldots, j-1\}$ can be derived  analytically  as
\bean\label{marginal_Sj}
\pi_\alpha(S_j \mid \bfX_n) &\propto& \pi_j(S_j) \left(1+ \frac{\alpha}{\gamma} \right)^{-\frac{|S_j|}{2}} (\what{d}_{S_j})^{- \frac{\alpha n + \nu_0}{2}}
\eean
for all $j=2,\ldots,p$, up to some normalizing constants.
Thus, we can run the Rao-Blackwellized Metropolis-Hastings algorithm for each $j=2,\ldots,p$ in parallel.
Here, we briefly summarize the algorithm used for the inference, where $L$ is the number of posterior samples:
\begin{enumerate}
	\item[] Run the following steps for  $j=2,\ldots, p$.
	\begin{enumerate}
		\item Set the initial value $S_j^{(1)}$.
		\item For each $l = 2,\ldots, L$,
		\begin{enumerate}
			\item sample $S_j^{new} \sim q(\cdot \mid S_j^{(l-1)})$;
			\item compute the acceptance probability
			\bea
			p_{acc} &=& \min \left\{ 1, \,\, \frac{\pi_\alpha(S_j^{new} \mid \bfX_n) q(S_j^{(l-1)} \mid S_j^{new}) }{\pi_\alpha(S_j^{(l-1)} \mid \bfX_n) q(S_j^{new} \mid S_j^{(l-1)})}  \right\},
			\eea
			and set $S_j^{(l)} = S_j^{new}$ with probability $p_{acc}$; otherwise, set $S_j^{(l)}= S_j^{(l-1)}$.
		\end{enumerate}
	\end{enumerate}
\end{enumerate}
We chose the kernel $q(S' \mid S)$ which forms a new set $S'$ by changing a randomly selected  nonzero component to 0 with probability $0.5$ or by changing a randomly selected  zero component to 1 with probability $0.5$.

The marginal posterior for $S_j$ is controlled by the prior $\pi_j(S_j)$, the penalty term $(1+ \alpha/\gamma)^{-|S_j|/2}$ and the estimated residual variance $\what{d}_{S_j}$.
The data favor to minimize the estimated residual while the prior and penalty term give more weight to the simpler models. The marginal posterior of $S_j$ will find the model that balances  data tracking and model complexity.

To use the Metropolis-Hastings algorithm, we need to choose the tuning parameters. 
Apart from the impact on theoretical results, the choice of tuning parameters  also influences the practical performance.
As \cite{martin2017empirical} suggested, we set $\alpha=0.999$ to mimic the Bayesian model with the ordinary likelihood.
In the simulation study, as long as $1-\alpha$ is close to zero, the performance was not dependent on the choice of $\alpha$.
The hyperparameters were chosen as $\gamma=0.1$, $\nu_0=0$, $c_1 = 0.0005$ and $c_2 = 2$ to satisfy Condition \hyperref[condP]{(P)}.

%\subsubsection{Simulation Setting}
\subsection{Simulation Setting}

For the simulation study, we considered the sparse Cholesky settings similar to those used in \cite{khare2016convex}.
We randomly chose 3\% or 4\% of the lower triangular entries of the Cholesky factor $A_{0n}$ and sampled their values from a uniform distribution on $[-0.7, -0.3] \cup [0.3, 0.7]$. 
The remaining entries were set to zero.
The entries of the diagonal matrix $D_{0n}$ were sampled from a uniform distribution on $[2, 5]$. 
Given the precision matrix $\Omega_{0n} = (I_p - A_{0n})^T D_{0n}^{-1}(I_p - A_{0n})$, the data sets were generated from the multivariate normal distribution $N_p(0, \Omega_{0n}^{-1})$ with $(n=100,p=300)$ and $(n=200,p=500)$.
%For each $(n,p)$, we generated the true precision matrix $\Omega_{0n}$ and data set as described.

%\subsubsection{Other Competing Methods}
\subsection{Other Competing Methods}

We compared the model selection performance of our method with those of other existing methods: the empirical Bayes (EB) procedure in \cite{martin2017empirical}, which we will denote as EB.M, hierarchical DAG-Wishart (DAG-W) prior \citep{cao2016posterior} and convex sparse Cholesky selection (CSCS) \citep{khare2016convex}.

\begin{enumerate}
	\item (EB.M) : 
	Because EB.M is originally proposed for the regression coefficient estimation, it can be applied independently to estimate each $a_{0j}$ for $j=2,\ldots,p$.
	We set the hyperparameters $\alpha$, $\gamma$, $c_1$ and $c_2$ to be the same as those in our setting for a fair comparison.
	Note that \cite{martin2017empirical} used $\gamma=0.001$, $c_1 = 1$ and $c_2 = 0.05$ in their simulation study, but in our simulations, these choices did not yield better results: they tended to make unacceptably large FDR values. 
	The key difference between our method and EB.M is on how to infer the diagonal matrix $D_{n}$.
	\cite{martin2017empirical} suggested plugging in the cross-validation based Lasso residual sum of squares estimate \citep{reid2016study} of $d_{0j}$, while we impose a prior on $d_j$ and integrate it out to obtain the marginal posterior for $S_j$.
	Thus, EB.M ignores the uncertainty of $d_j$ and replaces it with a plug-in estimate.
	
	\item (DAG-W) : 
	The hierarchical DAG-Wishart prior \citep{cao2016posterior} enables one to calculate the marginal posterior for the DAG analytically. 
	Note that, in \cite{cao2016posterior}, they conducted {\it log-posterior score search algorithm} instead of Markov chain Monte Carlo (MCMC) algorithm.
	Basically, they generated sets of candidate graphs by using frequentist approaches and thresholding the modified Cholesky factor of $(n^{-1} \bfX_n^T \bfX_n + 0.5 I_p)^{-1}$, and the graph which maximizes the log-posterior was chosen as the final estimate.  
	In our simulation study, we implemented the log-posterior score search algorithm as well as Metropolis-Hastings algorithm, using the marginal posterior for the DAG, for a comprehensive comparison. 
	For the implementation, we set the shape parameters at $\alpha_j(\calD) = S_j + 10$ and the scale matrix at $U_n = I_p$ as they suggested, where $\calD$ is the DAG corresponding to $\{S_j\}_{j=2}^p$.
	The critical part is the choice of the hyperparameter $q_n$, which is the individual edge probability. 
	It was shown that the choice of $q_n = e^{-\eta_n n}$ leads to strong model selection consistency, where $\eta_n = s_0 (\log p/n)^{1/(2+k)}$ for some $k>0$.
	Thus, the theoretical choice of $q_n$ depends on the unknown parameter $s_0$ and constant $k>0$.
	Furthermore, even with $s_0=1$ and $k=0$, the resulting $q_n$ is too small,  which does not allow the posterior to explore the model space efficiently. 
	We observed that the choice $q_n = e^{-\eta_n n}$ makes the posterior stuck in very small size models and not able to detect the true model.
	For example, for the setting $(n=100, p=300)$ with the sparsity 3\%, the corresponding posterior with $q_n = e^{-\eta_n n}$ concluded that the true Cholesky factor is a zero matrix, i.e. it never selected any nonzero variable.  
	Thus, in our simulation study, we conducted the simulation only for two choices, $q_n=0.01$ and $q_n = 0.001$, although they might not guarantee the strong model selection consistency. 
	For the log-posterior score search, we chose $q_n = e^{-\eta_n n}$ as in \cite{cao2016posterior}. 
%	Note that the case $q_n =0.01$ is quite close the true setting ($2 \%$), so it is not a fair comparison.
	
	\item (CSCS) : 
	We chose the CSCS method \citep{khare2016convex} as a state-of-the-art frequentist  competitor.
	% in the frequentist side.
	The tuning parameter $\lambda_n$ in the CSCS method was selected by the Bayesian Information Criterion (BIC)-like measure which is defined  in section 2.3 of \cite{khare2016convex}.
	We calculated the values of BIC-like measure for $\lambda_n$ from 0.1 to 5.1 with an increment of 0.1.
	The value of $\lambda_n$ minimizing the BIC-like measure was chosen for the estimation.
\end{enumerate}

%\subsubsection{Results}
\subsection{Results}

%We ran 5 Markov Chain Monte Carlo (MCMC) chains for each data set to conduct posterior inferences and diagnose the convergence of chains.
We ran the Metropolis-Hastings algorithm for each data set to conduct posterior inferences.
Every MCMC chain ran for 24,000 iterations with a burn-in period of 4,000, so  we obtained 20,000 posterior samples. 
We used the models selected by the CSCS method as the initial states for MCMC chains.
We constructed the final model by collecting indices with inclusion probabilities, $\pi( a_{jl} \neq 0 \mid \bfX_n )$, exceeding 0.5.

To measure the model selection performance, the number of errors, false discovery rate (FDR), true positive rate (TPR)  and inclusion probabilities were reported.
We calculated the mean inclusion probability for zero entries in $A_{0n}$ and denote it by $\bar{p}_0$.
Similarly, the mean inclusion probability for nonzero entries in $A_{0n}$ is denoted by $\bar{p}_1$.
More specifically, we calculated
\bea
\bar{p}_0 &=& \frac{1}{\sum_{j=2}^p(j-1 - s_{0j}) } \sum_{j=2}^p \sum_{l \notin S_{0j}} \pi(a_{jl} \neq 0 \mid \bfX_n ), \\
\bar{p}_1 &=& \frac{1}{\sum_{j=2}^p s_{0j} } \sum_{j=2}^p \sum_{l \in S_{0j}} \pi(a_{jl} \neq 0 \mid \bfX_n ).
\eea
%For the diagnosis, we calculated the Gelman-Rubin (GR) scale factor (Gelman and Rubin, 1992)\cite{gelman1992inference} for the statistics $R_{S_j}^2 = \tilde{X}_j^T \tilde{P}_{S_j} \tilde{X}_j / \|\tilde{X}_j\|_2^2$ as in Yang et al. (2016)\cite{yang2016computational}.
%Typically, the GR scale factor should be close to 1 if the chains have reached stationary. 
%In our settings, we have $p-1$ independent GR scale factors for $R_{S_2}^2, \ldots, R_{S_p}^2$, so we checked all the GR scale factors and reported the percentage of them which are larger than 1.5.

\begin{table}
	\centering
	\caption{
		ESC, EB.M, DAG-W and CSCS denote our method (empirical sparse Cholesky prior), the empirical Bayes procedure proposed by Martin et al. (2017), the hierarchical Bayesian model using DAG-Wishart prior (Cao et al., 2016) and Convex Sparse Cholesky Selection (Khare et al., 2016), respectively. 
		Sp: sparsity;
		FDR: false discovery rate; TPR: true positive rate; $\bar{p}_0$: the mean inclusion probability for zero entries in $A_{0n}$; $\bar{p}_1$: the mean inclusion probability for nonzero entries in $A_{0n}$. 
	}\vspace{.15cm}
	\begin{tabular}{cc ccccc}
		\hline
		$(n,p, \text{Sp})$ & Method & \# of errors & FDR & TPR & $\bar{p}_0$ & $\bar{p}_1$   \\ \hline
		\multicolumn{1}{c}{\multirow{5}{*}{(100, 300, 3\%)}}  & ESC & 264 & 0.0361 & 0.8349 & 0.0071 & 0.8321  \\
		& EB.M & 419 & 0.1083 & 0.7836 & 0.0041 & 0.7828  \\
		& DAG-W$(q_n=0.01)$ & 285 & 0.0208 & 0.8052 & 0.0024 & 0.8036  \\
		& DAG-W$(q_n=0.001)$ & 462 & 0.0122 & 0.6647 & 0.0006 & 0.6688  \\
		& DAG-W(log-score) & 1194 & 0.0065 & 0.1130 & $\cdot$ & $\cdot$  \\
		& CSCS & 2188 & 0.6433 & 0.7799 & $\cdot$  & $\cdot$  \\ \hline
		\multicolumn{1}{c}{\multirow{5}{*}{(100, 300, 4\%)}}  & ESC & 389 & 0.0494 & 0.8261 & 0.0084 & 0.8194  \\
		& EB.M & 325 & 0.0347 & 0.7866 & 0.0020 & 0.7815  \\
		& DAG-W$(q_n=0.01)$ & 422 & 0.0295 & 0.7887 & 0.0032 & 0.7873  \\
		& DAG-W$(q_n=0.001)$ & 644 & 0.0216 & 0.6555 & 0.0011 & 0.6556  \\
		& DAG-W(log-score) & 1619 & 0.0056 & 0.0981 & $\cdot$ & $\cdot$  \\
		& CSCS & 4025 & 0.7766 & 0.8045 & $\cdot$  & $\cdot$  \\ \hline
		\multicolumn{1}{c}{\multirow{5}{*}{(200, 500, 3\%)}} & ESC & 103 & 0.0118 & 0.9842 & 0.0039 & 0.9796  \\
		& EB.M & 212 & 0.0075 & 0.9506 & 0.0005 & 0.9509  \\
		& DAG-W$(q_n=0.01)$ & 98 & 0.0049 & 0.9786 & 0.0010 & 0.9773  \\
		& DAG-W$(q_n=0.001)$ & 182 & 0.0022 & 0.9535 & 0.0002 & 0.9519  \\
		& DAG-W(log-score) & 4285 & 0.0000 & 0.1412 & $\cdot$ & $\cdot$  \\
		& CSCS & 10214 & 0.7397 & 0.9388 & $\cdot$  & $\cdot$  \\ \hline
		\multicolumn{1}{c}{\multirow{5}{*}{(200, 500, 4\%)}} & ESC & 153 & 0.0061 & 0.9754 & 0.0043 & 0.9650  \\
		& EB.M & 281 & 0.0038 & 0.9473 & 0.0005 & 0.9457  \\
		& DAG-W$(q_n=0.01)$ & 163 & 0.0041 & 0.9713 & 0.0011 & 0.9684  \\
		& DAG-W$(q_n=0.001)$ & 295 & 0.0017 & 0.9425 & 0.0002 & 0.9416  \\
		& DAG-W(log-score) & 4341 & 0.0000 & 0.1301 & $\cdot$ & $\cdot$  \\
		& CSCS & 14632 & 0.7550 & 0.9285 & $\cdot$  & $\cdot$  \\ \hline
	\end{tabular}\label{table:sim}
\end{table}

The simulation results are summarized in Table \ref{table:sim}. 
%As expected, the overall performance of our method is better than that of EB.M.
%Furthermore, our method seems to reach stationary faster than EB.M based on the reported GR scale factors.
The ESC prior performs generally better than the other competing methods.
The EB.M works reasonably well, but the overall performance is worse than that of ESC prior.
The DAG-Wishart prior tends to have low TPR and mean inclusion probability $\bar{p}_1$ when $q_n=0.001$.
Note that when $q_n =0.01$, which is chosen to be close to the unknown true sparsity level, the DAG-Wishart prior performs reasonably well, but the ESC prior still works better. 
However,  the true sparsity is  in general unknown, so  fitting $q_n$ close to the true sparsity is a challenging task in practice.   
The log-posterior score search algorithm for DAG-Wishart is computationally efficient even for large $p$, but tends to have low FDR as well as TPR in our settings.  
The CSCS method has high TPR values, but at the same time, it has high FDR values.
Thus, from the simulation study, we confirm that our ESC prior not only has nice theoretical properties but also practically outperforms the other existing methods.

\section{Acknowledgement}

We thank Ryan Martin for helpful discussions about the techniques for proving selection consistency. We would like to thank two referees for their valuable comments which have led to improvements of an earlier version of the paper. 
We would also like to thank Kshitij Khare and Syed Rahman for sharing with us their code to implement the CSCS method \citep{khare2016convex}.
Kyoungjae Lee thanks Xuan Cao for sharing her code to implement the log-posterior score search algorithm and helpful discussions about the DAG-Wishart prior. 
%The contribution of KL and LL is funded by NSF grants IIS 1663870 and   DMS  Career 1654579, and a DARPA grant  N66001-17-1-4041.

\bibliographystyle{imsart-nameyear}
\bibliography{sparse-cholcov}

\begin{thebibliography}{54}
% BibTex style file: imsart-nameyear.bst, 2013-01-28
% Default style options (sort=1,type=nameyear).
% Used options (sort=1,type=nameyear).

\bibitem[\protect\citeauthoryear{Arias-Castro and
  Lounici}{2014}]{arias2014estimation}
\begin{barticle}[author]
\bauthor{\bsnm{Arias-Castro},~\bfnm{Ery}\binits{E.}} \AND
  \bauthor{\bsnm{Lounici},~\bfnm{Karim}\binits{K.}}
(\byear{2014}).
\btitle{Estimation and variable selection with exponential weights}.
\bjournal{Electron. J. Stat.}
\bvolume{8}
\bpages{328--354}.
\end{barticle}
\endbibitem

\bibitem[\protect\citeauthoryear{Banerjee and
  Ghosal}{2014}]{banerjee2014posterior}
\begin{barticle}[author]
\bauthor{\bsnm{Banerjee},~\bfnm{Sayantan}\binits{S.}} \AND
  \bauthor{\bsnm{Ghosal},~\bfnm{Subhashis}\binits{S.}}
(\byear{2014}).
\btitle{Posterior convergence rates for estimating large precision matrices
  using graphical models}.
\bjournal{Electron. J. Stat.}
\bvolume{8}
\bpages{2111--2137}.
\end{barticle}
\endbibitem

\bibitem[\protect\citeauthoryear{Banerjee and
  Ghosal}{2015}]{banerjee2015bayesian}
\begin{barticle}[author]
\bauthor{\bsnm{Banerjee},~\bfnm{Sayantan}\binits{S.}} \AND
  \bauthor{\bsnm{Ghosal},~\bfnm{Subhashis}\binits{S.}}
(\byear{2015}).
\btitle{Bayesian structure learning in graphical models}.
\bjournal{J. Multivariate Anal.}
\bvolume{136}
\bpages{147--162}.
\end{barticle}
\endbibitem

\bibitem[\protect\citeauthoryear{Ben-David et~al.}{2015}]{ben2015high}
\begin{barticle}[author]
\bauthor{\bsnm{Ben-David},~\bfnm{Emanuel}\binits{E.}},
  \bauthor{\bsnm{Li},~\bfnm{Tianxi}\binits{T.}},
  \bauthor{\bsnm{Massam},~\bfnm{Helene}\binits{H.}} \AND
  \bauthor{\bsnm{Rajaratnam},~\bfnm{Bala}\binits{B.}}
(\byear{2015}).
\btitle{High dimensional Bayesian inference for Gaussian directed acyclic graph
  models}.
\bjournal{arXiv:1109.4371v5}.
\end{barticle}
\endbibitem

\bibitem[\protect\citeauthoryear{Bhattacharya, Pati and
  Yang}{2018}]{bhattacharya2016bayesian}
\begin{barticle}[author]
\bauthor{\bsnm{Bhattacharya},~\bfnm{Anirban}\binits{A.}},
  \bauthor{\bsnm{Pati},~\bfnm{Debdeep}\binits{D.}} \AND
  \bauthor{\bsnm{Yang},~\bfnm{Yun}\binits{Y.}}
(\byear{2018}).
\btitle{Bayesian fractional posteriors}.
\bjournal{Annals of Statistics, to appear}.
\end{barticle}
\endbibitem

\bibitem[\protect\citeauthoryear{Bickel and
  Levina}{2008}]{bickel2008regularized}
\begin{barticle}[author]
\bauthor{\bsnm{Bickel},~\bfnm{Peter~J}\binits{P.~J.}} \AND
  \bauthor{\bsnm{Levina},~\bfnm{Elizaveta}\binits{E.}}
(\byear{2008}).
\btitle{Regularized estimation of large covariance matrices}.
\bjournal{Ann. Statist.}
\bvolume{36}
\bpages{199--227}.
\end{barticle}
\endbibitem

\bibitem[\protect\citeauthoryear{Boucheron, Lugosi and
  Massart}{2013}]{boucheron2013concentration}
\begin{bbook}[author]
\bauthor{\bsnm{Boucheron},~\bfnm{S.}\binits{S.}},
  \bauthor{\bsnm{Lugosi},~\bfnm{G.}\binits{G.}} \AND
  \bauthor{\bsnm{Massart},~\bfnm{P.}\binits{P.}}
(\byear{2013}).
\btitle{Concentration Inequalities: A Nonasymptotic Theory of Independence}.
\bpublisher{OUP Oxford}.
\end{bbook}
\endbibitem

\bibitem[\protect\citeauthoryear{B{\"u}hlmann and van~de
  Geer}{2011}]{buhlmann2011statistics}
\begin{bbook}[author]
\bauthor{\bsnm{B{\"u}hlmann},~\bfnm{P.}\binits{P.}} \AND
  \bauthor{\bparticle{van~de} \bsnm{Geer},~\bfnm{S.}\binits{S.}}
(\byear{2011}).
\btitle{Statistics for High-Dimensional Data: Methods, Theory and
  Applications}.
\bseries{Springer Series in Statistics}.
\bpublisher{Springer Berlin Heidelberg}.
\end{bbook}
\endbibitem

\bibitem[\protect\citeauthoryear{Cai, Liu and Zhou}{2016}]{cai2016estimating}
\begin{barticle}[author]
\bauthor{\bsnm{Cai},~\bfnm{T~Tony}\binits{T.~T.}},
  \bauthor{\bsnm{Liu},~\bfnm{Weidong}\binits{W.}} \AND
  \bauthor{\bsnm{Zhou},~\bfnm{Harrison~H}\binits{H.~H.}}
(\byear{2016}).
\btitle{Estimating sparse precision matrix: Optimal rates of convergence and
  adaptive estimation}.
\bjournal{Ann. Statist.}
\bvolume{44}
\bpages{455--488}.
\end{barticle}
\endbibitem

\bibitem[\protect\citeauthoryear{Cai, Ma and Wu}{2015}]{cai2015optimal}
\begin{barticle}[author]
\bauthor{\bsnm{Cai},~\bfnm{T~Tony}\binits{T.~T.}},
  \bauthor{\bsnm{Ma},~\bfnm{Zongming}\binits{Z.}} \AND
  \bauthor{\bsnm{Wu},~\bfnm{Yihong}\binits{Y.}}
(\byear{2015}).
\btitle{Optimal estimation and rank detection for sparse spiked covariance
  matrices}.
\bjournal{Probab. Theory Related Fields}
\bvolume{161}
\bpages{781--815}.
\end{barticle}
\endbibitem

\bibitem[\protect\citeauthoryear{Cai and Yuan}{2012}]{cai2012adaptive}
\begin{barticle}[author]
\bauthor{\bsnm{Cai},~\bfnm{T~Tony}\binits{T.~T.}} \AND
  \bauthor{\bsnm{Yuan},~\bfnm{Ming}\binits{M.}}
(\byear{2012}).
\btitle{Adaptive covariance matrix estimation through block thresholding}.
\bjournal{Ann. Statist.}
\bvolume{40}
\bpages{2014--2042}.
\end{barticle}
\endbibitem

\bibitem[\protect\citeauthoryear{Cai, Zhang and Zhou}{2010}]{cai2010optimal}
\begin{barticle}[author]
\bauthor{\bsnm{Cai},~\bfnm{T~Tony}\binits{T.~T.}},
  \bauthor{\bsnm{Zhang},~\bfnm{Cun-Hui}\binits{C.-H.}} \AND
  \bauthor{\bsnm{Zhou},~\bfnm{Harrison~H}\binits{H.~H.}}
(\byear{2010}).
\btitle{Optimal rates of convergence for covariance matrix estimation}.
\bjournal{Ann. Statist.}
\bvolume{38}
\bpages{2118--2144}.
\end{barticle}
\endbibitem

\bibitem[\protect\citeauthoryear{Cai and Zhou}{2012a}]{cai2012minimax}
\begin{barticle}[author]
\bauthor{\bsnm{Cai},~\bfnm{T~Tony}\binits{T.~T.}} \AND
  \bauthor{\bsnm{Zhou},~\bfnm{Harrison~H}\binits{H.~H.}}
(\byear{2012}a).
\btitle{Minimax estimation of large covariance matrices under $\ell_1$-norm}.
\bjournal{Statist. Sinica}
\bvolume{22}
\bpages{1319--1349}.
\end{barticle}
\endbibitem

\bibitem[\protect\citeauthoryear{Cai and Zhou}{2012b}]{cai2012optimal}
\begin{barticle}[author]
\bauthor{\bsnm{Cai},~\bfnm{T~Tony}\binits{T.~T.}} \AND
  \bauthor{\bsnm{Zhou},~\bfnm{Harrison~H}\binits{H.~H.}}
(\byear{2012}b).
\btitle{Optimal rates of convergence for sparse covariance matrix estimation}.
\bjournal{Ann. Statist.}
\bvolume{40}
\bpages{2389--2420}.
\end{barticle}
\endbibitem

\bibitem[\protect\citeauthoryear{Cao, Khare and Ghosh}{2017}]{cao2016posterior}
\begin{barticle}[author]
\bauthor{\bsnm{Cao},~\bfnm{Xuan}\binits{X.}},
  \bauthor{\bsnm{Khare},~\bfnm{Kshitij}\binits{K.}} \AND
  \bauthor{\bsnm{Ghosh},~\bfnm{Malay}\binits{M.}}
(\byear{2017}).
\btitle{Posterior Graph Selection and Estimation Consistency for
  High-dimensional Bayesian DAG Models}.
\bjournal{The Annals of Statistics}.
\bnote{Accepted}.
\end{barticle}
\endbibitem

\bibitem[\protect\citeauthoryear{Castillo, Schmidt-Hieber and van~der
  Vaart}{2015}]{castillo2015bayesian}
\begin{barticle}[author]
\bauthor{\bsnm{Castillo},~\bfnm{Isma{\"e}l}\binits{I.}},
  \bauthor{\bsnm{Schmidt-Hieber},~\bfnm{Johannes}\binits{J.}} \AND
  \bauthor{\bparticle{van~der} \bsnm{Vaart},~\bfnm{Aad}\binits{A.}}
(\byear{2015}).
\btitle{Bayesian linear regression with sparse priors}.
\bjournal{Ann. Statist.}
\bvolume{43}
\bpages{1986--2018}.
\end{barticle}
\endbibitem

\bibitem[\protect\citeauthoryear{Duchi}{2016}]{duchi2016lecture}
\begin{bmisc}[author]
\bauthor{\bsnm{Duchi},~\bfnm{John}\binits{J.}}
(\byear{2016}).
\btitle{Lecture Notes for Statistics 311/Electrical Engineering 377}.
\bnote{URL: \url{https://stanford.edu/class/stats311/Lectures/full_notes.pdf}.
  Last visited on 2016/02/23}.
\end{bmisc}
\endbibitem

\bibitem[\protect\citeauthoryear{Eldar and
  Kutyniok}{2012}]{eldar2012compressed}
\begin{bbook}[author]
\bauthor{\bsnm{Eldar},~\bfnm{Yonina~C}\binits{Y.~C.}} \AND
  \bauthor{\bsnm{Kutyniok},~\bfnm{Gitta}\binits{G.}}
(\byear{2012}).
\btitle{Compressed sensing: theory and applications}.
\bpublisher{Cambridge University Press}.
\end{bbook}
\endbibitem

\bibitem[\protect\citeauthoryear{Fan, Fan and Lv}{2008}]{fan2008high}
\begin{barticle}[author]
\bauthor{\bsnm{Fan},~\bfnm{Jianqing}\binits{J.}},
  \bauthor{\bsnm{Fan},~\bfnm{Yingying}\binits{Y.}} \AND
  \bauthor{\bsnm{Lv},~\bfnm{Jinchi}\binits{J.}}
(\byear{2008}).
\btitle{High dimensional covariance matrix estimation using a factor model}.
\bjournal{J. Econometrics}
\bvolume{147}
\bpages{186--197}.
\end{barticle}
\endbibitem

\bibitem[\protect\citeauthoryear{Gao and Zhou}{2015}]{gao2015rate}
\begin{barticle}[author]
\bauthor{\bsnm{Gao},~\bfnm{Chao}\binits{C.}} \AND
  \bauthor{\bsnm{Zhou},~\bfnm{Harrison~H}\binits{H.~H.}}
(\byear{2015}).
\btitle{Rate-optimal posterior contraction for sparse PCA}.
\bjournal{Ann. Statist.}
\bvolume{43}
\bpages{785--818}.
\end{barticle}
\endbibitem

\bibitem[\protect\citeauthoryear{Gr{\"u}nwald
  et~al.}{2017}]{grunwald2017inconsistency}
\begin{barticle}[author]
\bauthor{\bsnm{Gr{\"u}nwald},~\bfnm{Peter}\binits{P.}},
  \bauthor{\bparticle{van} \bsnm{Ommen},~\bfnm{Thijs}\binits{T.}}
  \betal{et~al.}
(\byear{2017}).
\btitle{Inconsistency of Bayesian inference for misspecified linear models, and
  a proposal for repairing it}.
\bjournal{Bayesian Analysis}
\bvolume{12}
\bpages{1069--1103}.
\end{barticle}
\endbibitem

\bibitem[\protect\citeauthoryear{Hero et~al.}{2001}]{hero2001alpha}
\begin{barticle}[author]
\bauthor{\bsnm{Hero},~\bfnm{Alfred~O}\binits{A.~O.}},
  \bauthor{\bsnm{Ma},~\bfnm{Bing}\binits{B.}},
  \bauthor{\bsnm{Michel},~\bfnm{Olivier}\binits{O.}} \AND
  \bauthor{\bsnm{Gorman},~\bfnm{John}\binits{J.}}
(\byear{2001}).
\btitle{Alpha-divergence for classification, indexing and retrieval}.
\bjournal{Communication and Signal Processing Laboratory, Technical Report
  CSPL-328, U. Mich}.
\end{barticle}
\endbibitem

\bibitem[\protect\citeauthoryear{Huang et~al.}{2006}]{huang2006covariance}
\begin{barticle}[author]
\bauthor{\bsnm{Huang},~\bfnm{Jianhua~Z}\binits{J.~Z.}},
  \bauthor{\bsnm{Liu},~\bfnm{Naiping}\binits{N.}},
  \bauthor{\bsnm{Pourahmadi},~\bfnm{Mohsen}\binits{M.}} \AND
  \bauthor{\bsnm{Liu},~\bfnm{Linxu}\binits{L.}}
(\byear{2006}).
\btitle{Covariance matrix selection and estimation via penalised normal
  likelihood}.
\bjournal{Biometrika}
\bvolume{93}
\bpages{85--98}.
\end{barticle}
\endbibitem

\bibitem[\protect\citeauthoryear{Johnstone and
  Lu}{2009}]{johnstone2009consistency}
\begin{barticle}[author]
\bauthor{\bsnm{Johnstone},~\bfnm{Iain~M}\binits{I.~M.}} \AND
  \bauthor{\bsnm{Lu},~\bfnm{Arthur~Yu}\binits{A.~Y.}}
(\byear{2009}).
\btitle{On consistency and sparsity for principal components analysis in high
  dimensions}.
\bjournal{J. Amer. Statist. Assoc.}
\bvolume{104}
\bpages{682--693}.
\end{barticle}
\endbibitem

\bibitem[\protect\citeauthoryear{Kalisch and
  B{\"u}hlmann}{2007}]{kalisch2007estimating}
\begin{barticle}[author]
\bauthor{\bsnm{Kalisch},~\bfnm{Markus}\binits{M.}} \AND
  \bauthor{\bsnm{B{\"u}hlmann},~\bfnm{Peter}\binits{P.}}
(\byear{2007}).
\btitle{Estimating high-dimensional directed acyclic graphs with the
  PC-algorithm}.
\bjournal{Journal of Machine Learning Research}
\bvolume{8}
\bpages{613--636}.
\end{barticle}
\endbibitem

\bibitem[\protect\citeauthoryear{Khare et~al.}{2016}]{khare2016convex}
\begin{barticle}[author]
\bauthor{\bsnm{Khare},~\bfnm{Kshitij}\binits{K.}},
  \bauthor{\bsnm{Oh},~\bfnm{Sang}\binits{S.}},
  \bauthor{\bsnm{Rahman},~\bfnm{Syed}\binits{S.}} \AND
  \bauthor{\bsnm{Rajaratnam},~\bfnm{Bala}\binits{B.}}
(\byear{2016}).
\btitle{A convex framework for high-dimensional sparse Cholesky based
  covariance estimation}.
\bjournal{arXiv preprint arXiv:1610.02436}.
\end{barticle}
\endbibitem

\bibitem[\protect\citeauthoryear{Laurent and
  Massart}{2000}]{laurent2000adaptive}
\begin{barticle}[author]
\bauthor{\bsnm{Laurent},~\bfnm{Beatrice}\binits{B.}} \AND
  \bauthor{\bsnm{Massart},~\bfnm{Pascal}\binits{P.}}
(\byear{2000}).
\btitle{Adaptive estimation of a quadratic functional by model selection}.
\bjournal{Ann. Statist.}
\bvolume{28}
\bpages{1302--1338}.
\end{barticle}
\endbibitem

\bibitem[\protect\citeauthoryear{Lee and Lee}{2017}]{lee2017estimating}
\begin{barticle}[author]
\bauthor{\bsnm{Lee},~\bfnm{Kyoungjae}\binits{K.}} \AND
  \bauthor{\bsnm{Lee},~\bfnm{Jaeyong}\binits{J.}}
(\byear{2017}).
\btitle{{Estimating Large Precision Matrices via Modified Cholesky
  Decomposition}}.
\bjournal{arXiv:1707.01143}.
\end{barticle}
\endbibitem

\bibitem[\protect\citeauthoryear{Lee and Lee}{2018}]{lee2017optimal}
\begin{barticle}[author]
\bauthor{\bsnm{Lee},~\bfnm{Kyoungjae}\binits{K.}} \AND
  \bauthor{\bsnm{Lee},~\bfnm{Jaeyong}\binits{J.}}
(\byear{2018}).
\btitle{{Optimal Bayesian Minimax Rates for Unconstrained Large Covariance
  Matrices}}.
\bjournal{Bayesian Anal.}
\bnote{Accepted}.
\end{barticle}
\endbibitem

\bibitem[\protect\citeauthoryear{Lee, Lee and Lin}{2018}]{lee2018supp}
\begin{barticle}[author]
\bauthor{\bsnm{Lee},~\bfnm{Kyoungjae}\binits{K.}},
  \bauthor{\bsnm{Lee},~\bfnm{Jaeyong}\binits{J.}} \AND
  \bauthor{\bsnm{Lin},~\bfnm{Lizhen}\binits{L.}}
(\byear{2018}).
\btitle{Supplementary material for ``Minimax Posterior Convergence Rates and
  Model Selection Consistency in High-dimensional DAG Models based on Sparse
  Cholesky Factors''}.
\end{barticle}
\endbibitem

\bibitem[\protect\citeauthoryear{Liang et~al.}{2008}]{liang2008mixtures}
\begin{barticle}[author]
\bauthor{\bsnm{Liang},~\bfnm{Feng}\binits{F.}},
  \bauthor{\bsnm{Paulo},~\bfnm{Rui}\binits{R.}},
  \bauthor{\bsnm{Molina},~\bfnm{German}\binits{G.}},
  \bauthor{\bsnm{Clyde},~\bfnm{Merlise~A}\binits{M.~A.}} \AND
  \bauthor{\bsnm{Berger},~\bfnm{Jim~O}\binits{J.~O.}}
(\byear{2008}).
\btitle{Mixtures of g priors for Bayesian variable selection}.
\bjournal{Journal of the American Statistical Association}
\bvolume{103}
\bpages{410--423}.
\end{barticle}
\endbibitem

\bibitem[\protect\citeauthoryear{Martin, Mess and
  Walker}{2017}]{martin2017empirical}
\begin{barticle}[author]
\bauthor{\bsnm{Martin},~\bfnm{Ryan}\binits{R.}},
  \bauthor{\bsnm{Mess},~\bfnm{Raymond}\binits{R.}} \AND
  \bauthor{\bsnm{Walker},~\bfnm{Stephen~G}\binits{S.~G.}}
(\byear{2017}).
\btitle{Empirical Bayes posterior concentration in sparse high-dimensional
  linear models}.
\bjournal{Bernoulli}
\bvolume{23}
\bpages{1822--1847}.
\end{barticle}
\endbibitem

\bibitem[\protect\citeauthoryear{Martin and
  Walker}{2014}]{martin2014asymptotically}
\begin{barticle}[author]
\bauthor{\bsnm{Martin},~\bfnm{Ryan}\binits{R.}} \AND
  \bauthor{\bsnm{Walker},~\bfnm{Stephen~G}\binits{S.~G.}}
(\byear{2014}).
\btitle{Asymptotically minimax empirical Bayes estimation of a sparse normal
  mean vector}.
\bjournal{Electron. J. Stat.}
\bvolume{8}
\bpages{2188--2206}.
\end{barticle}
\endbibitem

\bibitem[\protect\citeauthoryear{Miller and Dunson}{2018}]{miller2018robust}
\begin{barticle}[author]
\bauthor{\bsnm{Miller},~\bfnm{Jeffrey~W}\binits{J.~W.}} \AND
  \bauthor{\bsnm{Dunson},~\bfnm{David~B}\binits{D.~B.}}
(\byear{2018}).
\btitle{Robust Bayesian inference via coarsening}.
\bjournal{Journal of the American Statistical Association}
\bvolume{just-accepted}
\bpages{1--31}.
\end{barticle}
\endbibitem

\bibitem[\protect\citeauthoryear{Narisetty and
  He}{2014}]{narisetty2014bayesian}
\begin{barticle}[author]
\bauthor{\bsnm{Narisetty},~\bfnm{Naveen~Naidu}\binits{N.~N.}} \AND
  \bauthor{\bsnm{He},~\bfnm{Xuming}\binits{X.}}
(\byear{2014}).
\btitle{Bayesian variable selection with shrinking and diffusing priors}.
\bjournal{The Annals of Statistics}
\bvolume{42}
\bpages{789--817}.
\end{barticle}
\endbibitem

\bibitem[\protect\citeauthoryear{Pati et~al.}{2014}]{pati2014posterior}
\begin{barticle}[author]
\bauthor{\bsnm{Pati},~\bfnm{Debdeep}\binits{D.}},
  \bauthor{\bsnm{Bhattacharya},~\bfnm{Anirban}\binits{A.}},
  \bauthor{\bsnm{Pillai},~\bfnm{Natesh~S}\binits{N.~S.}} \AND
  \bauthor{\bsnm{Dunson},~\bfnm{David}\binits{D.}}
(\byear{2014}).
\btitle{Posterior contraction in sparse Bayesian factor models for massive
  covariance matrices}.
\bjournal{Ann. Statist.}
\bvolume{42}
\bpages{1102--1130}.
\end{barticle}
\endbibitem

\bibitem[\protect\citeauthoryear{Reid, Tibshirani and
  Friedman}{2016}]{reid2016study}
\begin{barticle}[author]
\bauthor{\bsnm{Reid},~\bfnm{Stephen}\binits{S.}},
  \bauthor{\bsnm{Tibshirani},~\bfnm{Robert}\binits{R.}} \AND
  \bauthor{\bsnm{Friedman},~\bfnm{Jerome}\binits{J.}}
(\byear{2016}).
\btitle{A study of error variance estimation in Lasso regression}.
\bjournal{Statist. Sinica}
\bvolume{26}
\bpages{35--67}.
\end{barticle}
\endbibitem

\bibitem[\protect\citeauthoryear{Ren et~al.}{2015}]{ren2015asymptotic}
\begin{barticle}[author]
\bauthor{\bsnm{Ren},~\bfnm{Zhao}\binits{Z.}},
  \bauthor{\bsnm{Sun},~\bfnm{Tingni}\binits{T.}},
  \bauthor{\bsnm{Zhang},~\bfnm{Cun-Hui}\binits{C.-H.}} \AND
  \bauthor{\bsnm{Zhou},~\bfnm{Harrison~H}\binits{H.~H.}}
(\byear{2015}).
\btitle{Asymptotic normality and optimalities in estimation of large Gaussian
  graphical models}.
\bjournal{Ann. Statist.}
\bvolume{43}
\bpages{991--1026}.
\end{barticle}
\endbibitem

\bibitem[\protect\citeauthoryear{Rothman, Levina and
  Zhu}{2010}]{rothman2010new}
\begin{barticle}[author]
\bauthor{\bsnm{Rothman},~\bfnm{Adam~J}\binits{A.~J.}},
  \bauthor{\bsnm{Levina},~\bfnm{Elizaveta}\binits{E.}} \AND
  \bauthor{\bsnm{Zhu},~\bfnm{Ji}\binits{J.}}
(\byear{2010}).
\btitle{A new approach to Cholesky-based covariance regularization in high
  dimensions}.
\bjournal{Biometrika}
\bvolume{97}
\bpages{539--550}.
\end{barticle}
\endbibitem

\bibitem[\protect\citeauthoryear{Roverato}{2000}]{roverato2000cholesky}
\begin{barticle}[author]
\bauthor{\bsnm{Roverato},~\bfnm{Alberto}\binits{A.}}
(\byear{2000}).
\btitle{Cholesky decomposition of a hyper inverse Wishart matrix}.
\bjournal{Biometrika}
\bvolume{87}
\bpages{99--112}.
\end{barticle}
\endbibitem

\bibitem[\protect\citeauthoryear{R{\"u}timann and
  B{\"u}hlmann}{2009}]{rutimann2009high}
\begin{barticle}[author]
\bauthor{\bsnm{R{\"u}timann},~\bfnm{Philipp}\binits{P.}} \AND
  \bauthor{\bsnm{B{\"u}hlmann},~\bfnm{Peter}\binits{P.}}
(\byear{2009}).
\btitle{High dimensional sparse covariance estimation via directed acyclic
  graphs}.
\bjournal{Electronic Journal of Statistics}
\bvolume{3}
\bpages{1133--1160}.
\end{barticle}
\endbibitem

\bibitem[\protect\citeauthoryear{Shang and
  Clayton}{2011}]{shang2011consistency}
\begin{barticle}[author]
\bauthor{\bsnm{Shang},~\bfnm{Zuofeng}\binits{Z.}} \AND
  \bauthor{\bsnm{Clayton},~\bfnm{Murray~K}\binits{M.~K.}}
(\byear{2011}).
\btitle{Consistency of Bayesian linear model selection with a growing number of
  parameters}.
\bjournal{Journal of Statistical Planning and Inference}
\bvolume{141}
\bpages{3463--3474}.
\end{barticle}
\endbibitem

\bibitem[\protect\citeauthoryear{Shin, Bhattacharya and
  Johnson}{2015}]{shin2015scalable}
\begin{barticle}[author]
\bauthor{\bsnm{Shin},~\bfnm{Minsuk}\binits{M.}},
  \bauthor{\bsnm{Bhattacharya},~\bfnm{Anirban}\binits{A.}} \AND
  \bauthor{\bsnm{Johnson},~\bfnm{Valen~E}\binits{V.~E.}}
(\byear{2015}).
\btitle{Scalable Bayesian variable selection using nonlocal prior densities in
  ultrahigh-dimensional settings}.
\bjournal{arXiv:1507.07106}.
\end{barticle}
\endbibitem

\bibitem[\protect\citeauthoryear{Shojaie and
  Michailidis}{2010}]{shojaie2010penalized}
\begin{barticle}[author]
\bauthor{\bsnm{Shojaie},~\bfnm{Ali}\binits{A.}} \AND
  \bauthor{\bsnm{Michailidis},~\bfnm{George}\binits{G.}}
(\byear{2010}).
\btitle{Penalized likelihood methods for estimation of sparse high-dimensional
  directed acyclic graphs}.
\bjournal{Biometrika}
\bvolume{97}
\bpages{519--538}.
\end{barticle}
\endbibitem

\bibitem[\protect\citeauthoryear{Syring and Martin}{2016}]{syring2015scaling}
\begin{barticle}[author]
\bauthor{\bsnm{Syring},~\bfnm{Nick}\binits{N.}} \AND
  \bauthor{\bsnm{Martin},~\bfnm{Ryan}\binits{R.}}
(\byear{2016}).
\btitle{Scaling the Gibbs posterior credible regions}.
\bjournal{arXiv preprint arXiv:1509.00922}.
\end{barticle}
\endbibitem

\bibitem[\protect\citeauthoryear{van~de Geer and
  B{\"u}hlmann}{2013}]{van2013ell}
\begin{barticle}[author]
\bauthor{\bparticle{van~de} \bsnm{Geer},~\bfnm{Sara}\binits{S.}} \AND
  \bauthor{\bsnm{B{\"u}hlmann},~\bfnm{Peter}\binits{P.}}
(\byear{2013}).
\btitle{$\ell_0$-penalized maximum likelihood for sparse directed acyclic
  graphs}.
\bjournal{Ann. Statist.}
\bvolume{41}
\bpages{536--567}.
\end{barticle}
\endbibitem

\bibitem[\protect\citeauthoryear{Wainwright}{2009a}]{wainwright2009information}
\begin{barticle}[author]
\bauthor{\bsnm{Wainwright},~\bfnm{Martin~J}\binits{M.~J.}}
(\byear{2009}a).
\btitle{Information-theoretic limits on sparsity recovery in the
  high-dimensional and noisy setting}.
\bjournal{IEEE Trans. Inform. Theory}
\bvolume{55}
\bpages{5728--5741}.
\end{barticle}
\endbibitem

\bibitem[\protect\citeauthoryear{Wainwright}{2009b}]{wainwright2009sharp}
\begin{barticle}[author]
\bauthor{\bsnm{Wainwright},~\bfnm{Martin~J}\binits{M.~J.}}
(\byear{2009}b).
\btitle{Sharp thresholds for High-Dimensional and noisy sparsity recovery using
  $\ell_1$-Constrained Quadratic Programming (Lasso)}.
\bjournal{IEEE Trans. Inform. Theory}
\bvolume{55}
\bpages{2183--2202}.
\end{barticle}
\endbibitem

\bibitem[\protect\citeauthoryear{Walker and Hjort}{2001}]{walker2001bayesian}
\begin{barticle}[author]
\bauthor{\bsnm{Walker},~\bfnm{Stephen}\binits{S.}} \AND
  \bauthor{\bsnm{Hjort},~\bfnm{Nils~Lid}\binits{N.~L.}}
(\byear{2001}).
\btitle{On Bayesian consistency}.
\bjournal{J. R. Stat. Soc. Ser. B. Stat. Methodol.}
\bvolume{63}
\bpages{811--821}.
\end{barticle}
\endbibitem

\bibitem[\protect\citeauthoryear{Xiang, Khare and Ghosh}{2015}]{xiang2015high}
\begin{barticle}[author]
\bauthor{\bsnm{Xiang},~\bfnm{Ruoxuan}\binits{R.}},
  \bauthor{\bsnm{Khare},~\bfnm{Kshitij}\binits{K.}} \AND
  \bauthor{\bsnm{Ghosh},~\bfnm{Malay}\binits{M.}}
(\byear{2015}).
\btitle{High dimensional posterior convergence rates for decomposable graphical
  models}.
\bjournal{Electron. J. Stat.}
\bvolume{9}
\bpages{2828--2854}.
\end{barticle}
\endbibitem

\bibitem[\protect\citeauthoryear{Yang, Wainwright and
  Jordan}{2016}]{yang2016computational}
\begin{barticle}[author]
\bauthor{\bsnm{Yang},~\bfnm{Yun}\binits{Y.}},
  \bauthor{\bsnm{Wainwright},~\bfnm{Martin~J}\binits{M.~J.}} \AND
  \bauthor{\bsnm{Jordan},~\bfnm{Michael~I}\binits{M.~I.}}
(\byear{2016}).
\btitle{On the computational complexity of high-dimensional Bayesian variable
  selection}.
\bjournal{Ann. Statist.}
\bvolume{44}
\bpages{2497--2532}.
\end{barticle}
\endbibitem

\bibitem[\protect\citeauthoryear{Ye and Zhang}{2010}]{ye2010rate}
\begin{barticle}[author]
\bauthor{\bsnm{Ye},~\bfnm{Fei}\binits{F.}} \AND
  \bauthor{\bsnm{Zhang},~\bfnm{Cun-Hui}\binits{C.-H.}}
(\byear{2010}).
\btitle{Rate Minimaxity of the Lasso and Dantzig Selector for the $\ell_q$ Loss
  in $\ell_r$ Balls}.
\bjournal{J. Mach. Learn. Res.}
\bvolume{11}
\bpages{3519--3540}.
\end{barticle}
\endbibitem

\bibitem[\protect\citeauthoryear{Yu and Bien}{2017}]{yu2017learning}
\begin{barticle}[author]
\bauthor{\bsnm{Yu},~\bfnm{Guo}\binits{G.}} \AND
  \bauthor{\bsnm{Bien},~\bfnm{Jacob}\binits{J.}}
(\byear{2017}).
\btitle{Learning Local Dependence In Ordered Data}.
\bjournal{J. Mach. Learn. Res.}
\bvolume{18}
\bpages{1--60}.
\end{barticle}
\endbibitem

\bibitem[\protect\citeauthoryear{Zellner}{1986}]{zellner1986assessing}
\begin{barticle}[author]
\bauthor{\bsnm{Zellner},~\bfnm{Arnold}\binits{A.}}
(\byear{1986}).
\btitle{On assessing prior distributions and Bayesian regression analysis with
  g-prior distributions}.
\bjournal{Bayesian inference and decision techniques: Essays in Honor of Bruno
  De Finetti}
\bvolume{6}
\bpages{233--243}.
\end{barticle}
\endbibitem

\end{thebibliography}

% AOS,AOAS: If there are supplements please fill:

	\begin{frontmatter}
		
		% "Title of the paper"
		\title{Supplementary to ``Minimax Posterior Convergence Rates and Model Selection Consistency in High-dimensional DAG Models based on Sparse Cholesky Factors''}
		\runtitle{Convergence Rate and Selection Consistency for DAG Model}
		%\thankstext{T1}{Footnote to the title with the ``thankstext'' command.}
		
		% indicate corresponding author with \corref{}
		\author{\fnms{Kyoungjae} \snm{Lee}\corref{}\thanksref{m1}\ead[label=e1]{klee25@nd.edu}}
		%\author{\fnms{Kyoungjae} \snm{Lee}\corref{}\thanksref{m1}\ead[label=e1]{klee25@nd.edu}\thanksref{t1}}
		\and
		\author{\fnms{Jaeyong} \snm{Lee}\thanksref{m2}\ead[label=e2]{leejyc@gmail.com}}
		\and
		\author{\fnms{Lizhen} \snm{Lin}\thanksref{m1}\ead[label=e3]{lizhen.lin@nd.edu}}
		\address{Department of Applied and Computational\\ 
			Mathematics and Statistics\\
			The University of Notre Dame\\
			Notre Dame, Indiana 46556\\
			USA \\
			\printead{e1}}
		\address{Department of Statistics \\
			Seoul National University\\
			1 Gwanak-ro, Gwanak-gu\\
			Seoul 08826\\
			South Korea\\
			\printead{e2}}
		\address{Department of Applied and Computational\\ 
			Mathematics and Statistics\\
			The University of Notre Dame\\
			Notre Dame, Indiana 46556\\
			USA \\
			\printead{e3}}
		\affiliation{The University of Notre Dame\thanksmark{m1} and Seoul National University\thanksmark{m2}}
		
		\runauthor{Kyoungjae Lee, Jaeyong Lee and Lizhen Lin}
		
		\begin{abstract}
			In this supplement, we present the remaining proofs for posterior convergence rates and other auxiliary results.
		\end{abstract}
		
	\end{frontmatter}
	
	%%%%%%%%%%%%%%%%%%%%%%%%%%%%%%%%%%%%%%%%%%%%%
	
	%%%%%%%%%%%%%%%%%%%%%%%%%%%%%%%%%%%%%%%%%%%%%
	\section{Proof of Theorem 3.1}\label{sec:proofsec}
	
	\begin{proof}
		Note that
		\bean
		\pi_\alpha(S_{A_n} \neq S_{A_{0n}} \mid \bfX_n)
		&=& \sum_{j=2}^p \pi_\alpha(S_{j} \neq S_{0j} \mid \bfX_n)\nonumber \\
		&=& \sum_{j=2}^p \Big\{\pi_\alpha(S_j \supsetneq S_{0j} \mid \bfX_n ) + \pi_\alpha(S_j \nsupseteq S_{0j} \mid \bfX_n ) \Big\}.\label{Sneq}
		\eean
		The first term of \eqref{Sneq} is of order $o(1)$ by Lemma 8.1.
		We only need to consider the second term of \eqref{Sneq}. 
		Note that
		\bea
		- \frac{\alpha n +\nu_0}{2} \log \left(\frac{\what{d}_{S_j}}{\what{d}_{S_{0j}}} \right)
		&=& - \frac{\alpha n +\nu_0}{2} \log \left[1 - \frac{\what{d}_{S_{0j}} - \what{d}_{S_{j}}}{ \what{d}_{S_{0j}} }  \right] \\
		&\le& \frac{\alpha n +\nu_0}{2} \cdot \frac{\what{d}_{S_{0j}} - \what{d}_{S_{j}}}{ \what{d}_{S_{0j}} } \left( 1 - \frac{\what{d}_{S_{0j}} - \what{d}_{S_{j}}}{ \what{d}_{S_{0j}} } \right)^{-1} \\
		&\equiv& \frac{\alpha  + \frac{\nu_0}{n}}{2d_{0j}} \cdot n\left(\what{d}_{S_{0j}} - \what{d}_{S_{j}} \right) (1 + \what{Q}_n)^{-1},
		\eea
		where 
		\bea
		\what{Q}_n &:=& \left( \frac{\what{d}_{S_{0j}}}{d_{0j}} -1 \right) - \left( \frac{\what{d}_{S_{0j}} - \what{d}_{S_{0j}\cup S_{j}} }{d_{0j}} \right) + \left( \frac{\what{d}_{S_{j}} - \what{d}_{S_{0j}\cup S_{j}} }{d_{0j}} \right). 
		\eea 
		The inequality holds because $-\log(1-x) \le x/(1-x)$ for any $x<1$.
		For a given $0<\alpha<1$, we define the events
		\bea
		N_{1,S_j,\alpha, \chi^2}^c &:=& \left\{ \bfX_n :  \Big| \frac{\what{d}_{S_{0j}}}{d_{0j}} -1 \Big|    \in \Big(-4 \sqrt{\epsilon'} \frac{n - |S_{0j}|}{n} - \frac{|S_{0j}|}{n},  4 \sqrt{\epsilon'} \frac{n - |S_{0j}|}{n} - \frac{|S_{0j}|}{n}  \Big)  \right\} , 
		\eea
		\bea
		N_{2,S_j,\alpha, \chi^2}^c &:=& \left\{ \bfX_n : 0< \frac{\what{d}_{S_{0j}} - \what{d}_{S_{0j}\cup S_j}}{d_{0j}} < 4 \epsilon' + \frac{|S_{0j}\cup S_j| - |S_{0j}| }{n}  \right\} , \\
		N_{3,S_j,\alpha, \chi^2}^c &:=& \left\{ \bfX_n : 0<  \frac{\what{d}_{S_{j}} - \what{d}_{S_{0j}\cup S_j}}{d_{0j}} < \epsilon' + \frac{\what{\lambda}_n}{n}  \right\} ,
		\eea
		where $\epsilon' := ((1-\alpha)/10)^2$ and $\what{\lambda}_n := \|(I_n - \tilde{P}_{S_j}) \tilde{Z}_j a_{0j}\|_2^2 /d_{0j}$.
		Let $N_{S_j,\alpha,\chi^2}^c := \cap_{k=1}^3 N_{k,S_j,\alpha, \chi^2}^c $ and $\nu_{1} :=(1+\alpha/\gamma)^{1/2}$, then
		\begin{align}
		& \sum_{j=2}^p \, \bbE_0 \left[ \pi_\alpha(S_j \nsupseteq S_{0j} \mid \bfX_n ) \right] \nonumber \\
		&\le \sum_{j=2}^p \sum_{S_j: S_j \nsupseteq S_{0j} } \left\{ \bbP_0 (N_{1,S_j,\alpha, \chi^2}) + \bbP_0 (N_{2,S_j,\alpha, \chi^2}) + \bbP_0 (N_{3,S_j,\alpha, \chi^2})   \right\} \label{N123part} \\
		&+ \sum_{j=2}^p \sum_{S_j: S_j \nsupseteq S_{0j} }  \bbE_0 \left[ \frac{\pi_j(S_j )}{\pi_j(S_{0j} )} \nu_1^{|S_{0j}| - |S_{j}|}  \exp \left( \frac{\alpha  + \frac{\nu_0}{n}}{2d_{0j}(1 + \what{Q}_n)} \cdot n\left(\what{d}_{S_{0j}} - \what{d}_{S_{j}} \right) \right) I_{N_{S_j,\alpha,\chi^2}^c}   \right]. \label{N123c}
		\end{align}
		If we show that \eqref{N123part} and \eqref{N123c} are of order $o(1)$, the proof is completed.
		Since \eqref{N123part} is of order $o(1)$ by Lemma 8.2, we will focus on \eqref{N123c} part.
		Note that on the event $\bfX_n \in  N_{S_j,\alpha,\chi^2}^c$, 
		\bea
		\min(\what{Q}_n) &:=& -4\sqrt{\epsilon'} - \frac{|S_{0j}|}{n} - 4\epsilon' - \frac{|S_{0j}\cup S_j| - |S_{0j}|}{n} \\
		&\le& \what{Q}_n \,\,\le\,\, 5\sqrt{\epsilon'} + \frac{\what{\lambda}_n}{n} \,\,=:\,\, \max(\what{Q}_n)
		\eea
		for all sufficiently large $n$.
		Also note that, for a fixed $S_j \nsupseteq S_{0j}$ and given $\tilde{Z}_j$, 
		\bea
		n (\what{d}_{S_{0j}} - \what{d}_{S_{j}} )
		&=&   \tilde{X}_j^T ( \tilde{P}_{S_j}- \tilde{P}_{S_{0j}}) \tilde{X}_j \\
		&\overset{d}{\equiv}&   (a_{0j}^T \tilde{Z}_j^T + \tilde{\epsilon}_j^T ) ( \tilde{P}_{S_j}- \tilde{P}_{S_{0j}}) (\tilde{Z}_j a_{0j} + \tilde{\epsilon}_j ) \\
		&=& -  \| (I_n - \tilde{P}_{S_j}) \tilde{Z}_j a_{0j} \|_2^2 - 2 \tilde{\epsilon}_j^T (I_n - \tilde{P}_{S_j}) \tilde{Z}_j a_{0j} +  \tilde{\epsilon}_j^T ( \tilde{P}_{S_j}- \tilde{P}_{S_{0j}}) \tilde{\epsilon}_j 	\\
		&\le& - d_{0j}\what{\lambda}_n - 2 \tilde{\epsilon}_j^T (I_n - \tilde{P}_{S_j}) \tilde{Z}_j a_{0j} \oplus  \tilde{\epsilon}_j^T ( \tilde{P}_{S_j}- \tilde{P}_{S_j \cap S_{0j}}) \tilde{\epsilon}_j \\
		&=:& - d_{0j}\what{\lambda}_n - 2 V_{j,S_j} \oplus W_{j,S_j} ,
		\eea
		where $\tilde{\epsilon}_j \sim N_n(0, d_{0j}I_n)$.
		For a given $\tilde{Z}_j$, it is easy to show that $V_{j,S_j}/\sqrt{d_{0j}} \sim N(0, d_{0j}\what{\lambda}_n )$ and $W_{j,S_j}/d_{0j} \sim \chi^2_{|S_j|-|S_{0j}\cap S_j|}$ under $\bbP_0$.
		Then,
		\bea
		&& \bbE_0 \left[ \exp \left( \frac{\alpha  + \frac{\nu_0}{n}}{2d_{0j}(1 + \what{Q}_n)} \cdot n\left(\what{d}_{S_{0j}} - \what{d}_{S_{j}} \right) \right) I_{N_{S_j,\alpha,\chi^2}^c} \,\Big|\, \tilde{Z}_j \right] \\
		&\le& \hspace{-.5cm}\sum_{Q \in \{\min(\what{Q}_n), \max(\what{Q}_n)\} }\hspace{-.5cm} \bbE_0 \Bigg[  \exp \Bigg( -\frac{(\alpha  + \frac{\nu_0}{n})(d_{0j}\what{\lambda}_n + 2 V_{j,S_j}) }{2d_{0j}(1 + Q)} \Bigg)  \exp \Bigg( \frac{(\alpha  + \frac{\nu_0}{n})W_{j,S_j} }{2d_{0j}(1 + Q)} \Bigg) \,\Big|\, \tilde{Z}_j \Bigg] \\
		&\le& \hspace{-.5cm}\sum_{Q \in  \{\min(\what{Q}_n), \max(\what{Q}_n)\} }\hspace{-.5cm} \bbE_0 \Bigg[  \exp \Bigg( -\frac{(\alpha  + \frac{\nu_0}{n})(d_{0j}\what{\lambda}_n + 2 V_{j,S_j}) }{2d_{0j}(1 + Q)} \Bigg)  \,\Big|\, \tilde{Z}_j \Bigg] \left(1 - \frac{\alpha + \frac{\nu_0}{n} }{1+Q} \right)^{-\frac{|S_j|-|S_{0j}\cap S_j|}{2} } 
		%	&&+ \quad \bbE_0 \Bigg[  \exp \Bigg( -\frac{(\alpha  + \frac{\nu_0}{n})(d_{0j}\what{\lambda}_n + 2 V_{j,S_j}) }{2d_{0j}(1 + \max(\what{Q}_n))} \Bigg) \,\Big|\, \tilde{Z}_j \Bigg] \times \left(1 - \frac{\alpha + \frac{\nu_0}{n} }{1+\min(\what{Q}_n)} \right)^{-\frac{|S_j|-|S_{0j}\cap S_j|}{2} }.
		\eea
		The second inequality follows from the moment generating function of the chi-square distribution because $\alpha+ \nu_0/n < 1 + \min(\what{Q}_n)$ for all sufficiently large $n$.
		From the moment generating function of the normal distribution, we have
		\bea
		&&\bbE_0 \Bigg[  \exp \Bigg( -\frac{(\alpha  + \frac{\nu_0}{n})(d_{0j}\what{\lambda}_n + 2 V_{j,S_j}) }{2d_{0j}(1 + Q)} \Bigg) \,\Big|\, \tilde{Z}_j \Bigg] \\
		&=& \exp \left\{ - \frac{\alpha  + \frac{\nu_0}{n} }{2 (1 + Q)} \cdot \bigg(1 - \frac{\alpha  + \frac{\nu_0}{n} }{1 + Q} \bigg) \cdot \what{\lambda}_n   \right\} \\
		&\le& \exp \left\{ - \frac{\alpha  + \frac{\nu_0}{n} }{2 (1 + \max(\what{Q}_n))} \cdot \bigg(1 - \frac{\alpha  + \frac{\nu_0}{n} }{1 + \min(\what{Q}_n)} \bigg) \cdot \what{\lambda}_n   \right\}\\
		&\le& \exp \left\{ - \frac{\alpha  + \frac{\nu_0}{n} }{2 (1 + 5\sqrt{\epsilon'}+ \frac{\what{\lambda}_n}{n})} \cdot \bigg(1 - \frac{\alpha  + \frac{\nu_0}{n} }{1 - 4\sqrt{\epsilon'} - 5\epsilon'} \bigg) \cdot \what{\lambda}_n   \right\},
		\eea
		where $Q= \min(\what{Q}_n)$ or $\max(\what{Q}_n)$.
		Note that 
		\bea
		d_{0j} \what{\lambda}_n &=& \| (I_n - \tilde{P}_{S_j})\tilde{Z}_{S_{0j}\cap S_j^c} a_{0,S_{0j}\cap S_j^c} \|_2^2 \\
		&\ge& \lambda_{\min} (\tilde{Z}_{S_{0j}\cup S_j}^T \tilde{Z}_{S_{0j}\cup S_j}) \| a_{0,S_{0j}\cap S_j^c} \|_2^2 \\
		&\ge& \lambda_{\min} (\tilde{Z}_{S_{0j}\cup S_j}^T \tilde{Z}_{S_{0j}\cup S_j}) \cdot (|S_{0j}| - |S_j\cap S_{0j}| ) \cdot \min_{j,l: a_{0,jl}\neq 0} |a_{0,jl}|^2
		\eea
		by Lemma 5 of \cite{arias2014estimation}.
		Define a set 
		\bea
		N_{j,S_j} &:=& \left\{\bfX_n : n^{-1} \lambda_{\min}(\tilde{Z}_{S_{0j}\cup S_j}^T \tilde{Z}_{S_{0j}\cup S_j} ) \le (1-2\epsilon_0)^2 \epsilon_0 \right\} \\
		&& \cap\,\, \left\{\bfX_n : n^{-1} \lambda_{\max}(\tilde{Z}_{S_{0j}\cup S_j}^T \tilde{Z}_{S_{0j}\cup S_j} ) \ge (1+2\epsilon_0)^2 \epsilon_0^{-1} \right\} .
		\eea
		By Corollary 5.35 in \cite{eldar2012compressed}, we have $\bbP_0 (N_{j,S_j}) \le 4 \exp(- n\epsilon_0^2 /2)$ for all sufficiently large $n$.  
		Thus, on the event $\bfX_n \in N_{j,S_j}^c$, we have $\what{\lambda}_n 
		\ge d_{0j}^{-1} (1-2\epsilon_0)^2\epsilon_0 \cdot (|S_{0j}| - |S_j\cap S_{0j}| ) \cdot n \min_{j,l: a_{0,jl}\neq 0} |a_{0,jl}|^2$,
		which implies
		\bea
		&& \bbE_0 \left[ \exp \left\{ - \frac{\alpha  + \frac{\nu_0}{n} }{2 (1 + 5\sqrt{\epsilon'}+ \frac{\what{\lambda}_n}{n})} \cdot \bigg(1 - \frac{\alpha  + \frac{\nu_0}{n} }{1 - 4\sqrt{\epsilon'} - 5\epsilon'} \bigg) \cdot \what{\lambda}_n   \right\}  \right] \\
		&\le& \bbE_0 \left[ \exp \left\{ - \frac{\alpha  }{2 } \cdot \Big(1 - \frac{\alpha }{1 - 4\sqrt{\epsilon'} - 5\epsilon'} \Big) \cdot \Big( \frac{1 + 5\sqrt{\epsilon'}}{\what{\lambda}_n} + \frac{1}{n}  \Big)^{-1}   \right\}  \right]  \\
		&\le& \bbE_0 \left[ \exp \left\{ - \frac{\alpha  }{2 } \cdot \Big( \frac{1-\alpha}{2} \Big) \cdot \Big( \frac{2 }{\what{\lambda}_n} + \frac{1}{n}  \Big)^{-1}   \right\}  \right]  \\
		&\le& \exp \left\{  - \frac{\alpha(1-\alpha)}{4} \cdot \frac{\epsilon_0^2 (1-2\epsilon_0)^2 }{4}(|S_{0j}| - |S_j\cap S_{0j}| ) \cdot n \min_{j,l: a_{0,jl}\neq 0} |a_{0,jl}|^2 \right\}
		+ \bbP_0(N_{j,S_j}) \\
		&\le& \exp \Big\{ - (|S_{0j}| - |S_j\cap S_{0j}| ) \cdot C_{\rm bm} \log p \Big\}  + \bbP_0(N_{j,S_j}) 
		\eea
		for all sufficiently large $n$.
		Note that the second inequality holds because $\sqrt{\epsilon'} = (1-\alpha)/10$.	
		Thus, \eqref{N123c} is bounded above by
		\begin{align*}
		& \sum_{j=2}^p \sum_{S_j: S_j \nsupseteq S_{0j} } \frac{\pi_j(S_j )}{\pi_j(S_{0j} )} \nu_1^{s_{0j} - |S_{j}|} \nu_2^{|S_j|- |S_{0j}\cap S_j|}  \cdot 2\exp \Big\{ - (s_{0j}  - |S_j\cap S_{0j}| ) C_{\rm bm} \log p \Big\}  \nonumber \\
		&+\quad \sum_{j=2}^p \sum_{S_j: S_j \nsupseteq S_{0j} } \frac{\pi_j(S_j )}{\pi_j(S_{0j} )} \nu_1^{s_{0j} - |S_{j}|} \nu_2^{|S_j|- |S_{0j}\cap S_j|}  \cdot 4 \exp \Big(- \frac{n\epsilon_0^2}{2} \Big) \nonumber  \\
		&\le \sum_{j=2}^p\sum_{s=0}^{R_j}\sum_{t=0}^{(s_{0j}-1)\wedge s}   \binom{s_{0j}}{t} \binom{j-s_{0j}}{s-t}  \frac{\binom{j}{s_{0j}}}{\binom{j}{s}} \Big( \frac{\nu_1}{\nu_2} c_1 p^{c_2} \Big)^{s_{0j}-s}  \, 6\, (\nu_2 p^{-C_{\rm bm}})^{s_{0j}-t}  
		\end{align*}
		for all sufficiently large $n$ and some constant $C_{\rm bm}>0$, where $\nu_2 := (1 - (\alpha+\nu_0/n)/(1-4\sqrt{\epsilon'}-5\epsilon'))^{-1/2}$.
		Note that 
		\bea
		\frac{ \binom{s_{0j}}{t} \binom{j-s_{0j}}{s-t}  \binom{j}{s_{0j}}}{\binom{j}{s}}  &=& \binom{s}{t} \binom{j-s}{s_{0j}-t} \,\,\le\,\, s^{s-t}\times  p^{s_{0j}-t } \,\,=\,\, (p s)^{s_{0j}-t} \times s^{-(s_{0j}-s)},
		\eea
		so the last term can be decomposed by
		\bea
		\sum_{j=2}^p\sum_{s=0}^{s_{0j}-1}\sum_{t=0}^{s}  \Big( \frac{\nu_1}{\nu_2 s} c_1 p^{c_2} \Big)^{s_{0j}-s}  ( \nu_2 p^{-C_{\rm bm}+1 }s )^{s_{0j}-t}  
		&\lesssim& \sum_{j=2}^p\sum_{s=0}^{s_{0j}-1} \left(\nu_1 c_1 p^{-C_{\rm bm} +c_2 +1}  \right)^{s_{0j}-s}\\
		&\lesssim& \nu_1 c_1 p^{-C_{\rm bm}+c_2 +2} \,\,=\,\, o(1), \text{ and}
		\eea
		\bea
		\sum_{j=2}^p\sum_{s=s_{0j}}^{R_j} \sum_{t=0}^{s_{0j}-1}  \Big( \frac{\nu_1}{\nu_2 s} c_1 p^{c_2} \Big)^{s_{0j}-s}  (\nu_2 p^{-C_{\rm bm}+1 } s )^{s_{0j}-t}  
		&\lesssim& \sum_{j=2}^p \nu_2 p^{-C_{\rm bm}+1 } R_j   \\
		&\le& \sup_j R_j \cdot \nu_2 p^{-C_{\rm bm}+2 } \,\,=\,\, o(1)  ,
		\eea
		provided that $C_{\rm bm}> c_2 +2$.
		Note that $s_{0j} \le R_j$ because of Condition \hyperref[condP]{(P)} and $s_0\log p \le n \,c_3/2$.
	\end{proof}

	\subsection{Lemmas for the proof of Theorem 3.1}

	\begin{lemma} Under the conditions in Theorem 3.1, we have
		\bea
		\sum_{j=2}^p \pi_\alpha (S_j \supsetneq S_{0j}\mid \bfX_n) &=& o(1).
		\eea
	\end{lemma}
	\begin{proof}
		For a given $S_j \supsetneq S_{0j}$, we have
		\bea
		\pi_\alpha(S_j \mid \bfX_n) &\le& \frac{\pi_\alpha(S_j \mid \bfX_n)}{\pi_\alpha(S_{0j} \mid \bfX_n)} \\
		&=& \frac{\pi_j(S_j )}{\pi_j(S_{0j} )} \left(1+ \frac{\alpha}{\gamma} \right)^{- \frac{|S_j| - |S_{0j}|}{2}} \cdot \left(\frac{\what{d}_{S_j}}{\what{d}_{S_{0j}}} \right)^{-\frac{\alpha n + \nu_0}{2}}.
		\eea
		Note that $nd_{0j}^{-1} \what{d}_{S_j}= d_{0j}^{-1}\tilde{X}_j^T(I_n - \tilde{P}_{S_j})\tilde{X}_j \sim \chi^2_{n-|S_j|}$ and $nd_{0j}^{-1}\what{d}_{S_{0j}} \overset{d}{\equiv} nd_{0j}^{-1} \what{d}_{S_j}\oplus \chi^2_{|S_j|-|S_{0j}|}$ given $\tilde{Z}_j = (Z_{1j},\ldots, Z_{nj})^T$ under $\bbP_{0}$, which implies $\what{d}_{S_{j}}/\what{d}_{S_{0j}} \sim Beta\left((n-|S_j|)/2, (|S_j|-|S_{0j}|)/2  \right)$ and
		\bea
		\bbE_0 \left(\frac{\what{d}_{S_j}}{\what{d}_{S_{0j}}} \right)^{-\frac{\alpha n +\nu_0}{2}} 
		&=& \frac{\Gamma\left(\frac{n-|S_{0j}|}{2}  \right) }{\Gamma\left( \frac{n-|S_j|}{2} \right)} \cdot \frac{\Gamma\left(\frac{n(1-\alpha) -\nu_0-|S_{j}|}{2}  \right) }{\Gamma\left( \frac{n(1-\alpha) -\nu_0-|S_{0j}|}{2} \right)} \\
		&\le& \left(\frac{n-|S_{0j}| -2}{2} \right)^{\frac{|S_j|-|S_{0j}|}{2}} \cdot \left(\frac{2}{n(1-\alpha) - \nu_0 -|S_j|} \right)^{\frac{|S_j|-|S_{0j}|}{2}} \\
		&\le& \left( \frac{2(n-2)}{n(1-\alpha)} \right)^{\frac{|S_j|-|S_{0j}|}{2}} \,\,\le\,\, \left(\frac{2}{1-\alpha} \right)^{\frac{|S_j|-|S_{0j}|}{2}},
		\eea
		where the second inequality holds because $\nu_0+|S_j|\le \nu_0+ R_j  \le n(1-\alpha)/2$ for all large $n$. 
		Let $c_{\alpha,\gamma}= (1+\alpha/\gamma)^{-1/2}(2/(1-\alpha))^{1/2}$ and $s_{0j}=|S_{0j}|$, 
		then we have
		\bea
		\sum_{j=2}^p \bbE_0 \pi_\alpha(S_j \supsetneq S_{0j} \mid \bfX_n )
		&\le& \sum_{j=2}^p \sum_{S_j: S_j \supsetneq S_{0j} } \frac{\pi_j(S_j )}{\pi_j(S_{0j} )} c_{\alpha,\gamma}^{|S_j| - |S_{0j}|} \\
		&\le& \sum_{j=2}^p \sum_{s= s_{0j}+1}^{R_j} \binom{s}{s_{0j}} \left(\frac{c_{\alpha,\gamma} }{c_1 p^{c_2}}  \right)^{s-s_{0j}} \\
		&\le& \sum_{j=2}^p \sum_{s= s_{0j}+1}^{R_j} \left(\frac{c_{\alpha,\gamma} s}{c_1 p^{c_2}}  \right)^{s-s_{0j}} \,\,\lesssim\,\, \sum_{j=2}^p \,\, \frac{c_{\alpha,\gamma} R_j}{c_1 p^{c_2}} .
		\eea
		The last display is of order $o(1)$ because we assume that $c_2 \ge 2$.
	\end{proof}

	\begin{lemma} Under the conditions in Theorem 3.1, we have
		\bea
		\sum_{j=2}^p \sum_{S_j: S_j \supsetneq S_{0j}} \Big\{ \bbP_0(N_{1,S_j,\alpha, \chi^2}) + \bbP_0(N_{2,S_j,\alpha, \chi^2})+ \bbP_0(N_{3,S_j,\alpha, \chi^2})  \Big\}  &=& o(1),
		\eea
		where $N_{1,S_j,\alpha, \chi^2}, N_{2,S_j,\alpha, \chi^2}$ and $N_{3,S_j,\alpha, \chi^2}$ are defined in the proof of Theorem 3.1.
	\end{lemma}
	\begin{proof}
		By Lemma 1 in \cite{laurent2000adaptive}, $P(\chi_k^2 -k \ge 2\sqrt{kx} + 2x) \le \exp(-x)$ and $P(k- \chi_k^2 \ge 2\sqrt{kx}) \le \exp(-x)$ for all $x>0$.
		It is easy to check that
		\bea
		\bbP_0 (N_{1,S_j,\alpha, \chi^2}) &=& \bbP_0 \Big( |(n- s_{0j})^{-1}\chi_{n- s_{0j}}^2 - 1| \ge 4 \sqrt{\epsilon'}  \Big) \\
		&\le& 2 e^{-\epsilon'(n- s_{0j})} \,\,\le\,\, 2 e^{- \frac{\epsilon' n}{2}}, \\
		\bbP_0 (N_{2,S_j,\alpha, \chi^2}) &=& \bbP_0 \Big(  (|S_{0j}\cup S_j| - s_{0j})^{-1} \chi_{|S_{0j}\cup S_j| - s_{0j}}^2 -1  \ge \frac{4\epsilon' n}{|S_{0j}\cup S_j| - s_{0j}}  \Big)\\
		&\le& e^{-\epsilon' n}
		\eea
		for all sufficiently large $n$.
		For the third term $\bbP_0(N_{3,S_j,\alpha, \chi^2})$, note that $n (\what{d}_{S_{j}} - \what{d}_{S_{0j}\cup S_j}) /d_{0j}$ follows the noncentral chi-square distribution with $|S_{0j}\cup S_j| - |S_j|$ degrees of freedom and the noncentrality parameter $\what{\lambda}_n$ under $\bbP_0$ given $\tilde{Z}_j$.
		Note that on the event $\bfX_n \in N_{j,S_j}^c$ defined in the proof of Theorem 3.1, 
		\bea
		\what{\lambda}_n &=& d_{0j}^{-1}\| (I_n - \tilde{P}_{S_j})\tilde{Z}_{j} a_{0j} \|_2^2 \,\,\le\,\, \lambda_{\max}(\tilde{Z}_{S_{0j}}^T \tilde{Z}_{S_{0j}} ) \cdot d_{0j}^{-1}  \|a_{0j}\|_2^2 \\
		&\le&  \epsilon_0^{-1} n (1+2\epsilon_0)^2 \cdot \| d_{0j}^{-1/2} a_{0j} \|^2\\
		&\le& \epsilon_0^{-1} n (1+2\epsilon_0)^2 \cdot 2 \Big\{ \|d_{0j}^{-1/2}(e_j - a_{0j})\|_2^2 + d_{0j}^{-1}  \Big\} \\
		&\le& \epsilon_0^{-1} n (1+2\epsilon_0)^2 \cdot 2 \epsilon_0^{-1} ,
		\eea
		where $e_j$ it the unit vector whose $j$th element is 1 and the others are zero.
		By Lemma 4 in \cite{shin2015scalable},
		\bea
		\bbP_0 (N_{3,S_j,\alpha, \chi^2})
		&\le& \bbE_0 \Bigg[ C \Big( \frac{\epsilon' n}{2(|S_{0j}\cup S_j| - |S_j| )} \Big)^{\frac{|S_{0j}\cup S_j| - |S_j|}{2}} e^{ \frac{|S_{0j}\cup S_j| - |S_j|}{2} - \frac{\epsilon'n}{2} } \\
		&&\quad+\quad \Big\{ C \frac{\what{\lambda}_n}{\epsilon'n} e^{- \frac{\epsilon'^2 n^2 }{32\what{\lambda}_n}}  \wedge 1  \Big\} \Bigg] \\
		&\le& e^{- \frac{\epsilon'n}{4}} + \bbE_0\left[ C \frac{\what{\lambda}_n}{\epsilon'n} e^{- \frac{\epsilon'^2 n^2 }{32\what{\lambda}_n}} I_{N_{j,S_j}^c} \right]  + \bbP_0 (N_{j,S_j}) \\
		&\le& e^{- \frac{\epsilon'n}{4}} + e^{- \frac{ \epsilon'^2 \epsilon_0^2 }{64 (1+2\epsilon_0)^2  } \, n } + 4 e^{-\frac{\epsilon_0^2 n}{2} }
		\eea
		for all sufficiently large $n$, for some constant $C>0$.
		Thus, by Condition \hyperref[condP]{(P)}, it completes the proof. 
	\end{proof}

	\section{Proofs of Posterior Convergence Rates for Precision Matrices}\label{sec:proof_conv_prec}
	
	Recall that we consider the model
	\bean\label{model}
	X_1,\ldots,X_n \mid \Omega_n &\overset{i.i.d}{\sim}& N_p(0, \Omega_n^{-1}),
	\eean
	where $\Omega_n = \sg_n^{-1}$ is a $p\times p$ precision matrix and $X_i = (X_{i1},\ldots, X_{ip})^T\in \bbR^p$ for all $i=1,\ldots,n$.

	We also introduce some notations here which will be used in the proofs in the supplementary material.
	We define $\what{\V}(X_j) = n^{-1}\|\tilde{X}_j\|_2^2$ for $j=1,\ldots,p$.
	For a given index set $S \subseteq \{1,\ldots,p\}$, we define $\what{\V}(Z_S) = n^{-1} \bfX_{S}^T \bfX_{S}$ and $\what{\C}(Z_S, X_j) = n^{-1} \bfX_{S}^T \tilde{X}_j$.

	\begin{lemma}\label{lemma:Nsets}
		Let $\bfX_n$ be the random sample of size $n$ from $N_p(0, \sg_{0n})$ with $\epsilon_0 \le \lambda_{\min}(\sg_{0n}) \le \lambda_{\max}(\sg_{0n}) \le \epsilon_0^{-1}$ for some constant $0< \epsilon_0 < 1/2$.
		Define $C_{\max} = (1+ 2\epsilon_0)^2$, $C_{\min} = (1- 2\epsilon_0)^{2}$,  
		\bea
		N_{1,R,\epsilon_0} &=& \Big\{ \bfX_n: n^{-1} \Psi_{\max}(R)^2 \ge C_{\max} \epsilon_0^{-1}  \Big\} \quad \text{ and} \\
		N_{2,R,\epsilon_0} &=& \Big\{ \bfX_n: n^{-1} \Psi_{\min}(R)^2 \le C_{\min} \epsilon_0 \Big\},
		\eea
		for some positive integer $R$. If $R = o(n)$, we have
		\bea
		\bbP_0 (N_{1,R,\epsilon_0}) &\le& 2 \exp \left( -\frac{n}{2}\epsilon_0^2 + R \log p + \log R \right) \quad\text{ and}  \\
		\bbP_0(N_{2,R,\epsilon_0}) &\le& 2 \exp \left( - \frac{n }{2}\epsilon_0^2 + R \log p + \log R \right)
		\eea
		for all sufficiently large $n$.
	\end{lemma}
	\begin{proof}
		We only prove the upper bound for $\bbP_0(N_{2,R,\epsilon_0})$, because the upper bound for $\bbP_0(N_{1,R,\epsilon_0})$ can be proved easily by the similar arguments.
		For any given index set $S \subseteq \{1,\ldots,p\}$ such that $0<|S|\le R$, it is easy to show that $n^{-1} \sg_{0n, S}^{-1/2}\bfX_{S}^T \bfX_{S}\sg_{0n, S}^{-1/2} \sim W_{|S|}(n, n^{-1}I_{|S|} )$ and $\lambda_{\min}(\sg_{0n, S}) \ge \epsilon_0$.
		Let $C_{\min} = (1-2\epsilon_0)^2$.
		By Corollary 5.35 in \cite{eldar2012compressed} with $t = \epsilon_0\sqrt{n}$,
		\bea
		\bbP_0 \left( n^{-1} \lambda_{\min}(\bfX_{S}^T \bfX_{S} ) \le C_{\min} \epsilon_0 \right)  
		&\le& \bbP_0 \left( n^{-1} \epsilon_0 \lambda_{\min}(\sg_{0n, S}^{-1/2}\bfX_{S}^T \bfX_{S}\sg_{0n, S}^{-1/2} ) \le C_{\min} \epsilon_0 \right)  \\
		&\le&  2 \exp\left( -\frac{n}{2}\epsilon_0^2  \right)
		\eea
		for all sufficiently large $n$ because $R=o(n)$.
		Thus, we have
		\bea
		\bbP_0  \left( R_{2,R,\epsilon_0}   \right) 
		&=& \bbP_0 \left(  \inf_{S: 0< |S|\le R} n^{-1}\lambda_{\min}( \bfX_{S}^T \bfX_{S})  \le C_{\min} \epsilon_0  \right) \\
		&\le& \sum_{S: 0< |S|\le R}   \bbP_0 \left(  n^{-1}\lambda_{\min}(\bfX_{S}^T \bfX_{S})  \le C_{\min} \epsilon_0  \right) \\
		&\le& R \times p^{R} \times 2 \exp\left( -\frac{n }{2} \epsilon_0^2  \right) \\
		&=& 2 \exp \left( - \frac{n }{2}\epsilon_0^2 + R \log p + \log R \right)
		\eea
		for all sufficiently large $n$.
	\end{proof}

	\begin{lemma}\label{lemma:Nsets4}
		Let $\bfX_n$ be the random sample of size $n$ from $N_p(0, \sg_{0n})$ with $\epsilon_0 \le \lambda_{\min}(\sg_{0n}) \le \lambda_{\max}(\sg_{0n}) \le \epsilon_0^{-1}$ for some constant $0< \epsilon_0 < 1/2$.
		For a given constant $K_{\rm diff}>0$, define 
		\bea
		N_{1, S_0, \epsilon_0} &=& \left\{ \bfX_n : \max_{2\le j\le p}\| \what{\V}(Z_{S_{0j}\cup \{j\}})  \|  \ge  (1+2\epsilon_0)^2\epsilon_0^{-1}  \right\}, \\
		N_{2, S_0, \epsilon_0} &=& \left\{ \bfX_n : \max_{2\le j\le p}\| \what{\V}^{-1}(Z_{S_{0j}\cup \{j\}})  \|  \ge (1-2\epsilon_0)^{2}\epsilon_0^{-1}  \right\}, \\
		N_{3, S_0, \epsilon_0} &=& \left\{ \bfX_n : \max_{2\le j\le p}\| \what{\V}(Z_{S_{0j}\cup \{j\}}) - \V(Z_{S_{0j}\cup \{j\}}) \|  \ge \sqrt{ K_{\rm diff}\cdot\frac{s_0 + \log p}{n} }  \right\} \quad \text{ and}\\
		N_{4, S_0, \epsilon_0} &=& \left\{ \bfX_n : \max_{2\le j\le p}\| \what{\V}^{-1}(Z_{S_{0j}\cup \{j\}}) - \V^{-1}(Z_{S_{0j}\cup \{j\}}) \|  \ge C_{\epsilon_0}\sqrt{K_{\rm diff} \cdot\frac{s_0 + \log p}{n} }  \right\} ,
		\eea
		where $C_{\epsilon_0}  = (1-2\epsilon_0)^{2}\epsilon_0^{-2}$.
		Let $N_{S_0,\epsilon_0} = \cup_{j=1}^4 N_{j,S_0,\epsilon_0}$.
		If $s_0+ \log p = o(n)$, there exists a universal constant $C>0$ such that
		\bea
		\bbP_0 \big(  N_{S_0, \epsilon_0}  \big) 
		&\le& 6\cdot p \exp\left(-\frac{n }{2}\epsilon_0^2 \right) +  4 \cdot p 5^{s_0} \exp \Big( -K_{\rm diff} C \epsilon_0^2 (s_0 + \log p) \Big)
		\eea
		for all sufficiently large $n$.
	\end{lemma}
	\begin{proof}
		It is easy to show that
		\bea
		\bbP_0 \big( N_{1,S_0,\epsilon_0} \big) + \bbP_0 \big( N_{2,S_0,\epsilon_0} \big)
		&\le& 4 \cdot p \exp \left( - \frac{n }{2}\epsilon_0^2 \right)
		\eea
		by the similar arguments in the proof of Lemma \ref{lemma:Nsets}.
		Thus, it suffices to show that
		\begin{align}
		\bbP_0 \big( N_{3,S_0,\epsilon_0}  \big) 
		&\le 2\cdot p 5^{s_0} \exp \Big( -K_{\rm diff} C \epsilon_0^2 (s_0+ \log p ) \Big), \label{N3}\\
		\bbP_0 \big( N_{4,S_0,\epsilon_0}  \big) 
		&\le 2 p\cdot \left\{ 5^{s_0} \exp \Big( -K_{\rm diff} C \epsilon_0^2 (s_0 +\log p ) \Big) +  \exp \left( - \frac{n }{2}\epsilon_0^2 \right) \right\} \label{N4}
		\end{align}
		for all sufficiently large $n$ and some constant $C>0$.
		Note that for any $p\times p$ symmetric matrix $V$, there exist $v_j \in \bbR^p$ with $\|v_j\|_2=1$ for $j=1,\ldots,5^p$ such that
		\bea
		\|V\| &\le& 4 \cdot \sup_{1\le j\le 5^p} |v_j^T V v_j |,
		\eea
		by page 2141 of \cite{cai2010optimal}.
		Thus,
		\bea
		&& \bbP_0 \big(N_{3,S_0, \epsilon_0} \big) \\
		&=& \bbP_0 \left( \max_{2\le j\le p}  \| \what{\V}(Z_{S_{0j}\cup \{j\}}) - \V(Z_{S_{0j}\cup \{j\}}) \|  \ge \sqrt{K_{\rm diff}\frac{s_0 + \log p}{n} } \right)\\
		&\le& \bbP_0 \left( \max_{2\le j\le p} \| \what{W}_{S_{0j}\cup \{j\} } - I \|  \ge \epsilon_0 \sqrt{K_{\rm diff}\frac{s_0 + \log p}{n} } \right) \\
		&\le&  p 5^{s_0} \max_{2\le j\le p} \sup_{1\le j\le 5^{s_{0j}+1}} \bbP_0 \left(  | v_j^T ( \what{W}_{S_{0j}\cup \{j\} } - I )v_j |  \ge \frac{\epsilon_0}{4} \sqrt{K_{\rm diff}\frac{s_0 + \log p}{n} } \right) ,
		\eea
		where $\what{W}_{S_{0j}\cup \{j\} } := \V(Z_{S_{0j}\cup \{j\}})^{-1/2}\what{\V}(Z_{S_{0j}\cup \{j\}})\V(Z_{S_{0j}\cup \{j\}})^{-1/2}$, and $\what{W}_{S_{0j}\cup \{j\} } \sim W_{|S_{0j}|+1}(n, n^{-1}I )$.
		Note that $n \,v_j^T \what{W}_{S_{0j}\cup \{j\} } v_j \sim \chi_n^2$ by the property of Wishart distribution, and $P(|\chi_n^2 - n| \ge 2\sqrt{nt}+2t ) \le \exp(-t)$ for all $t>0$ by Lemma 1 in \cite{laurent2000adaptive}.
		Thus, 
		\bea
		\bbP_0 \big(N_{3,S_0, \epsilon_0} \big)
		&\le& \exp \left( - \frac{\epsilon_0^2}{4^4} K_{\rm diff }(s_0+\log p) \right).
		\eea

		Similarly, 
		\bea
		\bbP_0 \big(N_{4,S_0, \epsilon_0} \big)
		&\le& \bbP_0 \big(N_{4,S_0, \epsilon_0} \cap N_{2,S_0,\epsilon_0}^c \big) + \bbP_0 \big(N_{2,S_0, \epsilon_0} \big) \\
		&\le& \bbP_0 \left( \max_{2\le j\le p}  \| \what{\V}(Z_{S_{0j}\cup \{j\}}) - \V(Z_{S_{0j}\cup \{j\}}) \|  \ge \sqrt{K_{\rm diff}\frac{s_0 + \log p}{n} } \right)  \\
		&& +\,\, 2 \cdot p \exp \left( - \frac{n }{2}\epsilon_0^2 \right) \\
		&\le& 2 \cdot p 5^{s_0} \exp \Big( -K_{\rm diff} C \epsilon_0^2 (s_0+ \log p ) \Big) + 2 \cdot p \exp \left( - \frac{n }{2}\epsilon_0^2 \right)
		\eea 
		for all sufficiently large $n$,	thus, we have \eqref{N4}.
		The second inequality follows from 
		\bea
		&&\| \what{\V}^{-1}(Z_{S_{0j}\cup \{j\}}) - \V^{-1}(Z_{S_{0j}\cup \{j\}}) \| \\
		&\le& \|\what{\V}^{-1}(Z_{S_{0j}\cup \{j\}})\| \| \V^{-1}(Z_{S_{0j}\cup \{j\}}) \| \| \what{\V}(Z_{S_{0j}\cup \{j\}}) - \V(Z_{S_{0j}\cup \{j\}}) \| .
		\eea
	\end{proof}

	\begin{lemma}\label{lemma:freq_conv}
		Let $\bfX_n$ be the random sample of size $n$ from $N_p(0, \Omega_{0n}^{-1})$ with $\Omega_{0n}$ satisfying \hyperref[A1]{\rm(A1)}, \hyperref[A2]{\rm(A2)} and \hyperref[A4]{\rm(A4)} for some constant $0< \epsilon_0 < 1/2$ and a sequence of positive integers $s_0$.
		Let $N_{S_0,\epsilon_0}$ be the set defined at Lemma \ref{lemma:Nsets4}.
		If $s_0 + \log p =o(n)$ and $s_0^{3/2} (s_0 + \log p) =O(n)$, we have
		\bea
		\| \what{\Omega}_n - \Omega_{0n}\| &\lesssim&  s_0^{3/4}\left( \frac{s_0 + \log p}{n} \right)^{1/2} 
		\eea
		on $\bfX_n \in N_{S_0,\epsilon_0}^c$, for all sufficiently large $n$.
		If we further assume $s_0 (s_0 + \log p) =O(n)$, then
		\bea
		\| \what{\Omega}_n - \Omega_{0n}\|_\infty &\lesssim&  \|I_p - A_{0n}\|_\infty \cdot  s_0 \left( \frac{s_0 + \log p}{n} \right)^{1/2} 
		\eea
		on $\bfX_n \in N_{S_0,\epsilon_0}^c$, for all sufficiently large $n$.
	\end{lemma}
	
	\begin{proof}
		Throughout the proof, we only consider the event $\bfX_n \in N_{S_0,\epsilon_0}^c$.
		Consider the spectral norm case first.
		By the triangle inequality, 
		\bean
		\begin{split}\label{whatO_diff_O}
			&\quad\,\,\, \| \what{\Omega}_n - \Omega_{0n}\| \\
			\,\,&\le\,\, \|I_p - \what{A}_n \|^2 \cdot \|\what{D}_n^{-1} - D_{0n}^{-1}\| 
			+ \|I_p - \what{A}_n \| \cdot \|D_{0n}^{-1}\| \cdot \| \what{A}_n - A_{0n}\| \\
			&+\,\, \|I_p - A_{0n}\| \cdot \|D_{0n}^{-1}\| \cdot \| \what{A}_n - A_{0n}\|.
		\end{split}
		\eean
		Note that
		\bean
		\begin{split}\label{Ahat_diff_0_infty}
			\| \what{A}_n -  A_{0n}\|_\infty
			\,\,&=\,\, \max_j \| \what{a}_{S_{0j}} - a_{0,S_{0j}}\|_1 \\
			\,\,&\le\,\, \sqrt{s_0} \max_j \| \what{a}_{S_{0j}} - a_{0,S_{0j}}\|_2 \\
			\,\,&\le\,\, \sqrt{s_0} \Big\{ \max_j \| \V^{-1}(Z_{S_{0j}})\cdot \big( \what{\C}(Z_{S_{0j}}, X_j) - \C(Z_{S_{0j}}, X_j) \big)  \|_2  \\
			\,\,& \quad+\,\, \max_j \| \big(\what{\V}^{-1}(Z_{S_{0j}}) - \V^{-1}(Z_{S_{0j}})\big) \cdot \what{\C}(Z_{S_{0j}}, X_j) \|  \Big\} \\
			\,\,&\lesssim\,\,  \sqrt{s_0}\left( \frac{s_0 + \log p}{n} \right)^{1/2}
		\end{split}
		\eean
		by the definition of $N_{S_0,\epsilon_0}^c$. 
		Similarly, it is easy to show that
		\bea
		\| \what{A}_n -  A_{0n}\|_1 
		&\le& s_0 \max_j  \| \what{a}_{S_{0j}} - a_{0,S_{0j}}\|_{\max} \\
		&\le& s_0 \max_j  \| \what{a}_{S_{0j}} - a_{0,S_{0j}}\|_2 \\
		&\lesssim& s_0 \left(\frac{s_0 + \log p}{n} \right)^{1/2}.
		\eea
		Thus, we have
		\bea
		\| \what{A}_n -  A_{0n}\|
		&\le& \| \what{A}_n -  A_{0n}\|_\infty^{1/2} \cdot \| \what{A}_n -  A_{0n}\|_1^{1/2} \\
		&\le& s_0^{3/4} \left(\frac{s_0 + \log p}{n} \right)^{1/2}.
		\eea
		On the other hand, note that
		\bea
		\| \what{D}_n^{-1} - D_{0n}^{-1}\|
		&\le& \|\what{D}_n^{-1}\| \cdot \|{D}_{0n}^{-1}\| \cdot \|\what{D}_n - {D}_n\| \\
		&\le& (1- 2\epsilon_0)^{-2} \epsilon_0^{-1} \cdot \epsilon_0^{-1} \cdot \|\what{D}_n - {D}_n\| 
		\eea
		and
		\bea
		\|\what{D}_n - {D}_{0n}\|
		&=& \max_j | \what{d_{S_{0j}}} - d_{0j}| \\
		&\le& \max_j \Big| \what{\V}(X_j) - \V(X_j)\Big| \\
		&&+\quad \max_j \Big| \what{\C}(X_j, Z_{S_{0j}}) \what{a}_{S_{0j}} - \C(X_j, Z_{S_{0j}}) a_{0,S_{0j}}\Big| \\
		&\lesssim& \left(\frac{s_0 + \log p}{n} \right)^{1/2}.
		\eea
		Also note that
		\bea
		\|I_p - \what{A}_n \|
		&\le& \|I_p - A_{0n} \| + \| \what{A}_n -  A_{0n}\| \\
		&\le& \epsilon_0^{-1} + s_0^{3/4} \left(\frac{s_0 + \log p}{n} \right)^{1/2},
		\eea
		where the last display is of order $O(1)$ provided that $s_0^{3/2}(s_0 + \log p) = O(n)$.
		The second inequality follows from $\epsilon_0^{-1} \ge \|(I_p- A_{0n})^T D_{0n}^{-1}(I_p-A_{0n}) \| \ge \lambda_{\min}(D_{0n}^{-1}) \|(I_p- A_{0n})(I_p- A_{0n})^T\|  \ge \epsilon_0 \|I_p- A_{0n}\|^2 $. 
		By \eqref{whatO_diff_O}, we have shown the spectral norm result.

		Now, consider the matrix $\ell_\infty$ norm case. 
		Similar to \eqref{whatO_diff_O}, by the triangle inequality,
		\bean
		\begin{split}\label{whatO_diff_01}
			\| \what{\Omega}_n - \Omega_{0n}\|_\infty
			\,\,&\le\,\, \|I_p - \what{A}_n \|_1 \cdot \|I_p - \what{A}_n \|_\infty \cdot \|\what{D}_n^{-1} - D_{0n}^{-1}\|  \\
			&+\,\, \|I_p - \what{A}_n \|_1 \cdot \|D_{0n}^{-1}\| \cdot \| \what{A}_n - A_{0n}\|_\infty \\
			&+\,\, \|I_p - A_{0n}\|_\infty \cdot \|D_{0n}^{-1}\| \cdot \| \what{A}_n - A_{0n}\|_1.
		\end{split}
		\eean
		From the above arguments and \eqref{whatO_diff_01}, we only need to show that 
		\bea
		\|I_p - \what{A}_n \|_1 &\lesssim& \sqrt{s_0}. 
		\eea
		It is easy to show that
		\bea
		\| I_p - \what{A}_n \|_1 &\le& \| I_p - A_{0n}\|_1 + \| \what{A}_n- A_{0n}\|_1 \\
		&\lesssim& \| I_p - A_{0n}\|_1 + s_0 \left( \frac{s_0 + \log p}{n} \right)^{1/2} \\
		&\lesssim& \| I_p - A_{0n}\|_1 + \sqrt{s_0} 
		\eea
		because we assume that $s_0 (s_0 + \log p) =O(n)$.
		If we show that $\| I_p - A_{0n}\|_1 \le \sqrt{s_0 + 1} \, \|I_p - A_{0n}\|$, it completes the proof.
		Let $a_{c,0j}$ be the $j$th column vector of $A_{0n}$ and $e_j \in \bbR^p$ be the unit vector whose $j$th element is 1 and the others are 0, then
		\bea
		\| I_p - A_{0n} \|_1 &=& \max_j \| e_j - a_{c,0j}\|_1\\
		&\le& \sqrt{s_0 + 1} \max_j \| e_j - a_{c,0j}\|_2  ,
		\eea
		by the condition \hyperref[A4]{(A4)}.
		Note that $\max_j \| e_j - a_{c,0j}\|_2$ is the maximum $\ell_2$ norm of columns of $I_p - A_{0n}$, which is smaller than $\|I_p - A_{0n} \|$. Since $\|I_p - A_{0n} \| \le \epsilon_0^{-1}$, we have $\|I_p - \what{A}_n \|_1 \lesssim \sqrt{s_0}$. 
	\end{proof}

	Recall that we consider the beta-min condition \hyperref[A3]{\rm (A3)} such that
	\bea
	\min_{j,l: a_{0,jl}\neq 0 }|a_{0,jl}|^2 &\ge& \frac{16}{\alpha(1-\alpha)\, \epsilon_0^2(1-2\epsilon_0)^2 }\cdot  C_{\rm bm} \cdot \frac{\log p}{n} 
	\eea
	for some constant $C_{\rm bm}>0$ and $0<\alpha<1$.
	
	\begin{lemma}\label{lemma:An_post}
		For given positive constants $0<\alpha < 1$, $0<\epsilon_0< 1/2$, $C_{\rm bm} > c_2+2$ and an integer $s_0$, assume model \eqref{model} and the ESC prior with Condition \hyperref[condP]{\rm(P)}.
		If $s_0 \log p =o(n)$, then
		\bea
		\sup_{\Omega_{0n}\in \calU_p^*} \bbE_0 \left[\pi_\alpha \left( \|A_n - \what{A}_n \|_\infty \ge K_{1} \sqrt{s_0} \left(\frac{s_0 + \log p}{n} \right)^{1/2}  \,\,\Big|\,\, \bfX_n \right)\right]
		&=& o(1), \\
		\sup_{\Omega_{0n}\in \calU_p^*} \bbE_0 \left[\pi_\alpha \left( \|A_n - \what{A}_n \|_1 \ge K_{1} \sqrt{s_0} \left(\frac{s_0 + \log p}{n} \right)^{1/2}  \,\,\Big|\,\, \bfX_n \right)\right]
		&=& o(1) 
		\eea
		for some constant $K_{1}>0$.
	\end{lemma}

	\begin{proof}
		We follow closely the line of the proof of Lemma 7.4 in \cite{lee2017estimating}.
		Let $\Omega_{0n}\in \calU_p^*$ and $N_{S_0,\epsilon_0}$ be the set defined at Lemma \ref{lemma:Nsets4}.
		Note that
		\bea
		&& \bbE_0 \left[\pi_\alpha \left( \|A_n - \what{A}_n \|_\infty \ge K_{1} \sqrt{s_0} \left(\frac{s_0 + \log p}{n} \right)^{1/2}  \,\,\Big|\,\, \bfX_n \right) \right] \\
		&\le& \bbE_0 \left[ \pi_\alpha \left( \|A_n - \what{A}_n \|_\infty \ge K_{1} \sqrt{s_0} \left(\frac{s_0 + \log p}{n} \right)^{1/2}  \,\,\Big|\,\, \bfX_n \right)I_{N_{S_0,\epsilon_0}^c}\right] + o(1)
		\eea
		by Lemma \ref{lemma:Nsets4}.
		Then, on $\bfX_n \in N_{S_0,\epsilon_0}^c$, if $s_0 \log p =o(n)$,
		\bea
		\|A_n - \what{A}_n \|_\infty 
		&\le& \max_j \sqrt{s_0} \big\| a_{S_{0j}} - \what{a}_{S_{0j}} \big\|_2 \\
		&\lesssim& \max_j \sqrt{\frac{s_0 \,d_j}{n }} \cdot  \bigg\| \sqrt{\frac{n(\alpha+\gamma)}{d_{j}}}\cdot \what{\V}^{1/2}(Z_{S_{0j}}) (a_{S_{0j}} - \what{a}_{S_{0j}}) \bigg\|_2 \\
		&=:& \max_j \sqrt{\frac{s_0 \,d_j}{n }} \cdot \| std(a_{S_{0j}}) \|_2 .
		\eea
		Note that the first inequality follows from the strong model selection consistency in Theorem 3.1, so we can always concentrate on the set $S_{A_n} = S_{A_{0n}}$.
		Also note that
		\bea
		&& \bbE_0 \left[ \pi_\alpha \left( \max_j \sqrt{ d_j} \cdot  \| std(a_{S_{0j}}) \|_2 \ge K_{1}'  \left(s_0 + \log p \right)^{1/2}  \,\,\Big|\,\, \bfX_n \right)I_{N_{S_0,\epsilon_0}^c}\right] \\
		&\le& \bbE_0 \left[ \pi_\alpha \left( \max_j \sqrt{(1+2\epsilon_0)^4\epsilon_0^{-1}} \cdot  \| std(a_{S_{0j}}) \|_2 \ge K_{1}'  \left(s_0 + \log p \right)^{1/2}   \,\,\Big|\,\, \bfX_n \right)I_{N_{S_0,\epsilon_0}^c}\right] + o(1)
		\eea
		for some constant $K_1'>0$, by the similar arguments used in the proof of Lemma \ref{lemma:d_supp}.
		We only need to show that
		\bea
		\pi_\alpha \left( \max_j \sqrt{(1+2\epsilon_0)^4\epsilon_0^{-1}} \cdot  \| std(a_{S_{0j}}) \|_2 \ge K_{1}'  \left(s_0 + \log p \right)^{1/2}   \,\,\Big|\,\, \bfX_n \right)
		&=& o (1).
		\eea
		
		We can check that $\| std(a_{S_{0j}}) \|_2^2  \mid  \bfX_n \overset{ind}{\sim} \chi^2_{s_{0j}} $.
		By Lemma 1 in \cite{laurent2000adaptive}, we have
		$P \big(\chi_k^2 \ge 3(k + x) \big) \le \exp(-x)$ for all $x>0$, where $\chi_k^2$ is the chi-square random variable with degrees of freedom $k$. 
		Thus, 
		\bea
		&& \pi_\alpha \left( \max_j \| std(a_{S_{0j}}) \|_2^2 \ge \frac{K_{1}'^2}{(1+2\epsilon_0)^4\epsilon_0^{-1}}(s_0 + \log p) \,\, \mid \,\, \bfX_n \right) \\
		&\le& p \cdot  \exp \left( - \frac{K_{1}'^2}{3(1+2\epsilon_0)^4\epsilon_0^{-1}}(s_0 + \log p) + s_0 \right),
		\eea
		where the last display is of order $o(1)$ by taking $K_{1}'^2 = 6(1+2\epsilon_0)^4$.
		
		Note that all rows of $A_n$ are posteriori independent, so each column of $A_n$ has a multivariate normal posterior distribution with a diagonal covariance matrix.
		Because of the condition \hyperref[A4]{\rm (A4)}, there are at most $s_0$ nonzero elements in each column of $A_{0n}$. 
		Then, by the similar arguments, we have
		\bea
		\sup_{\Omega_{0n}\in \calU_p^*} \bbE_0 \left[\pi_\alpha \left( \|A_n - \what{A}_n \|_1 \ge K_{1} \sqrt{s_0} \left(\frac{s_0 + \log p}{n} \right)^{1/2}  \,\,\Big|\,\, \bfX_n \right)\right]
		&=& o(1) . 
		\eea
	\end{proof}

	\begin{lemma}\label{lemma:d_supp}
		Let $\bfX_n$ be the random sample of size $n$ from $N_p(0, \sg_{0n})$ with $\epsilon_0 \le \lambda_{\min}(\sg_{0n}) \le \lambda_{\max}(\sg_{0n}) \le \epsilon_0^{-1}$ for some small constant $0< \epsilon_0 < 1/2$.
		Consider the model \eqref{model} and the ESC prior with Condition \hyperref[condP]{\rm(P)}.
		Let $N_{1,R,\epsilon_0}$ and $N_{2,R,\epsilon_0}$ be the sets defined at Lemma \ref{lemma:Nsets}. 
		Then, for a given constant $0<\alpha<1$ and an integer $2\le j \le p$, we have 
		\bea
		\pi_\alpha(M_1 \le d_j \le M_2 \mid \bfX_n) &\ge& 1- 2e^{-n C_{\alpha,\epsilon_0}}  \quad \text{on \,\,$\bfX_n \in N_{1,R_j,\epsilon_0}^c \cap N_{2,R_j,\epsilon_0}^c$},
		\eea
		for all sufficiently large $n$ and some constant $C_{\alpha,\epsilon_0}>0$ depending only on $\alpha$ and $\epsilon_0$, where $M_1 \le (1-2\epsilon_0)^4 \epsilon_0$ and $M_2 \ge (1+2\epsilon_0)^4 \epsilon_0^{-1}$.
	\end{lemma}
	
	\begin{proof}
		Let $N_{R_j, \epsilon_0}^c = N_{1,R_j,\epsilon_0}^c \cap N_{2,R_j,\epsilon_0}^c$. Throughout the proof, we consider only the event $\bfX_n \in N_{R_j, \epsilon_0}^c$.
		Since 
		\bea
		&& \min_{S_j\subseteq \{1,\ldots,j-1\}: \atop 0<|S_j|\le R_j } \pi_\alpha(M_1 \le d_j \le M_2 \mid S_j, \bfX_n)  \\
		&\le& \sum_{S_j\subseteq \{1,\ldots,j-1\}: \atop 0<|S_j|\le R_j }  \pi_\alpha(M_1 \le d_j \le M_2 \mid S_j, \bfX_n) \pi_\alpha(S_j \mid \bfX_n)  \\
		&=& \pi_\alpha(M_1\le d_j \le M_2 \mid \bfX_n),
		\eea
		it suffices to prove that
		\bean\label{d_supp_M12}
		\pi_\alpha(M_1 \le d_j \le M_2 \mid S_j, \bfX_n)  &\ge& 1 - 2e^{-n C_{\alpha,\epsilon_0} }  
		\eean
		for any $j$ and $S_j \subseteq \{1,\ldots, j-1\}$ such that $0<|S_j|\le R_j$.
		Note that 
		\bea
		\what{d}_{S_j}&\le& \what{\V}(X_j) \,\,\le\,\, C_{\max}\epsilon_0^{-1} , \\
		\what{d}_{S_j}^{-1} &\le& \Big\| \what{\V}^{1/2}(Z_{S_j \cup \{j\} }) \cdot \binom{- \what{a}_{S_j}}{1} \Big\|_2^{-2} \\
		&\le& \lambda_{\min}(\what{\V}^{1/2}(Z_{S_j \cup \{j\} }) )^{-1}  \,\,\le\,\, C_{\min}^{-1} \epsilon_0^{-1},
		\eea
		where $C_{\max} = (1+2\epsilon_0)^2$ and $C_{\min} = (1- 2\epsilon_0)^2$.
		It is easy to check that $d_j^{-1} \mid S_j, \bfX_n \sim Gamma((\alpha n + \nu_0 )/2 , \alpha n \,\what{d}_{S_j}/2 )$, where $Gamma(a,b)$ is the gamma distribution with the shape parameter $a>0$ and rate parameter $b>0$.
		Note that
		\bea
		&& \pi_\alpha(d_j < M_1 \mid S_j, \bfX_n ) \\
		&=& \pi_\alpha  \left(d_j^{-1} - \frac{\alpha n + \nu_0}{\alpha n}\what{d}_{S_j}^{-1} > M_1^{-1} - \frac{\alpha n + \nu_0}{\alpha n}\what{d}_{S_j}^{-1} \mid S_j, \bfX_n  \right)\\
		&\le& \pi_\alpha  \left(d_j^{-1} - \frac{\alpha n+ \nu_0}{\alpha n}\what{d}_{S_j}^{-1} > M_1^{-1} - \frac{\alpha n + \nu_0}{\alpha n} C_{\min}^{-1} \epsilon_0^{-1}  \mid S_j, \bfX_n  \right)
		\eea
		If $Y$ is a sub-gamma distribution with variance factor $\nu$ and scale parameter $c$,
		\bean\label{subgamma}
		P\left(Y> \sqrt{2\nu t} + ct \right) \vee P\left(Y < -\sqrt{2\nu t} - ct \right) &\le& e^{-t}
		\eean
		for all $t>0$, by the page 29 of \cite{boucheron2013concentration}, where a centered $Gamma(a,b)$ random variable follows the sub-gamma distribution with $\nu=a/b^2$ and $c=1/b$. 
		Thus, by \eqref{subgamma} with $t= \alpha n ( \epsilon_0/2)^2$,
		\bea
		e^{- \alpha n ( \epsilon_0/2)^2} 
		&\ge& \pi_\alpha \left(d_j^{-1} - \frac{\alpha n + \nu_0}{\alpha n}\what{d}_{S_j}^{-1} > \what{d}_{S_j}^{-1} \Big(\epsilon_0 \sqrt{\frac{\alpha n + \nu_0}{\alpha n}} + \frac{\epsilon_0^2}{2}  \Big)  \,\Big|\, S_j, \bfX_n  \right) \\
		&\ge& \pi_\alpha \left(d_j^{-1} - \frac{\alpha n+ \nu_0}{\alpha n}\what{d}_{S_j}^{-1} >  C_{\min}^{-1} \Big( \sqrt{\frac{\alpha n + \nu_0}{\alpha n}} + \frac{\epsilon_0}{2}  \Big)  \,\Big|\, S_j, \bfX_n  \right).
		\eea
		Note that for all sufficiently large $n$ and small $\epsilon_0$,
		\bea
		\frac{\alpha n + \nu_0}{\alpha n} C_{\min}^{-1} \epsilon_0^{-1} + C_{\min}^{-1} \Big( \sqrt{\frac{\alpha n + \nu_0}{\alpha n}} + \frac{\epsilon_0}{2}  \Big)
		&\le& (1-2\epsilon_0)^{-4} \epsilon_0^{-1},
		\eea
		which implies
		\bean\label{d_supp_M12a}
		\pi_\alpha(d_j < M_1 \mid S_j, \bfX_n )
		&\le& e^{- \alpha n ( \epsilon_0/2)^2}
		\eean
		provided that $M_1 \le (1-2\epsilon_0)^4\epsilon_0$. 
		On the other hand, note that
		\bea
		\pi_\alpha(d_j > M_2 \mid S_j, \bfX_n ) 
		&=& \pi_\alpha  \left(d_j^{-1} - \frac{\alpha n +\nu_0}{\alpha n}\what{d}_{S_j}^{-1} < M_2^{-1} - \frac{\alpha n  +\nu_0}{\alpha n}\what{d}_{S_j}^{-1} \mid S_j, \bfX_n  \right)
		\eea
		and by \eqref{subgamma} with $t= \alpha n(\epsilon_0/2)^2$,
		\bea
		e^{-\alpha n (\epsilon_0/2)^2} 
		&\ge& \pi_\alpha \left(d_j^{-1} - \frac{\alpha n +\nu_0}{\alpha n}\what{d}_{S_j}^{-1} < -\what{d}_{S_j}^{-1} \Big( \epsilon_0 \sqrt{\frac{\alpha n + \nu_0}{\alpha n}} + \frac{\epsilon_0^2}{2} \Big)  \,\Big|\,  S_j, \bfX_n  \right) .
		\eea
		Similarly, they imply
		\bean\label{d_supp_M12b}
		\pi_\alpha(d_j > M_2 \mid S_j, \bfX_n ) 
		&\le& e^{- \alpha n ( \epsilon_0/2)^2}
		\eean
		for all sufficiently large $n$ and small $\epsilon_0$, provided that $M_2 \ge (1+2\epsilon_0)^4 \epsilon_0^{-1}$. 
		
		Hence, \eqref{d_supp_M12a} and \eqref{d_supp_M12b} imply the desired result \eqref{d_supp_M12} with $C_{\alpha,\epsilon_0} = \alpha(\epsilon_0/2)^2$.
	\end{proof}

	\begin{lemma}\label{lemma:Dn_post}
		For given positive constants $0<\alpha < 1$, $0<\epsilon_0< 1/2$, $C_{\rm bm} > c_2+2$ and an integer $s_0$, assume model \eqref{model} and the ESC prior with Condition \hyperref[condP]{\rm(P)} and $\nu_0^2 = O(n\log p)$.
		If $s_0 \log p =o(n)$, then
		\bea
		\sup_{\Omega_{0n}\in \calU_p^*} \bbE_0 \left[\pi_\alpha \left( \|D_n^{-1} - \what{D}_n^{-1} \| \ge K_{2}  \left(\frac{\log p}{n} \right)^{1/2}  \,\,\Big|\,\, \bfX_n \right)\right]
		&=& o(1)
		\eea
		for some constant $K_2>0$.
	\end{lemma}

	\begin{proof}
		Let $\Omega_{0n}\in \calU_p^*$ and $N_{S_0,\epsilon_0}$ be the set defined at Lemma \ref{lemma:Nsets4}.
		Similar to the proof of Lemma \ref{lemma:An_post}, if $s_0\log p =o(n)$, we can always concentrate on $d_j \mid S_j = S_{0j},  \bfX_n \overset{ind}{\sim} IG((\alpha n+\nu_0)/2, \alpha n \what{d}_{S_{0j}}/2 )$ for all $j=1,\ldots,p$.
		Then, $\alpha n \,\what{d}_{S_{0j}} d_j^{-1} \mid S_j=S_{0j},\bfX_n \overset{ind}{\sim} \chi^2_{\alpha n + \nu_0}$  for all $j=1,\ldots,p$.
		By Lemma 1 in \cite{laurent2000adaptive}, $P(\chi_k^2 -k \ge 2\sqrt{kx} + 2x) \le \exp(-x)$ and $P(k- \chi_k^2 \ge 2\sqrt{kx}) \le \exp(-x)$ for all $x>0$.  
		Thus,
		\bea
		&& \exp (-x) \\
		&\ge& \pi_\alpha \left( \alpha n \what{d}_{S_{0j}} d_j^{-1} - (\alpha n + \nu_0) \ge 2\sqrt{(\alpha n + \nu_0) x} + 2x \mid S_j=S_{0j}, \bfX_n  \right) \\
		&=& \pi_\alpha \left( d_j^{-1} - \frac{\alpha n + \nu_0}{\alpha n}\what{d}_{S_{0j}}^{-1} \ge \frac{2}{\alpha n} \what{d}_{S_{0j}}^{-1} \sqrt{(\alpha n + \nu_0) x} + \frac{2x}{\alpha n}\what{d}_{S_{0j}}^{-1} \mid S_j=S_{0j}, \bfX_n  \right) \\
		&=& \pi_\alpha \left( d_j^{-1} - \what{d}_{S_{0j}}^{-1} \ge \frac{\what{d}_{S_{0j}}^{-1}}{ \alpha n} \Big[2 \sqrt{(\alpha n + \nu_0) x} + 2x + \nu_0 \Big] \mid S_j=S_{0j}, \bfX_n  \right) \\
		&\ge& \pi_\alpha \left( d_j^{-1} - \what{d}_{S_{0j}}^{-1} \ge \frac{(1-2\epsilon_0)^{-2}}{\epsilon_0 \alpha n} \Big[2 \sqrt{(\alpha n + \nu_0) x} + 2x + \nu_0 \Big] \mid S_j=S_{0j}, \bfX_n  \right)I_{N_{S_0,\epsilon_0}^c}.
		\eea
		Similarly, also note that
		\bea
		\exp(-x) 
		&\ge& \pi_\alpha \left( - \alpha n \what{d}_{S_{0j}} d_j^{-1} + (\alpha n + \nu_0) \ge 2\sqrt{(\alpha n + \nu_0) x}  \mid S_j=S_{0j}, \bfX_n  \right) \\
		&\ge& \pi_\alpha \left( -d_j^{-1} + \what{d}_{S_{0j}}^{-1} \ge \frac{\what{d}_{S_{0j}}^{-1}}{ \alpha n} \Big[2 \sqrt{(\alpha n + \nu_0) x} - \nu_0 \Big] \mid S_j=S_{0j}, \bfX_n  \right) \\
		&\ge& \pi_\alpha \left( -d_j^{-1} + \what{d}_{S_{0j}}^{-1} \ge \frac{(1-2\epsilon_0)^{-2}}{\epsilon_0 \alpha n}\cdot 2 \sqrt{(\alpha n + \nu_0) x}  \mid S_j=S_{0j}, \bfX_n  \right).
		\eea
		Let $x=C \log p$ with some constant $C>1$, then 
		\bea
		\frac{(1-2\epsilon_0)^{-2}}{\epsilon_0 \alpha n} \Big[2 \sqrt{(\alpha n + \nu_0) x} + 2x + \nu_0 \Big]
		&\le& K_2 \left( \frac{\log p}{n} \right)^{1/2}
		\eea
		for all sufficiently large $n$ and some large constant $K_2>0$.
		Thus, we have
		\bea
		&&\bbE_0 \left[\pi_\alpha \left( \|D_n^{-1} - \what{D}_n^{-1} \| \ge K_{2}  \left(\frac{\log p}{n} \right)^{1/2}  \,\,\Big|\,\, \bfX_n \right)\right]\\
		&\le& \bbE_0 \left[\pi_\alpha \left( \max_j |d_j^{-1} - \what{d}_{S_{0j}}^{-1} | \ge K_{2}  \left(\frac{\log p}{n} \right)^{1/2}  \,\,\Big|\,\, \bfX_n \right) I_{N_{S_0,\epsilon_0}^c} \right]
		+ \bbP_0 (N_{S_0,\epsilon_0}) \\
		&\le& 2p \cdot \exp(-C \log p) + o(1) \,\,=\,\, o(1). 
		\eea
	\end{proof}

	\begin{proof}[Proof of Theorem 3.6]
		Let $\Omega_{0n} \in \calU_p^*$, $\epsilon_n = s_0^{3/4}\sqrt{(s_0 + \log p)/n}$ and assume $s_0^{3/2}(s_0 + \log p) = o(n)$. 
		Consider the spectral norm case first.
		Then,
		\bea
		&& \bbE_0 \left[ \pi_\alpha \big( \|\Omega_n -\Omega_{0n}\| \ge K_{\rm conv} \epsilon_n   \mid \bfX_n\big)  \right] \\
		&\le& \bbE_0 \left[ \pi_\alpha \Big( \|\Omega_n -\what{\Omega}_{n}\| \ge \frac{K_{\rm conv}}{2} \epsilon_n   \mid \bfX_n\Big)   \right] 
		+ \bbP_0 \Big( \|\what{\Omega}_n -\Omega_{0n}\| \ge \frac{K_{\rm conv}}{2} \epsilon_n \Big)\\
		&=& \bbE_0 \left[ \pi_\alpha \Big( \|\Omega_n -\what{\Omega}_{n}\| \ge \frac{K_{\rm conv}}{2} \epsilon_n   \mid \bfX_n\Big)  \right] 
		+  o(1)
		\eea
		for some large constant $K_{\rm conv}>0$ by Lemma \ref{lemma:Nsets4} and Lemma \ref{lemma:freq_conv}, so it suffices to prove
		\bea
		\bbE_0 \left[ \pi_\alpha \Big( \|\Omega_n -\what{\Omega}_{n}\| \ge \frac{K_{\rm conv}}{2} \epsilon_n   \mid \bfX_n\Big)  \right] 
		&=& o(1).
		\eea
		By applying \eqref{whatO_diff_O}, Lemma \ref{lemma:An_post} and Lemma \ref{lemma:Dn_post}, it is easy to prove the above result for some large constant $K_{\rm conv}>0$. 
		
		Let $\epsilon_n^* = \|I_p - A_{0n}\|_\infty s_0 \sqrt{(s_0 + \log p)/n}$ and assume $s_0 (s_0 + \log p) = o(n)$.
		Note that
		\bea
		&& \bbE_0 \left[ \pi_\alpha \big( \|\Omega_n -\Omega_{0n}\|_\infty \ge K_{\rm conv} \epsilon_n^*   \mid \bfX_n\big)  \right] \\
		&\le& \bbE_0 \left[ \pi_\alpha \Big( \|\Omega_n -\what{\Omega}_{n}\|_\infty \ge \frac{K_{\rm conv}}{2} \epsilon_n^*   \mid \bfX_n\Big)   \right] 
		+ \bbP_0 \Big( \|\what{\Omega}_n -\Omega_{0n}\|_\infty \ge \frac{K_{\rm conv}}{2} \epsilon_n^* \Big)\\
		&=& \bbE_0 \left[ \pi_\alpha \Big( \|\Omega_n -\what{\Omega}_{n}\|_\infty \ge \frac{K_{\rm conv}}{2} \epsilon_n^*   \mid \bfX_n\Big)  \right] 
		+  o(1)
		\eea
		for some large constant $K_{\rm conv}>0$ by Lemma \ref{lemma:Nsets4} and Lemma \ref{lemma:freq_conv}.
		Again, the last display is of order $o(1)$ by \eqref{whatO_diff_01}, Lemma \ref{lemma:An_post} and Lemma \ref{lemma:Dn_post}.
	\end{proof}

	\section{Proofs of Posterior Convergence Rates for Cholesky Factors}\label{sec:proof_conv_chol}

	\begin{proof}[Proof of Theorem 3.2]
		Consider $\Omega_{0n} \in \calU_p$ and assume $s_0\log p =o(n)$.
		Then,
		\bean
		&& \bbE_0 \left[ \pi_\alpha \Big( \|A_n - A_{0n}\|_\infty \ge K_{\rm chol}\sqrt{s_0} \Big( \frac{s_0 + \log p}{n} \Big)^{1/2}   \mid \bfX_n \Big)  \right] \nonumber \\
		&\le& \bbE_0 \left[ \pi_\alpha \Big( \|A_n - \what{A}_{n}\|_\infty \ge \frac{K_{\rm chol}}{2}\sqrt{s_0} \Big( \frac{s_0 + \log p}{n} \Big)^{1/2}   \mid \bfX_n \Big)  \right] \label{A_inf_1} \\
		&+& \bbP_0 \Big( \| \what{A}_n - A_{0n} \|_\infty \ge  \frac{K_{\rm chol}}{2}\sqrt{s_0} \Big( \frac{s_0 + \log p}{n} \Big)^{1/2}  \Big). \label{A_inf_2}
		\eean
		Note that \eqref{A_inf_1} is of order $o(1)$ for some constant $K_{\rm chol} >0$ by Lemma \ref{lemma:An_post}.
		On the other hand, \eqref{A_inf_2} is also of order $o(1)$ for some constant $K_{\rm chol} >0$ by Lemma \ref{lemma:Nsets4} and \eqref{Ahat_diff_0_infty}.
		Note that $\|A_n - A_{0n}\|_F^2 = \sum_{j=2}^p \|a_j - a_{0j} \|_2^2 $ and
		\bea
		&& \bbE_0 \left[ \pi_\alpha \Big(  \sum_{j=2}^p \|a_j - a_{0j}\|_2^2 \ge K_{\rm chol} \sum_{j=2}^p \Big( \frac{s_{0j} + \log j}{n} \Big)   \mid \bfX_n \Big)  \right] \\
		&\le&  \sum_{j=2}^p \bbE_0 \left[ \pi_\alpha \Big(  \|a_j - \what{a}_{j}\|_2^2 \ge \frac{K_{\rm chol}}{2}\Big( \frac{s_{0j} + \log j}{n} \Big)   \mid \bfX_n \Big)  \right]  \\
		&+& \sum_{j=2}^p \bbP_0 \Big(  \| \what{a}_j - a_{0j} \|_2^2 \ge  \frac{K_{\rm chol}}{2} \Big( \frac{s_{0j} + \log j}{n} \Big)  \Big),
		\eea
		where the last displays are of order $o(1)$ for some constant $K_{\rm chol}>0$. 
		It is easy to check from the slight modifications of Lemma \ref{lemma:Nsets4}, Lemma \ref{lemma:An_post} and the proof of Lemma \ref{lemma:freq_conv}, by using $s_{0j}$ and $\log j$ instead of $s_0$ and $\log p$. 
	\end{proof}

	\begin{lemma}\label{lemma:denom}
		For a given constant $0<\epsilon_0< 1/2$ and an integer $s_0$, assume model \eqref{model} and the MESC prior with Condition \hyperref[condP]{\rm(P)} and $\nu_0 =O(1)$.
		Let $B_n = \Big\{ \bfX_n : \big| \|\tilde{X}_j \|_2^2 - \|\bfX_{S_{0j}} \what{a}_{S_{0j}}\|_2^2  - n\, d_{0j} \big| \ge [ j^2 \log n ]^{-1} \Big\}$.
		For given constants $0<\alpha<1$, $M_1 \le (1-2\epsilon_0)^4\epsilon_0$, $M_2\ge (1+2\epsilon_0)^4 \epsilon_0^{-1}$ and $2 \le j \le p$, 
		\bea
		D_{nj} &:=& \sum_{S_j: 0< |S_j|\le R_j} \int_{M_1}^{M_2} \int R_{nj}(a_j, d_j)^\alpha \pi(a_{S_j} \mid d_j, S_j) \pi_j(S_j) \pi(d_j) d a_{S_j} \delta_0(d a_{j,S_j^c}) d d_j \\
		&\ge&  \pi_j (S_{0j}) \cdot \left( 1+ \frac{\alpha}{\gamma} \right)^{-\frac{s_{0j}}{2}} \frac{C_{\rm den}}{n j^2 \log n} \quad \quad \text{ on the event } B_n,
		\eea
		for some constant $C_{\rm den} = C_{\rm den}(\epsilon_0, \nu_0, \nu_0')>0$ for any $\Omega_{0n}\in \calU_p^0$.
	\end{lemma}
	\begin{proof}
		Note that
		\bea
		&& D_{nj}  \\
		&\ge& \pi_j (S_{0j}) \int_{M_1}^{M_2} \int R_{nj}(a_j, d_j)^\alpha \pi(a_{S_{0j}} \mid d_j ,S_{0j})  \pi(d_j) d a_{S_{0j}} \delta_0(d a_{j,S_{0j}^c}) d d_j  \\
		&=& \pi_j (S_{0j}) \int_{M_1}^{M_2} \int \Big( \frac{d_{0j}}{d_{j}} \Big)^{ \frac{\alpha n}{2}} e^{-\frac{\alpha}{2} \big\{ d_j^{-1}\|\tilde{X}_j - \bfX_{S_{0j}}a_{S_{0j}}\|_2^2 - d_{0j}^{-1} \|\tilde{X}_j - \bfX_{S_{0j}} a_{0,S_{0j}}\|_2^2 \big\} } \\
		&& \times \,\, \det \Big[ 2\pi \frac{d_j}{\gamma} (\bfX_{S_{0j}}^T \bfX_{S_{0j}} )^{-1} \Big]^{-\frac{1}{2}} e^{-\frac{1}{2}(a_{S_{0j}} - \what{a}_{S_{0j}} )^T \frac{\gamma}{d_j} \bfX_{S_{0j}}^T \bfX_{S_{0j}}  (a_{S_{0j}} - \what{a}_{S_{0j}} ) } \pi(d_j) d a_{S_{0j}}  d d_j \\
		&=& \pi_j (S_{0j}) \left( 1+ \frac{\alpha}{\gamma} \right)^{-\frac{s_{0j}}{2}}  e^{ \frac{\alpha}{2 d_{0j}} \|\bfX_{S_{0j}}(\what{a}_{S_{0j}} - a_{0,S_{0j}} )\|_2^2  } \\
		&& \times \,\, \int_{M_1}^{M_2}   \Big( \frac{d_{0j}}{d_{j}} \Big)^{ \frac{\alpha n}{2}} e^{-\frac{\alpha}{2} (d_j^{-1} - d_{0j}^{-1}) \big\{ \|\tilde{X}_j\|_2^2 - \|\bfX_{S_{0j}} \what{a}_{S_{0j}}\|_2^2 \big\}   } \pi(d_j) d d_j \\
		&\ge& \pi_j (S_{0j}) \left( 1+ \frac{\alpha}{\gamma} \right)^{-\frac{s_{0j}}{2}} \int_{M_1}^{M_2}   \Big( \frac{d_{0j}}{d_{j}} \Big)^{ \frac{\alpha n}{2}} e^{-\frac{\alpha}{2} (d_j^{-1} - d_{0j}^{-1}) \big\{ \|\tilde{X}_j\|_2^2 - \|\bfX_{S_{0j}} \what{a}_{S_{0j}}\|_2^2 \big\}  } \pi(d_j) d d_j.
		\eea
		The second equality follows from the integration of the multivariate normal distribution.
		Denote $IG(x \mid a,b)$ as the density function of $IG(a,b)$ at $x$, then
		\bean
		&& \int_{M_1}^{M_2}   \Big( \frac{d_{0j}}{d_{j}} \Big)^{ \frac{\alpha n}{2}} e^{-\frac{\alpha}{2} (d_j^{-1} - d_{0j}^{-1}) \big\{ \|\tilde{X}_j\|_2^2 - \|\bfX_{S_{0j}} \what{a}_{S_{0j}}\|_2^2 \big\} } \pi(d_j) d d_j  \nonumber \\
		&=& \int_{M_1}^{M_2}  \frac{IG \big(d_j \mid \alpha n/2+1, \alpha/2 \big\{ \|\tilde{X}_j\|_2^2 - \|\bfX_{S_{0j}} \what{a}_{S_{0j}}\|_2^2 \big\} \big) }{IG \big(d_{0j} \mid \alpha n/2+1, \alpha/2 \big\{ \|\tilde{X}_j\|_2^2 - \|\bfX_{S_{0j}} \what{a}_{S_{0j}}\|_2^2 \big\} \big)} \cdot \pi(d_j) d d_j. \label{ratio_IG}
		\eean
		Note that the ratio of the inverse-gamma density functions is larger than 1 if $d_j$ is between $d_{0j}$ and $\big\{ \|\tilde{X}_j\|_2^2 - \|\bfX_{S_{0j}} \what{a}_{S_{0j}}\|_2^2 \big\}/n$.
		Since we focus on the event $ \Big\{ \bfX_n : \big| \|\tilde{X}_j \|_2^2 - \|\bfX_{S_{0j}} \what{a}_{S_{0j}}\|_2^2  - n\, d_{0j} \big| \ge [ j^2 \log n ]^{-1} \Big\}$, if $d_{0j} \ge \big\{ \|\tilde{X}_j\|_2^2 - \|\bfX_{S_{0j}} \what{a}_{S_{0j}}\|_2^2 \big\}/n$, \eqref{ratio_IG} is bounded below by
		\bea
		\int_{M_1 \vee \big\{ \|\tilde{X}_j\|_2^2 - \|\bfX_{S_{0j}} \what{a}_{S_{0j}}\|_2^2 \big\}/n}^{d_{0j}} \pi(d_j) d d_j 
		&\ge& \frac{C_{\rm den}}{n j^2 \log n}
		\eea
		for some constant $C_{\rm den} = C_{\rm den}(\epsilon_0, \nu_0, \nu_0')>0$, because we assume $M_1 \le (1-2\epsilon_0)^4\epsilon_0$.
		Similarly, if $d_{0j} \le \big\{ \|\tilde{X}_j\|_2^2 - \|\bfX_{S_{0j}} \what{a}_{S_{0j}}\|_2^2 \big\}/n$, \eqref{ratio_IG} is bounded below by $C_{\rm den}/(n j^2 \log n)$ for some constant $C_{\rm den}>0$.
	\end{proof}

	\begin{lemma}\label{lemma:effect_dim}
		For a given constant $0<\epsilon_0< 1/2$ and an integer $s_0$, assume model \eqref{model} and the MESC prior with Condition \hyperref[condP]{\rm(P)} and $\nu_0 = O(1)$.
		If $s_0 \log p = o(n)$ and $0<\alpha<1$, then
		\bea
		\sum_{j=2}^p \bbE_0 \pi_\alpha \big(S_j \ge C_{\rm dim} s_{0} \mid \bfX_n \big) &=& o(1)
		\eea
		for some constant $C_{\rm dim}>0$ and for any $\Omega_{0n}\in \calU_p^0$.
	\end{lemma}
	\begin{proof}
		Let $B_n = \Big\{ \bfX_n : \big| \|\tilde{X}_j \|_2^2 - \|\bfX_{S_{0j}} \what{a}_{S_{0j}}\|_2^2  - n\, d_{0j} \big| \ge [ j^2 \log n ]^{-1} \Big\}$ as defined at Lemma \ref{lemma:denom} and
		\bea
		&& N_{nj}(S_j \ge C_{\rm dim} s_0)  \\
		&=& \sum_{S_j: |S_j| \ge C_{\rm dim} s_0 } \int_{M_1}^{M_2} \int R_{nj}(a_j, d_j)^\alpha \pi(a_{S_j} \mid d_j, S_j) \pi_j(S_j) \pi(d_j) d a_{S_j}\delta_0(d a_{j,S_j^c}) d d_j .
		\eea
		By Lemma \ref{lemma:d_supp}, we have
		\bea
		&& \sum_{j=2}^p \bbE_0 \pi_\alpha \big(S_j \ge C_{\rm dim} s_{0} \mid \bfX_n \big) \\
		&\le& \sum_{j=2}^p \bbE_0 \pi_\alpha \big(S_j \ge C_{\rm dim} s_{0} ,\,\, M_1 \le d_j \le M_2 \mid \bfX_n \big) \,\,+ \,\, o(1)
		\eea
		for some constants $M_1 \le (1-2\epsilon_0)^4\epsilon_0$, $M_2 \ge (1+2\epsilon_0)^4\epsilon_0$ and $C_{\alpha, \epsilon_0} >0$ because we assume $s_0\log p =o(n)$.
		In fact, Lemma \ref{lemma:d_supp} assumes the ESC prior, but it is easy to show that it also holds for the MESC prior for some constant $\nu_0'>0$.
		For any $2\le j \le p$, we have
		\bean
		&& \bbE_0 \pi_\alpha \big(S_j \ge C_{\rm dim} s_{0} ,\,\, M_1 \le d_j \le M_2 \mid \bfX_n \big) \nonumber  \\
		&\le& \bbE_0 \big[ N_{nj}(S_j \ge C_{\rm dim} s_0) \big] \frac{e^{C_1 s_{0j}} }{\pi_j(S_{0j})} C_2 n j^2 \log n \label{Sj_dim_1st}  \\
		&& +\,\,  \bbP_0 \Big(  \big| \|\tilde{X}_j \|_2^2 - \|\bfX_{S_{0j}} \what{a}_{S_{0j}}\|_2^2  - n \, d_{0j} \big| \le \frac{1}{ j^2 \log n } \Big) \nonumber
		\eean
		for some positive constants $C_1$ and $C_2$ by Lemma \ref{lemma:denom}.
		Since $d_{0j}^{-1} \big\{ \|\tilde{X}_j \|_2^2 - \|\bfX_{S_{0j}} \what{a}_{S_{0j}}\|_2^2 \big\} \sim \chi_{n - s_{0j}}^2$ under $\bbP_0$ given $\tilde{Z}_j$, 
		\bea
		\sum_{j=2}^p \bbP_0 \Big(  \big| \|\tilde{X}_j \|_2^2 - \|\bfX_{S_{0j}} \what{a}_{S_{0j}}\|_2^2  - n \, d_{0j} \big| \le \frac{1}{ j^2 \log n } \Big)  
		&\le&  \sum_{j=2}^p  \frac{1}{d_{0j}\, j^2 \log n } \,\,=\,\, o(1), 
		\eea
		so we only need to focus on \eqref{Sj_dim_1st}.
		Note that 
		\bea
		&& \bbE_0 \big[ N_{nj}(S_j \ge C_{\rm dim} s_0) \,\big|\, \tilde{Z}_j \big] \\
		&=& \int_{M_1}^{M_2} \int \sum_{S_j: |S_j| \ge C_{\rm dim} s_0 } \pi_j(S_j) \cdot  \bbE_0 \left[ R_{nj}(a_{j}, d_j)^{\alpha} \pi(a_{S_j} \mid d_j, S_j) \,\big|\, \tilde{Z}_j  \right] \pi(d_j) d a_{S_j} \delta_0(d a_{j,S_j^c}) dd_j .
		\eea
		Let $a_{S_j +}$ be a $p$-dimensional vector such that $(a_{S_j+})_{S_j} = a_{S_j}$ and $(a_{S_j+})_l = 0$ for all $l \notin S_j$.
		It is easy to see that
		\bea
		&& \bbE_0 \left[ R_{nj}(a_{S_j+}, d_j)^{\alpha} \pi(a_{S_j} \mid d_j, S_j) \,\big|\, \tilde{Z}_j \right]  \\
		&\le& \left\{ \bbE_0 \big[ R_{nj}(a_{S_j+}, d_j)^{h_1 \alpha} \,\big|\, \tilde{Z}_j  \big] \right\}^{\frac{1}{h_1}} \times \left\{ \bbE_0 \big[ \pi(a_{S_j} \mid d_j, S_j)^{h_2} \,\big|\, \tilde{Z}_j  \big] \right\}^{\frac{1}{h_2}} \\
		&\le& \exp \left( - \frac{h_1\alpha(1 - h_1\alpha)}{2 (h_1\alpha d_{0j} + (1-h_1\alpha)d_j )} \| \tilde{Z}_j (a_{S_j +} - a_{0}) \|_2^2 \right) \\
		&& \times \,\, \exp \left( - \frac{n}{2 h_1} \log \Big[ \frac{h_1\alpha d_{0j} + (1-h_1\alpha)d_j}{d_{0j}^{h_1\alpha} d_j^{1-h_1\alpha} } \Big]  \right) \times  \left\{ \bbE_0 \big[ \pi(a_{S_j} \mid d_j, S_j)^{h_2}  \big] \right\}^{\frac{1}{h_2}}
		\eea
		for any constants $h_1, h_2 >1$ such that $h_1^{-1}+ h_2^{-1} =1$ and $h_1 \alpha <1$.
		The first inequality follows from the H\"{o}lder's inequality, and the second inequality follows from the R\'{e}nyi divergence between two multivariate normal distributions (\cite{hero2001alpha}).
		Note that the first term in the last display is bounded above by
		\bea
		\exp \left( - \frac{h_1\alpha(1 - h_1\alpha)}{2 ( d_{0j} + M_2 )} \|\tilde{Z}_j (a_{S_j +} - a_{0}) \|_2^2 \right) 
		&\le& 1
		\eea
		for any $M_1\le d_j \le M_2$.
		Also note that the second term in the last display is bounded above by 1 because 
		\bea
		\log \big( h_1\alpha d_{0j} + (1-h_1\alpha)d_j \big) \ge h_1\alpha \log d_{0j} + (1-h_1\alpha) \log d_j
		\eea
		by the Jensen's inequality.
		For the last term, we have
		\bea
		&& \left\{ \bbE_0 \big[ \pi(a_{S_j} \mid d_j, S_j)^{h_2}  \,\big|\, \tilde{Z}_j \big] \right\}^{\frac{1}{h_2}}  \\
		&=&  \det \Big[ 2\pi \frac{d_j}{\gamma} (\bfX_{S_j}^T \bfX_{S_j})^{-1} \Big]^{-\frac{1}{2}} \cdot  
		\left\{ \bbE_0 \Big( \exp \big\{ - \frac{\gamma h_2}{2d_j} \| \bfX_{S_j} (a_{S_j}- \what{a}_{S_j}) \|_2^2 \big\} \,\big|\, \tilde{Z}_j \Big) \right\}^{\frac{1}{h_2}}  \\
		&=&  \det \Big[ 2\pi \frac{d_j}{\gamma} (\bfX_{S_j}^T \bfX_{S_j})^{-1} \Big]^{-\frac{1}{2}} \times \Big(1 + \frac{\gamma}{d_j} h_2 d_{0j} \Big)^{-\frac{|S_j|}{2h_2}}  \\
		&& \times \,\, \exp \left(- \frac{\gamma}{ 2(d_j + \gamma h_2 d_{0j})} \| \bfX_{S_j}(a_{S_j} - [\bfX_{S_j}^T\bfX_{S_j}]^{-1}\bfX_{S_j}^T \tilde{Z}_j a_{0j} ) \|_2^2  \right)  \\
		&=& \left(1 + \frac{\gamma}{d_j} h_2 d_{0j} \right)^{ \frac{1}{2}(1 - h_2^{-1}) |S_j| } \cdot  N_{|S_j|} \Big(a_{S_j} \,\big|\, [\bfX_{S_j}^T\bfX_{S_j}]^{-1}\bfX_{S_j}^T \tilde{Z}_j a_{0j},\, \big( \frac{d_j}{\gamma} + h_2 d_{0j}\big) [\bfX_{S_j}^T\bfX_{S_j}]^{-1}  \Big) .
		\eea
		The second equality follows from the moment generating function of the noncentral chi-square distribution because $d_{0j}^{-1}\| \bfX_{S_j}(\what{a}_{S_j} - a_{S_j}) \|_2^2$ is the noncentral chi-square random variable with $|S_j|$ degrees of freedom and the noncentrality parameter $\| \bfX_{S_j}(a_{S_j} - [\bfX_{S_j}^T\bfX_{S_j}]^{-1}\bfX_{S_j}^T \tilde{Z}_j a_{0j} ) \|_2^2$ under $\bbP_0$ given $\tilde{Z}_j$.
		Thus, 
		\bea
		&& \sum_{j=2}^p \bbE_0 \pi_\alpha \big(S_j \ge C_{\rm dim} s_{0} \mid \bfX_n \big) \\
		&\le& \sum_{j=2}^p \bbE_0 \big[ N_{nj}(S_j \ge C_{\rm dim} s_0) \big]\cdot \frac{e^{C_1 s_{0j}} }{\pi_j(S_{0j})} C_2 n j^2 \log n \,\,+\,\, o(1) \\
		&\le& \sum_{j=2}^p \sum_{S_j: |S_j| \ge C_{\rm dim} s_0 } \pi_j(S_j)  \cdot \left(1 + \frac{\gamma}{M_1} h_2 d_{0j} \right)^{ \frac{1}{2}(1 - h_2^{-1}) |S_j| } \cdot \frac{e^{C_1 s_{0j}} }{\pi_j(S_{0j})} C_2 n j^2 \log n \,\,+\,\, o(1) \\
		&\le& \sum_{j=2}^p \, \frac{e^{C_1 s_{0j}} }{\pi_j(S_{0j})} C_2 n j^2 \log n \sum_{S_j: |S_j| \ge C_{\rm dim} s_0 } e^{ C_3 |S_j| } \pi_j(S_j) \,\,+\,\, o(1) \\
		&\le& \sum_{j=2}^p \, e^{ C_1 s_{0j} + C_4 s_{0j}\log j + 4 \log(n\vee j)  } \cdot \left(\frac{e^{C_3}}{c_1 p^{c_2}} \right)^{C_{\rm dim}s_0} \,\,+\,\, o(1)
		\eea
		for some positive constants $C_3$ and $C_4$.
		The last term is of order $o(1)$ for some large constant $C_{\rm dim}>0$. 	
	\end{proof}

	\begin{lemma}\label{lemma:numer}
		For a given constant $0<\epsilon_0< 1/2$ and an integer $s_0$, assume model \eqref{model} and the MESC prior with Condition \hyperref[condP]{\rm(P)}.
		For given constants $0<\alpha<1$, $M_1 \le (1-2\epsilon_0)^4\epsilon_0$, $M_2\ge (1+2\epsilon_0)^4 \epsilon_0^{-1}$ and $2 \le j \le p$, define $\delta_{n}' = \sqrt{s_{0} \log p/ \Psi_{\min}(C_2 s_{0})^2 }$,
		\bea
		B_{nj}(C_1) &=& \big\{ a_j : \| a_j - a_{0j}\|_2^2 \ge C_1 \delta_{n}'^2  \, \big\} \quad \text{ and} \\
		N_{nj} &=& \sum_{S_j: 0< |S_j|\le C_3 s_0} \int_{M_1}^{M_2} \int_{B_{nj}(C_1)} R_{nj}(a_j, d_j)^\alpha \pi(a_{S_j} \mid d_j, S_j) \pi_j(S_j) \pi(d_j) d a_{S_j} \delta_0(d a_{j,S_j^c}) d d_j ,
		\eea
		for some positive constants $C_1$, $C_2>1$ and $C_3 = C_2 -1 >0$.
		Then, we have
		\bea
		\bbE_0 \big( N_{nj} \big)
		&\le& e^{- C_{\rm num,1} \cdot C_1 s_{0}\log p } \sum_{S_j: 0<|S_j|\le R_j } C_{\rm num,2}^{|S_j|} \pi_j(S_j)
		\eea
		for some positive constants $C_{\rm num,1} = C_{\rm num,1}(M_2, \epsilon_0, \alpha)$ and $C_{\rm num,2} = C_{\rm num,2}(M_1, \epsilon_0, \gamma)$ for any $\Omega_{0n} \in \calU_p^0$.
	\end{lemma}
	\begin{proof}
		By the definition of $\Psi_{\min}$, it is easy to see that
		\bea
		\| \tilde{Z}_j (a_j - a_{0j})\|_2^2 &\ge&
		\Psi_{\min}(|S_{a_j - a_{0j}}|)^2 \| a_j - a_{0j}\|_2^2.
		\eea
		Note that $\Psi_{\min}(|S_{a_j - a_{0j}}|) \ge \Psi_{\min}(|S_{a_j}|+ |S_{a_{0j}}|) \ge \Psi_{\min}(C_2 s_0)$ for any $S_j$ such that $0<|S_j|\le C_3 s_0$.
		Thus, it suffices to prove the result with respect to the set $B_{nj}'(C_1) := \big\{ a_j : \| \tilde{Z}_j (a_j - a_{0j})\|_2^2 \ge C_1 s_0 \log p \big\}$ instead of $B_{nj}(C_1)$.
		Let $a_{S_j +}$ be a $p$-dimensional vector such that $(a_{S_j+})_{S_j} = a_{S_j}$ and $(a_{S_j+})_l = 0$ for all $l \notin S_j$.
		The rest part of the proof is straightforward from the proof of Lemma \ref{lemma:effect_dim}.
		We have
		\bea
		&& \bbE_0 \big( N_{nj} \big)  \\
		&\le& \bbE_0 \Bigg[  \int_{M_1}^{M_2} \int_{B_{nj}'(C_1)} \sum_{S_j: 0<|S_j|\le R_j }     \exp \left( - \frac{h_1\alpha(1 - h_1\alpha)}{2 ( d_{0j} + M_2 )} \|\tilde{Z}_j (a_{S_j +} - a_{0}) \|_2^2 \right)  \\
		&& \quad\times \,\, N_{|S_j|} \Big(a_{S_j}\mid [\bfX_{S_j}^T\bfX_{S_j}]^{-1}\bfX_{S_j}^T \tilde{Z}_j a_{0j},\, \big( \frac{d_j}{\gamma} + h_2 d_{0j}\big) [\bfX_{S_j}^T\bfX_{S_j}]^{-1}  \Big)   \\
		&& \quad\times \,\,  \left(1 + \frac{\gamma}{M_1} h_2 d_{0j} \right)^{ \frac{1}{2}(1 - h_2^{-1}) |S_j| } \pi_j(S_j) \pi(d_j) d a_{S_j} d d_j  \Bigg] \\
		&\le& e^{- C_{\rm num,1} \cdot C_1 s_{0}\log p } \sum_{S_j: 0<|S_j|\le R_j } C_{\rm num,2}^{|S_j|} \pi_j(S_j),
		\eea
		where $C_{\rm num,1} = h_1\alpha(1-h_1 \alpha)/(2 [\epsilon_0^{-1}+ M_2 ] )$ and $C_{\rm num,2} = ( 1 + \gamma h_2 \epsilon_0^{-1}/M_1)^{2^{-1}(1-h_2^{-1})}$ for any constants $h_1, h_2 >1$ such that $h_1 \alpha <1$ and $h_1^{-1}+h_2^{-1}=1$.
		The second inequality holds because for any $\Omega_{0n}\in \calU_p^0$, we have $d_{0j} \le \epsilon_0^{-1}$. 
	\end{proof}

	\begin{proof}[Proof of Theorem 3.4]
		For some constant $K_{\rm chol}' >0$, let $\delta_n = K_{\rm chol}'  \sqrt{s_0\log p / n }$.
		Note that
		\bea
		&& \bbE_0 \pi_\alpha \left( \|A_n - A_{0n}\|_\infty \ge \sqrt{s_0} \delta_n  \,\,\big|\,\, \bfX_n \right)  \\
		&=& \bbE_0 \pi_\alpha \left( \max_{2\le j\le p} \| a_j - a_{0j}\|_1 \ge \sqrt{s_0} \delta_n  \,\,\big|\,\, \bfX_n \right) \\
		&\le& \sum_{j=2}^p \bbE_0 \pi_\alpha \left(  \| a_j - a_{0j}\|_1 \ge \sqrt{s_0} \delta_n  \,\,\big|\,\, \bfX_n \right) \\
		&\le& \sum_{j=2}^p \bbE_0 \pi_\alpha \left(  \| a_j - a_{0j}\|_1 \ge \sqrt{s_{0}} \delta_n ,\, S_j \le C_1 s_{0}  \,\,\big|\,\, \bfX_n \right) + \sum_{j=2}^p \bbE_0 \pi_\alpha \big(S_j \ge C_1 s_{0} \mid \bfX_n \big)
		\eea
		for some constant $C_1>0$.
		The second term in the last display is of order $o(1)$ by Lemma \ref{lemma:effect_dim}.
		Also note that 
		\bea
		&& \bbE_0 \pi_\alpha \Big(  \| a_j - a_{0j}\|_1 \ge \sqrt{s_{0}} \delta_n ,\, S_j \le C_1 s_{0}  \,\,\big|\,\, \bfX_n \Big) \\
		&\le& \bbE_0 \pi_\alpha \Big(  \| a_j - a_{0j}\|_2 \ge (C_1+1)^{-1} \delta_n ,\, S_j \le C_1 s_{0}  \,\,\big|\,\, \bfX_n \Big)  \\
		&\le& \bbE_0 \pi_\alpha \left(  \| a_j - a_{0j}\|_2 \ge C_2 K_{\rm chol}' \Big( \frac{s_{0} \log p}{\Psi_{\min}(C_2 s_{0})^2 } \Big)^{1/2} ,\, S_j \le C_1 s_{0} \,\,\big|\,\, \bfX_n \right)
		\\
		&& +\,\, \bbP_0 \Big( n^{-1}\Psi_{\min}(C_2 s_{0})^2 \le C_{\min} \epsilon_0  \Big),
		\eea
		where $C_2 = (C_1+1)^{-1} \sqrt{C_{\min} \epsilon_0}$ and $C_{\min} = (1-2\epsilon_0)^4 (1-\epsilon_0)$.
		Since we assume $s_0 = o(n)$,
		\bea
		\sum_{j=2}^p \bbP_0 \Big( n^{-1}\Psi_{\min}(C_2 s_{0})^2 \le C_{\min} \epsilon_0  \Big) &=& o(1)
		\eea
		by Lemma \ref{lemma:Nsets}.
		Let $\delta_{n}' = \sqrt{s_{0} \log p/ \Psi_{\min}(C_2 s_{0})^2 }$, then by Lemma \ref{lemma:d_supp},
		\bea
		&& \sum_{j=2}^p \bbE_0 \pi_\alpha \left(  \| a_j - a_{0j}\|_2 \ge C_2 K_{\rm chol}' \delta_{n}' ,\, S_j \le C_1 s_{0} \,\,\big|\,\, \bfX_n \right) \\
		&\le& \sum_{j=2}^p \bbE_0 \pi_\alpha \left(  \| a_j - a_{0j}\|_2 \ge C_2 K_{\rm chol}' \delta_{n}',\, S_j \le C_1 s_{0} ,\,\, M_1 \le d_j \le M_2 \,\,\big|\,\, \bfX_n \right) \,\,+\,\, o(1)
		\eea
		for some constants $M_1 \le (1-2\epsilon_0)^4\epsilon_0$ and $M_2\ge (1+2\epsilon_0)^4 \epsilon_0^{-1}$ because we assume that $s_0 \log p =o(n)$.
		In fact, Lemma \ref{lemma:d_supp} assumes the ESC prior, but it is easy to show that it also holds for the MESC prior for some constant $\nu_0'>0$.
		Let $a_{S_j +}$ be a $p$-dimensional vector such that $(a_{S_j+})_{S_j} = a_{S_j}$ and $(a_{S_j+})_l = 0$ for all $l \notin S_j$.
		By Lemma \ref{lemma:denom}, we have
		\begin{align}
		& \pi_\alpha \left(  \| a_j - a_{0j}\|_2 \ge C_2 K_{\rm chol}' \delta_{n}', \,\, M_1 \le d_j \le M_2 \,\,\big|\,\, \bfX_n \right) \nonumber  \\
		&= \frac{\sum_{S_j: 0< |S_j|\le C_1s_0} \int_{M_1}^{M_2} \int_{\|a_j- a_{0j}\|_2 \ge C_2 K_{\rm chol}' \delta_{n}' } R_{nj}(a_{S_j +}, d_j)^\alpha \pi(a_{S_j} \mid d_j, S_j) \pi_j(S_j) \pi(d_j) d a_{S_j} d d_j}{\sum_{S_j: 0< |S_j|\le R_j} \int_{M_1}^{M_2} \int R_{nj}(a_{S_j +}, d_j)^\alpha \pi(a_{S_j} \mid d_j, S_j) \pi_j(S_j) \pi(d_j) d a_{S_j} d d_j  }  \nonumber \\
		&=: \frac{N_{nj}}{D_{nj}} \nonumber \\
		&\le N_{nj} \cdot \frac{e^{C_3 s_{0j}}}{\pi_j(S_{0j})} \cdot C_4 n j^2 \log n \label{aj_l2_con_Nnj} \\
		&+ I \Big( \big| \|\tilde{X}_j \|_2^2 - \|\bfX_{S_{0j}} \what{a}_{S_{0j}}\|_2^2  - n \, d_{0j} \big| \le [ j^{2} \log n ]^{-1}  \Big)  \nonumber
		\end{align}
		for some positive constants $C_3$ and $C_4$.
		Note that 
		\bea
		\sum_{j=2}^p \bbP_0 \Big(  \big| \|\tilde{X}_j \|_2^2 - \|\bfX_{S_{0j}} \what{a}_{S_{0j}}\|_2^2  - n \, d_{0j} \big| \le \frac{1}{ j^2 \log n } \Big) 
		&\le&  \sum_{j=2}^p  \frac{1}{d_{0j}\, j^2 \log n } \,\,=\,\, o(1),
		\eea
		so we only need to focus on \eqref{aj_l2_con_Nnj}.
		By Lemma \ref{lemma:numer},
		\bea
		\frac{e^{C_3 s_{0j}}}{\pi_j(S_{0j})} \cdot C_4 n j^2 \log n \cdot \bbE_0 \big( N_{nj} \big) 
		&\le& \frac{e^{C_3 s_{0j}}}{\pi_j(S_{0j})} \cdot C_4 n j^2 \log n \cdot e^{- C_5 K_{\rm chol}' s_{0}\log p } \sum_{S_j: 0< |S_j|\le R_j} C_6^{|S_j|} \pi_j(S_j) \\
		&\le& e^{C_3 s_{0j} + C_7 s_{0j}\log j + 4\log (n \vee j) - C_5 K_{\rm chol}' s_{0}\log p }
		\eea
		for some positive constants $C_5$, $C_6$ and $C_7$.
		The summation of last display with respect to all $j$ is of order $o(1)$ for some large $K_{\rm chol}'$.
		Thus, we have proved 
		\bea
		\bbE_0 \pi_\alpha \left( \|A_n - A_{0n}\|_\infty \ge \sqrt{s_0} \delta_n  \,\,\big|\,\, \bfX_n \right)
		&=& o(1)
		\eea
		for some large constant $K_{\rm chol}'>0$.
	\end{proof}

	\section{Proofs of Minimax Lower Bounds}\label{sec:proof_minimax_Lbounds}
	
	\begin{proof}[Proof of Theorem 3.3]
		Note that 
		\bea
		\|\what{A}_n - A_{0n} \|_\infty
		&=& \max_j \| \what{a}_j - a_{0j} \|_1 \\
		&\ge& \| \what{a}_{S_{0j}} - a_{0,S_{0j}} \|_1
		\eea
		for any estimator $\what{A}_n$ and any $2\le j \le p$.
		Thus, it suffices to show that
		\bean\label{Lbound_s0reg}
		\inf_{\what{a}_{S_{0j}}} \sup_{\Omega_{0n} \in \calU_p} \bbE_0 \| \what{a}_{S_{0j}} - a_{0,S_{0j}} \|_1
		&\ge& c \cdot \frac{s_{0j}}{\sqrt{n}}
		\eean
		for some constant $c>0$ and any $2\le j \le p$. 
		Since it is the minimax lower bound for the standard linear regression having $s_0$-dimensional coefficient, the inequality \eqref{Lbound_s0reg} holds by a slight modification of Example 13.12 in \cite{duchi2016lecture}.
		Similarly, it is easy to check 
		\bea
		\inf_{\what{a}_{S_{0j}}} \sup_{\Omega_{0n} \in \calU_p} \bbE_0 \|\what{A}_n - A_{0n} \|_F^2
		&=&  \inf_{\what{a}_{S_{0j}}} \sup_{\Omega_{0n} \in \calU_p} \sum_{j=2}^p \bbE_0 \| \what{a}_{S_{0j}} - a_{0,S_{0j}} \|_2^2 \\
		&\ge& c \frac{\sum_{j=2}^p s_{0j} }{n}
		\eea
		for some constant $c>0$, by \cite{duchi2016lecture}.
	\end{proof}

	\begin{proof}[Proof of Theorem 3.5]	
		Note that
		\bea
		\inf_{\what{A}_n} \sup_{\Omega_{0n} \in \calU_p^0} \bbE_0 \|\what{A}_n - A_{0n} \|_\infty &=&
		\inf_{\what{A}_n} \sup_{\Omega_{0n} \in \calU_p^0} \bbE_0 \left[ \max_j \| \what{a}_j - a_{0j} \|_1 \right] \\
		&\ge& \inf_{\what{A}_n} \sup_{\Omega_{0n} \in \calU_p^0} \max_j \bbE_0   \| \what{a}_j - a_{0j} \|_1 
		\eea
		and
		\bea
		\inf_{\what{A}_n} \sup_{\Omega_{0n} \in \calU_p^0} \max_j \bbE_0   \| \what{a}_j - a_{0j} \|_1
		&\ge& \sup_{\Omega_{0n}\in \calU_p^0} \max_j c\cdot s_{0j}  \left(\frac{ \log (j/s_{0j})}{n} \right)^{1/2} \\
		&=& c\cdot s_0 \left(\frac{  \log (p/s_{0})}{n} \right)^{1/2}
		\eea
		for some constant $c>0$ by \cite{ye2010rate}. 
		
	\end{proof}

	\section{Simulation under Misspecified Models}
	
	Although it  slightly departs from the main topic of the paper, it might be worth comparing results for different choices of $\alpha$ when the model is misspecified.
	To investigate misspecified DAG models, we generated the data sets from the $p$-dimensional multivariate Laplace distribution with zero mean and covariance matrix $\sg_{0n}$.
	The covariance matrix $\sg_{0n}$ was generated based on the MCD of $\Omega_{0n} = \sg_{0n}^{-1}$ as before, where only 3$\%$ of entries of the Cholesky factor $A_{0n}$ were drawn from a uniform distribution on $[-0.7, -0.3] \cup [0.3, 0.7]$. 
	The entries of the diagonal matrix $D_{0n}$ were sampled from a uniform distribution on $[2, 5]$. 
	We generated the data sets under two settings: $(n=100,p=300)$ and $(n=200,p=500)$.

	\begin{table}
		\centering
		\caption{
			Results for ESC prior with different choices of $\alpha$ are shawn.
			Sp: sparsity;
			FDR: false discovery rate; TPR: true positive rate; $\bar{p}_0$: the mean inclusion probability for zero entries in $A_{0n}$; $\bar{p}_1$: the mean inclusion probability for nonzero entries in $A_{0n}$. 
		}\vspace{.15cm}
		\begin{tabular}{cc ccccc}
			\hline
			$(n,p, \text{Sp})$ & $\alpha$ & \# of errors & FDR & TPR & $\bar{p}_0$ & $\bar{p}_1$   \\ \hline
			\multicolumn{1}{c}{\multirow{5}{*}{(100, 300, 3\%)}}  & 0.999 & 778 & 0.3234 & 0.8074 & 0.0252 & 0.8038  \\
			& 0.8 & 547 & 0.1843 & 0.7665 & 0.0157 & 0.7620  \\
			& 0.6 & 568 & 0.0773 & 0.6305 & 0.0100 & 0.6287  \\
			& 0.4 & 899 & 0.0275 & 0.3413 & 0.0073 & 0.3625  \\
			& 0.2 & 1272 & 0.0133 & 0.0550 & 0.0065  & 0.0903  \\ \hline
			\multicolumn{1}{c}{\multirow{5}{*}{(200, 500, 3\%)}} & 0.999 & 1053 & 0.2020 & 0.9621 & 0.0161 & 0.9552  \\
			& 0.8 & 591 & 0.0980 & 0.9447 & 0.0105 & 0.9362  \\
			& 0.6 & 466 & 0.0370 & 0.9105 & 0.0066 & 0.8980  \\
			& 0.4 & 955 & 0.0146 & 0.7560 & 0.0047 & 0.7430  \\
			& 0.2 & 2856 & 0.0100 & 0.2392 & 0.0042  & 0.2559  \\ \hline
		\end{tabular}\label{table:sim_mis}
	\end{table}
	
	The simulation results are summarized in Table \ref{table:sim_mis}.
	Based on the results, when the model is misspecified, the choice $\alpha \approx 1$ might not be good because it tends to have high FDR value.
	Instead, a slightly smaller choice of $\alpha$ would give reasonable performance.
	It seems that $\alpha=0.8$ gives reasonable results in terms of the number of errors (and others).
	
	In our settings, as the power $\alpha$ decreases, one can see that FDR, TPR, $\bar{p}_0$ and $\bar{p}_1$ also decrease.
	If we take a close look at the posterior samples, selected variables with smaller $\alpha$ tend to be a subset of those with larger $\alpha$, which supports our observations.
	We are not sure whether this trend is always true or not, but here is a rough intuition: smaller value of $\alpha$ weakens the effect of $(\what{d}_{S_j})^{-(\alpha n+\nu_0)/2}$ in $\pi_\alpha(S_j\mid \bfX_n)$, which pushes $\pi_\alpha(S_j\mid \bfX_n)$ to select larger $S_j$, while the main penalty term $\pi_j(S_j)$ is not changed, so $\pi_\alpha(S_j\mid \bfX_n)$ tends to prefer smaller subset $S_j$ as $\alpha$ decreases.
	
	In summary, in our problem, the different choice of $0<\alpha<1$ can improve the variable selection performance compared to $\alpha \approx 1$.
	The choice $\alpha = 0.8$ gave reasonable results in our simulation, but there is no theoretical guideline to choose $\alpha$, which might be an interesting topic for the future research.
	
%\bibliographystyle{imsart-nameyear}
%\bibliography{sparse-cholcov-supp}

\end{document}